\documentclass[11pt]{amsart}
\usepackage{geometry}                % See geometry.pdf to learn the layout options. There are lots.
\usepackage{subfigure,graphicx}
\graphicspath{{figures/}}
\usepackage{amssymb,amsmath,mathtools}
\usepackage{epstopdf}
\usepackage{bm}
\usepackage{xcolor}

\usepackage{hyperref}
\hypersetup{
    colorlinks=true, %set true if you want colored links
    linktoc=all,     %set to all if you want both sections and subsections linked
    linkcolor=blue,  %choose some color if you want links to stand out
    citecolor=blue,
}

\usepackage{setspace}
\newcommand{\pl}{\parallel}
\newcommand{\eps}{\varepsilon}
\newcommand{\curl}{\mbox{curl}}
\newcommand{\cA}{\mathcal{A}}
\newcommand{\cE}{\mathcal{E}}
\renewcommand{\div}{\mbox{div}}
\newtheorem{thm}{Theorem}
\newtheorem{defn}{Definition}
\newtheorem{prop}{Proposition}
\newcommand{\R}{\mathbb{R}}
\newcommand{\Z}{\mathbb{Z}}
\renewcommand{\S}{\mathbb{S}}
\newcommand{\la}{\langle}
\newcommand{\ra}{\rangle}
\newcommand{\sign}{\textrm{sign}}
\newcommand{\cB}{\mathcal{B}}
\newcommand{\cM}{\mathcal{M}}

\title[Isodrastic magnetic fields]{Isodrastic Magnetic fields for suppressing transitions in guiding-centre motion}
\author{J.W.Burby$^{1}$, R.S.MacKay$^{2,\dagger}$, S.Naik$^{2,\ddagger}$}
\address{$^{1}$Los Alamos National Laboratory, Los Alamos, New Mexico 87545, USA, $^{2}$Mathematics Institute, University of Warwick, Coventry CV4 7AL, U.K.}
\email{$^{1}$jburby@lanl.gov, $^{\dagger}$R.S.MacKay@warwick.ac.uk,
$^{\ddagger}$shibabratnaik@gmail.com}
\date{\today}   % Activate to display a given date or no date

\begin{document}
\begin{abstract}
%\rsm{Question whether title needs to be more informative, e.g.~Transitions in guiding-centre motion and isodrastic magnetic fields?}
In a magnetic field, transitions between classes of guiding-centre motion can lead to cross-field diffusion and escape.
We say a magnetic field is {\em isodrastic} if guiding centres make no transitions between classes of motion.  
This is an important ideal for enhancing confinement.
First, we present a weak formulation, based on the longitudinal adiabatic invariant, generalising omnigenity.
To demonstrate that isodrasticity is strictly more general than omnigenity, we construct weakly isodrastic mirror fields that are not omnigenous.
Then we present a strong  formulation that is exact for guiding-centre motion.  We develop a first-order treatment of the strong version via a Melnikov function and show that it recovers the weak version.
The theory provides quantification of deviations from isodrasticity that can be used as objective functions in optimal design.  The theory is illustrated with some simple examples.  
%Some open questions remain, in particular, whether one can construct isodrastic fields without axisymmetry, and how to reduce transitions while simultaneously confining suitable sets of trajectories that do not make any transitions. 
\end{abstract}
\maketitle

\onehalfspacing
\tableofcontents
\singlespacing

\section{Introduction}
\label{sec:intro}
%\subsection{}
On a short timescale, charged particles (mass $m$, charge $e$) in a strong magnetic field $B$ perform helices around magnetic field lines with gyrofrequency $\tfrac{e}{m}|B|$ and gyroradius 
\begin{equation}
\rho = \tfrac{mv_\perp}{e|B|},
\label{eq:gyror}
\end{equation}
$v_\perp$ being the magnitude of the component of the velocity perpendicular to $B$.  We consider fields for which $|B|\ne 0$ in the region of interest, indeed large enough to make the gyroradius smaller than typical length-scales for variation of $B$.  

On longer time-scales, the centre-line, radius and pitch angle of the helices drift, but there is an adiabatic invariant, the magnetic moment, whose asymptotic expansion starts 
\begin{equation}
\mu = \tfrac{mv_\perp^2}{2|B|},
\end{equation}
and thereby makes $\rho \propto |B|^{-1/2}$ along trajectories.  The relevant small parameter $\eps$ is the relative change in $B$ (in magnitude and direction) seen by the particle during one gyro-period. 
The adiabatic invariant allows one to reduce the dynamics to rapid gyro-oscillation about a ``guiding centre'' whose motion is governed by a relatively slow Hamiltonian system of two degrees of freedom (DoF).  

To zeroth order in $\eps$ the motion of the guiding centre is along magnetic field lines, governed by the canonical Hamiltonian dynamics of 
\begin{equation}
H(s,p_\pl) = \tfrac{1}{2m} p_\pl^2 + \mu |B(s)|
\label{eq:ZGCM}
\end{equation}
for arc-length $s$ and momentum $p_\pl=mv_\pl$ along $B$ (the effect of an electrostatic field can be included but we leave it out for now; so can gravitational fields and relativistic effects, see Appendix~\ref{app:esrel}).  We will write $$b = B/|B|,$$ and $'$ for derivative with respect to arc length along a fieldline, so e.g.~$$|B|' = i_b d|B| = b \cdot \nabla |B|, $$in differential forms and vector-calculus notation, respectively.
We shall often express relations by differential forms, but we provide vector-calculus translations where feasible.  
For an example in the other direction, the condition $\div B = 0$ for a magnetic field can be written as $d\beta=0$, where $\beta$ is the magnetic flux 2-form $i_B\Omega$ with $\Omega$ being the volume-form in physical space; this says that $\beta$ is closed.
For a tutorial on differential forms for plasma physics, see \cite{M20}.

For time-independent fields, the zeroth-order guiding-centre motion (ZGCM), equation (\ref{eq:ZGCM}), conserves the energy $H=E$, so the trajectories on a given fieldline can be classified into (see Figure~\ref{fig:zgcm}):
\begin{figure}[htbp] %  figure placement: here, top, bottom, or page
   \centering
   \includegraphics[width=5.5in]{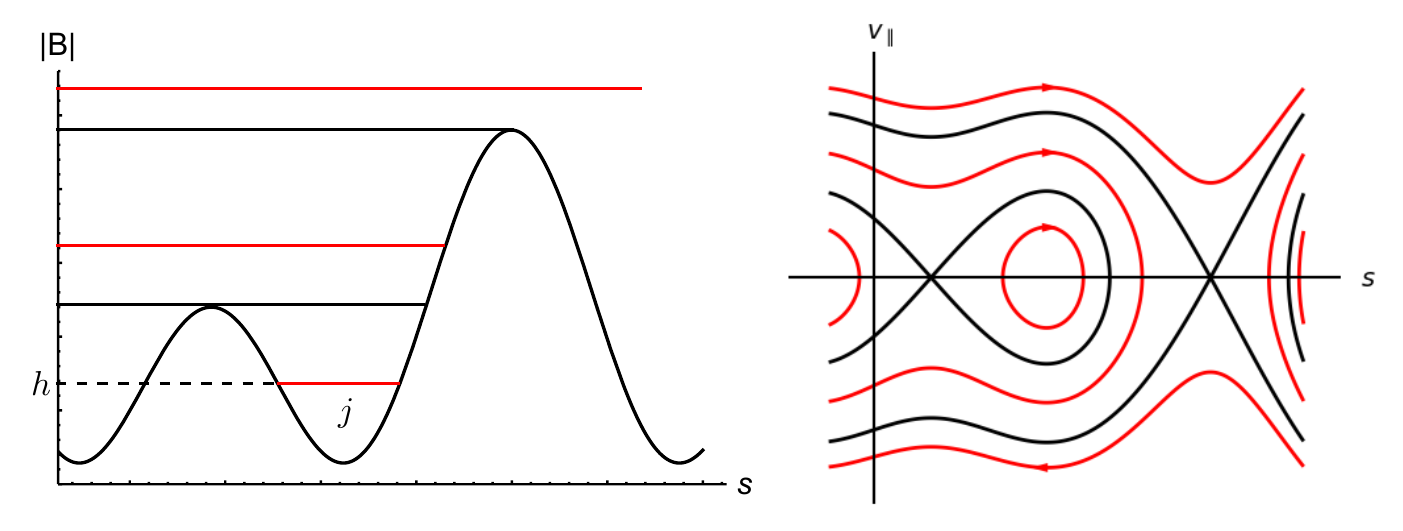}
   \caption{Example of field strength along a field line and resulting ZGCM. Marginal cases are shown in black.}
   \label{fig:zgcm}
\end{figure}
\begin{itemize}
\item passing:~$|B| < E/\mu$ along the whole fieldline so the guiding centre moves unidirectionally along it (this includes the case $\mu=0, E \ne 0$);
\item one-sided bouncing:~$|B| < E/\mu$ along the fieldline in one direction (say $s<0$) but exceeds $E/\mu$ somewhere in the other direction, with 
$|B|' >0$ at the first point $s_0$ where $|B|=E/\mu$, so a guiding centre moving from $s<0$ with $v_\pl>0$ reverses direction at $s_0$;
\item two-sided bouncing:~$|B| < E/\mu$ in an interval $(s_1,s_2)$ along the fieldline with $|B|=E/\mu$ and $|B|' \ne 0$ at both ends, so the guiding centre bounces periodically between $s_1$ and $s_2$ (the standard terminology for ``bouncing'' is ``trapped'', but it is only at zeroth order that they are trapped; the confinement problem is that at first order they might not be trapped; another term used in some contexts is ``blocked'', e.g.~\cite{F+}); 
\item marginal:~if $|B|' = 0$ at a point where $|B|=E/\mu$ then the guiding centre takes infinite time to reach it.
\end{itemize}

\begin{defn}
We call the portion of fieldline between a pair of turning points for bouncing motion a {\em segment}. It is sometimes useful to identify a segment with its corresponding interval $[s_1,s_2]$ of arclength values. 
\end{defn}

To first order in $\eps$, however, the guiding centre drifts across the field.  The description of first-order guiding-centre motion (FGCM) that we prefer (a reformulation of \cite{L}) is that the motion is Hamiltonian on the space of guiding-centre position $X$ in 3D and parallel velocity $v_\pl$, with the Hamiltonian 
\begin{equation}
H = \tfrac12 mv_\pl^2 + \mu |B(X)|
\label{eq:H}
\end{equation} 
and closed 2-form
\begin{equation}
\omega = e\beta + m d(v_\pl b^\flat),
\label{eq:omega}
\end{equation}
where $\beta$ is the magnetic flux 2-form %$i_B\Omega$ 
%($\Omega$ being the volume-form on physical space) 
and $b^\flat$ is the 1-form $b\cdot dX$. 
%Note that $\beta$ being closed ($d\beta = 0$) is equivalent to $\div\ B = 0$.
This description generates the dynamics $(\dot{X},\dot{v}_\pl)=V$ by solving $i_V\omega = -dH$ for $V$ (except where $\tilde{B}_\pl=0$, defined below, at which $\omega$ is degenerate; see \cite{BuEl} for a modification that has no singularity).  The solution can be written as
\begin{align}
\dot{X} &= \left({v_\pl} \widetilde{B} + \frac{\mu}{e} b \times \nabla|B|\right)/\widetilde{B}_\pl \label{eq:Xdrift}\\
\dot{v}_\pl &= -\frac{\mu}{m} \frac{\widetilde{B}}{\widetilde{B}_\pl} \cdot \nabla |B|, \label{eq:vdrift1}
\end{align}
where 
\begin{equation}
\widetilde{B} = B+ \tfrac{m}{e}v_\pl c,\ c = \curl\ b \mbox{ and } \widetilde{B}_\pl = \widetilde{B}\cdot b .
\label{eq:tildeB}
\end{equation}
The above use of differential forms is convenient because it produces a Hamiltonian form for guiding-centre motion despite there being no natural canonical coordinates for it; the phase space has three position-dimensions and one velocity-dimension.    

Alternative expressions for FGCM appear in the literature (e.g.~\cite{LC,H}).  One alternative that exhibits the standard curvature and grad-B drifts is
\begin{align}
\dot{X} &= v_\pl b + \tfrac{mv_\pl^2}{e|B|}c_\perp + \tfrac{\mu}{e|B|}b\times\nabla |B| \label{eq:approxdrift}\\
\dot{v}_\pl &= -\tfrac{\mu}{m}(b+\tfrac{mv_\pl}{e|B|}c_\perp)\cdot\nabla |B|. \nonumber
\end{align}
Note that $c_\perp$ can be written as $b\times \kappa$, where the curvature vector $\kappa= b\cdot \nabla b$.  Equation (\ref{eq:approxdrift}) can be seen as the first-order expansion of (\ref{eq:Xdrift},\ref{eq:vdrift1},\ref{eq:tildeB}) in $\eps$ (think of $\eps$ as $1/e$).
An equivalent way to write  it is to define $v_\pl = \pm\sqrt{\frac{2}{m}(E-\mu|B|)}$ as a (multivalued) function of position $X$, and then 
$$\dot{X} = v_\pl b + \tfrac{mv_\pl}{e|B|}(\curl(v_\pl b))_\perp,$$
where $E$ is the value of $H$ treated as constant in computing the curl.
Although (\ref{eq:approxdrift}) conserves $H$, we doubt that it has a Hamiltonian formulation in general (but see \cite{Bo84} for some discussion), thus we prefer (\ref{eq:Xdrift},\ref{eq:vdrift1},\ref{eq:tildeB}) on the grounds of preserving as much structure as possible from the original model.  
%It does not even conserve $H$ unless $\kappa \cdot J = 0$, where $J=\curl\, B$ is the current, though that can be rectified by changing the evolution of $v_\pl$ to $\dot{v}_\pl = -\tfrac{\mu}{m}(b + \tfrac{mv_\pl}{e|B|}c_\perp) \cdot \nabla |B|$.

%\jwb{I do not see this unless the current vanishes. The ZGC terms conserve H by themselves, so its just the perp drifts that matter. The grad-B drift obviously preserves H so it boils down to the curvature drift. The derivative of $|B|$ in the direction of curvature drift does not appear to vanish. See below. I see no reason why field line curvature should be perp to curl B in general, do you? }, 
%\begin{align*}
%    &c_\perp\cdot  \frac{1}{2}\nabla |B|^2  =\frac{1}{2} b\cdot \kappa\times \nabla|B|^2 = b\cdot \kappa\times(B\cdot\nabla B + B\times(\nabla\times B))\\
%    &=b\cdot \kappa\times( B\times(\nabla\times B))=b\cdot (  [\kappa\cdot \nabla\times B]B  ) = |B| \kappa\cdot \nabla\times B 
%\end{align*}

Whichever of the above equations are used, the motion still conserves $H$ but produces drift across the field. {By assigning to each point in guiding-center phase-space its corresponding ZGCM orbit, this cross-field drift may be visualized as motion through the space of ZGCM trajectories. Since the profile of $|B|$ with respect to arc-length along nearby fieldlines is in general different, a drifting solution's instantaneous ZGCM trajectory may transition between classes.}  
%The profile of $|B|$ with respect to arc-length along nearby fieldlines is in general different, leading to transitions between classes.  
Repeated transitions between classes  produces an effective diffusion (e.g.~\cite{GT,Men}) and hence poor confinement.
In some classes (e.g.~ripple-trapped) the drifts produce large, even unbounded, excursions, so transitions into such classes also produce poor confinement.

For passing motion, one may be able to design the field so that many of the fieldlines remain in some desired region $R$, e.g.~a solid torus, and furthermore, so that the guiding centres for many initial conditions remain on fieldlines in $R'$, the subset of $R$ whose gyro-orbits are contained in $R$.  For the guiding centres, it suffices to have invariant 2-tori {with spatial projections in $R$} for all values of energy and magnetic moment used, because such tori confine all initial conditions inside them.  Near-integrability {(e.g.~from an approximate flux function for $B$)}, $C^3$-smoothness and generic conditions on the field guarantee existence of such tori, by KAM theory (there are various references, e.g.~\cite{Mo,SZ} and the semi-popular \cite{Du}).  Thus, ignoring the effects of interaction between particles, particles on passing trajectories inside such a torus remain within $R$.

One-sided bouncing can be treated together with passing.  For example, for a mirror machine, guiding centres that enter one end, make one bounce inside and then leave by the same end, have much the same effect as passing ones that go out of the other end.  For fields with an invariant set of finite volume, one-sided bouncing trajectories are rare, because Poincar\'e recurrence implies that almost every fieldline comes back arbitrarily close to any value of $|B|$ it takes.  So henceforth, ``bouncing'' will refer to two-sided bouncing.

For (two-sided) bouncing motion there is a second adiabatic invariant $L$, called ``longitudinal'', whose asymptotic expansion starts with 
\begin{equation}
L = \int_{s_1}^{s_2} mv_\pl \, ds,
\label{eq:L}
\end{equation} 
as long as the fieldline seen by the guiding centre changes relatively little during one bounce period $T = 2 \int_{s_1}^{s_2} {v_\pl}^{-1} ds$, e.g.~\cite{H}.  Then the motion in a bouncing class can be reduced to one DoF for the intersection of the fieldline segment with a transverse section $\Sigma$; the Hamiltonian is still $H$ but evaluated at the value of $v_\pl$ for which the bouncing segment has the given value of $L$; the symplectic form is $e\beta$ (by conservation of magnetic flux along fieldlines, this gives the same dynamics regardless of which transverse section is used).  Closed level-sets of $H$ on $\Sigma$ (for given $L$) give invariant 2-tori for the guiding-centre motion.  So the corresponding particles are confined to within one gyro-radius of the projections of such tori to physical space (ignoring the effects of interaction between particles).  If the particles are on tori that fit in a desired region $R$ then they stay in that region.  Examples are the ``banana'' trajectories in a tokamak.  They bounce above and below the outer midplane (where $|B|$ is minimum along fieldlines), moving alternately inwards and outwards between bounces relative to a flux surface (thus tracing out a banana-shape in projection to a poloidal section), but not in general closing because of a slight rotation around the central axis so that the banana drifts around the central axis, tracing out a torus of banana cross-section.

If the drift motion leads towards a marginal case, however, the guiding centre may make transitions between the above classes or between different bouncing classes.  Such transitions can lead to large changes in the region explored by a particle \cite{N}. % [XXX Could also cite Tennyson et al, but they do slowly varying 1DoF]  
Transitions between bouncing classes may lead to larger tori that no longer fit in the desired region, or even unbounded motion.  Transitions from bouncing to passing may lead to motion that is not confined by the tori for passing trajectories.  Repeated transitions can lead to large accumulated changes.  Transitions where one bouncing class gets divided into two lead to pseudo-random choice of new class.  These are all particular problems for high-energy particles, e.g.~\cite{B+, F+} and \cite{P+}.

Thus, transitions between classes of guiding-centre trajectory are generally bad for confinement.

One way to solve this problem was proposed in 1975 by \cite{HM}, called ``omnigenity''.  For a review, see \cite{H}.  It assumes the field $B$ has a flux function $\psi$, i.e.~such that $i_B d\psi = 0$ ($B\cdot \nabla \psi = 0$) and $d\psi \ne 0$ ($\nabla \psi \ne 0$) almost everywhere. 
For a review of magnetic fields with a flux function, including some less known properties, see Appendix~\ref{app:fluxfn}.

The field is called {\em omnigenous} if the time-average of $i_V d\psi$ ($V\cdot \nabla \psi$) along each zeroth-order trajectory is zero, using the non-Hamiltonian $V$ of (\ref{eq:approxdrift}).   This is automatic for passing trajectories on irrational flux surfaces; see \cite{H} or Appendix~\ref{app:omnigen}. %\rsm{perhaps move the justification to App C. JB: agree} [because ZGCM at given energy for circulating particles in given direction is uniquely ergodic (i.e.~it has a unique ergodic probability measure), so the time-average for each trajectory is its flux-surface average with respect to the invariant area-form $\tfrac{|B|}{v_\pl} \cA$, where $\cA$ is an  invariant area-form for fieldline flow defined in Appendix~\ref{app:fluxfn}; the time-average can be computed to be $0$ \cite{H} (see Appendix~\ref{app:omnigen} for an equivalent proof).]
%, though the time for which one has to average to obtain a given accuracy might be long for winding ratios that do not have small Diophantine constants.  
%In contrast to the claim in \cite{H}, it fails in general for rational flux surfaces (see Appendix~\ref{app:omnigen}).  For them and for the bouncing classes, zero time-average of $i_Vd\psi$ is a non-trivial requirement.  
A necessary and sufficient condition for omnigenity is that $L$ be constant for all bouncing trajectories with given energy and class on a flux surface.
Necessity was proved in \cite{LC} (extended to full generality in \cite{PCHL}), and sufficiency (which seems not to have been proved before) is proved in Appendix~\ref{app:omnigen}. 
%[but one might ask why one does not also require a condition for circulating trajectories on each rational flux surface; the continuity argument of \cite{H} shows only that the flux-surface average of $i_V d\psi$ is zero on each rational surface but not that its average is zero on each closed trajectory. {\color{red}make less terse. How about: the continuity argument of \cite{H} shows only that the flux-surface average of $i_V d\psi$ is zero on each rational surface, but flux-surface averages need not agree with time averages on individual fieldlines on rational surfaces. }  It turns out that $\la i_Vd\psi\ra$ is automatically zero for passing particles on a rational flux surface if it is zero for all bouncing particles on the same surface (see Appendix~\ref{app:omnigen}).]
%its extension $\int e A^\flat + p_\pl b^\flat$ ($\int (e A + p_\pl b)\cdot dx$) is constant, where $A$ is a vector potential for $B$, i.e.~$i_B\Omega = dA^\flat$ ($B= \curl\ A$) 
%(see \cite{LC} for bouncing trajectories and Appendix~\ref{app:omnigen} for rational circulating trajectories). 

Axisymmetric fields are omnigenous (indeed, so are all quasi-symmetric fields).  Constructions of non-axisymmetric omnigenous fields appear in \cite{CS,LC,PCHL} but depend on existence of Boozer coordinates, which are derived assuming the field is non-degenerate MHS (e.g.~\cite{H}) and we are not aware that any non-axisymmetric non-degenerate MHS fields are known.  In particular, we are not aware that $|B|$ as a given function of Boozer coordinates can be realised by a magnetic field $B$ in 3D.

For $C^3$ omnigenous fields, some generic conditions imply existence of many invariant tori for the passing trajectories, close in projection to the flux surfaces, by KAM theory, at least for low energies, as already mentioned.  It furthermore implies that the invariant tori for bouncing motion are close in projection to parts of flux surfaces.  Lastly it implies that transitions between classes are a second-order effect \cite{CS}.

Thus, omnigenity (assuming it is realisable) sounds a good solution for confinement.  Nonetheless, not every field has a flux function, e.g.~most vacuum fields do not, nor do most magnetohydrodynamic equilibria with anisotropic pressure or mean flow.  Secondly, reducing transitions to second order is perhaps not enough for good confinement.  Thirdly, even for fields with a flux function, the requirement that $\la v_d\cdot \nabla \psi\ra=0$ for all GC trajectories is stronger than necessary for confinement; prevention of transitions and existence of KAM tori in each class would suffice.
Fourthly, analytic exactly omnigenous fields have to be quasi-symmetric \cite{CS} and it is suspected that non-axisymmetric quasi-symmetric fields do not exist \cite{GB}; this last one is a minor objection, however, because analyticity is not necessary for real applications, where the current distribution need not be analytic (except in the vacuum case where analyticity is automatic because locally the field is the gradient of a function satisfying Laplace's equation).  
%But the first three objections are more substantial.
A weaker condition than omnigenity, called pseudo-symmetry, was introduced by Mikhailov \cite{M+} but does not prevent transitions; the relation to our work is discussed in Appendix~\ref{app:ps}.

In this paper we first extend \cite{CS} to present a set of conditions for absence, to first order, of transition between classes of guiding-centre motion. The set does not require a flux function, and so has wider applicability than omnigenity.  Even if there is a flux function, our conditions are weaker than omnigenity, so they stand more of a chance of being realisable without axisymmetry.  Indeed, we construct non-axisymmetric mirror fields with no transitions.  Then we present a ``strong'' form that prevents all transitions exactly.  We call it ``isodrasticity''. We are unaware of any previous work that formulates such a criterion for non-perturbative suppression of transitions.

Isodrasticity has the further advantage  that it can be adapted to high-energy particles, such as the $\alpha$-particles produced by $D-T$ fusion, for which a higher-order guiding-centre approximation may be required (see \cite{Bu} for explicit computation of higher-order guiding-centre approximations). 

Our theory provides clear objective functions to contribute to the optimisation of magnetic fields.
It provides enhanced understanding of the effects of imperfections in tokamaks, and more generally in quasi-symmetric fields. It suggests the prospect of controlling  transitions between classes by weak breaking of isodrasticity via trim coils.

% We present a weak and a strong version.  The weak one is based on the second adiabatic invariant, so reduces transitions just to second order, like omnigenity but without requiring a flux function nor its average drift to be zero everywhere.  It is a useful starting point, both pedagogically and practically.  The strong version is an exact treatment of guiding-centre motion, preventing all transitions exactly.

We emphasise that isodrasticity prevents transitions between classes but that one would still need to confine trajectories that stay within each class, as mentioned  above.  Failure to achieve that was the main problem with early stellarators.  The basic way we propose to achieve it for bouncing classes is by designing a suitable subset of level curves of $L$ to be closed and fit in the machine and finding a corresponding band of KAM tori for the guiding-centre motion at each value of energy and magnetic moment, which confine all trajectories inside, though this might  be challenging for ripple-trapped classes.  For circulating classes, we propose to achieve confinement by KAM tori derived from an approximate flux function.
%Such KAM tori are expected from the longitudinal adiabatic invariant, provided the level sets of the latter close and fit in the machine.  
This important aspect of confinement is not addressed further here.

\section{Reduced guiding-centre motion and weak isodrasticity}
\label{sec:wiso}
To explain our concept of weak isodrasticity, we first give a more detailed description of the reduction of first-order guiding-centre motion by the longitudinal invariant and of the set of critical points of field strength along the fieldlines.  We end the section by computing the flux of reduced trajectories that make a transition if the field is not weak isodrastic.

\subsection{Reduced guiding-centre motion and critical points of $|B|$ along the field}

We suppose that in the domain of interest, $B$ is nowhere zero and is $C^r$ with $r\ge 2$ (and for some purposes more). 

Using energy conservation (\ref{eq:H}), for $\mu > 0$ the second adiabatic invariant (\ref{eq:L}) can be written as $L = \sqrt{m\mu}\, j$ with 
\begin{equation}
j = \int_{s_1}^{s_2} \sqrt{2(h-|B|)}\ ds,
\label{eq:j}
\end{equation} 
where $h = E/\mu$ and the bounce points are at arclengths $s_1, s_2$.
%(Warning:~later in this paper, $h$ will denote the value of $|B|$ on $\Sigma$, to be defined in (\ref{eq:Sigma})). 
We note in passing that this is the Abel transform $\int_{-\infty}^h \sqrt{2(h-v)}\ d\ell(v)$ (using the convention with a factor $\sqrt{2}$) of the length $\ell(v)$ of the subsegments with $|B|<v$, which might provide a useful alternative way to compute it and is used in Section~\ref{sec:construct} and Appendix~\ref{app:omnigen}.

Bouncing segments along a given fieldline come in one-parameter families, parametrised by the value $h$ of the scaled field strength at its ends.
The action $j$ is differentiable with respect to $h$ and a standard calculation shows that the derivative is
\begin{equation}
\frac{dj}{dh} = \int_{s_1}^{s_2} \frac{ds}{\sqrt{2(h-|B|)}}.   
\label{eq:djdh}
\end{equation}
This can be recognised as the time taken from $s_1$ to $s_2$ by the dynamics $\frac{ds}{dt} = \pm \sqrt{2(h-|B|)}$ of the scaled Hamiltonian $H=\tfrac12 u^2 + |B(s)|$.  Thus the period of bouncing in the original time is $T = 2\sqrt{\frac{m}{\mu}}\frac{dj}{dh}$.
As a bouncing segment approaches marginality, the period goes to infinity.  For example,  approaching a non-degenerate local maximum at one end with $|B|=h_0$ and $|B|'' = -a$, the period diverges asymptotically like $T \sim \sqrt\frac{m}{\mu a} \log \frac{2a}{|h-h_0|}$.

A key role is played in the reduced dynamics by
the set $\Sigma$ of critical points of $|B|$ along fieldlines, i.e.
\begin{equation}
\Sigma = \{x \in \R^3 : |B(x)|'=0\},
\label{eq:Sigma}
\end{equation}
where, as before, $'$ denotes derivative with respect to arclength along a fieldline.
Subdivide $\Sigma$ into the disjoint union
$$\Sigma = \Sigma^+ \cup \Sigma^0 \cup \Sigma^-,$$ according as $|B|'' >0, =0$ or $<0$, respectively (nondegenerate local minima, degenerate critical points, nondegenerate local maxima).  By the implicit function theorem, $\Sigma^\pm$ are $C^{r-1}$ surfaces (with possibly several components).
For $r\ge 3$, $\Sigma^0$ is generically a $C^{r-2}$ curve (with possibly several components) and generically forms the common boundary of $\Sigma^\pm$ (see Appendix~\ref{app:generic}).
Some examples are shown in Figure~\ref{fig:exSigma}.
\begin{figure}[htbp] %  figure placement: here, top, bottom, or page
\centering
\subfigure[]{\includegraphics[width=1.2in]{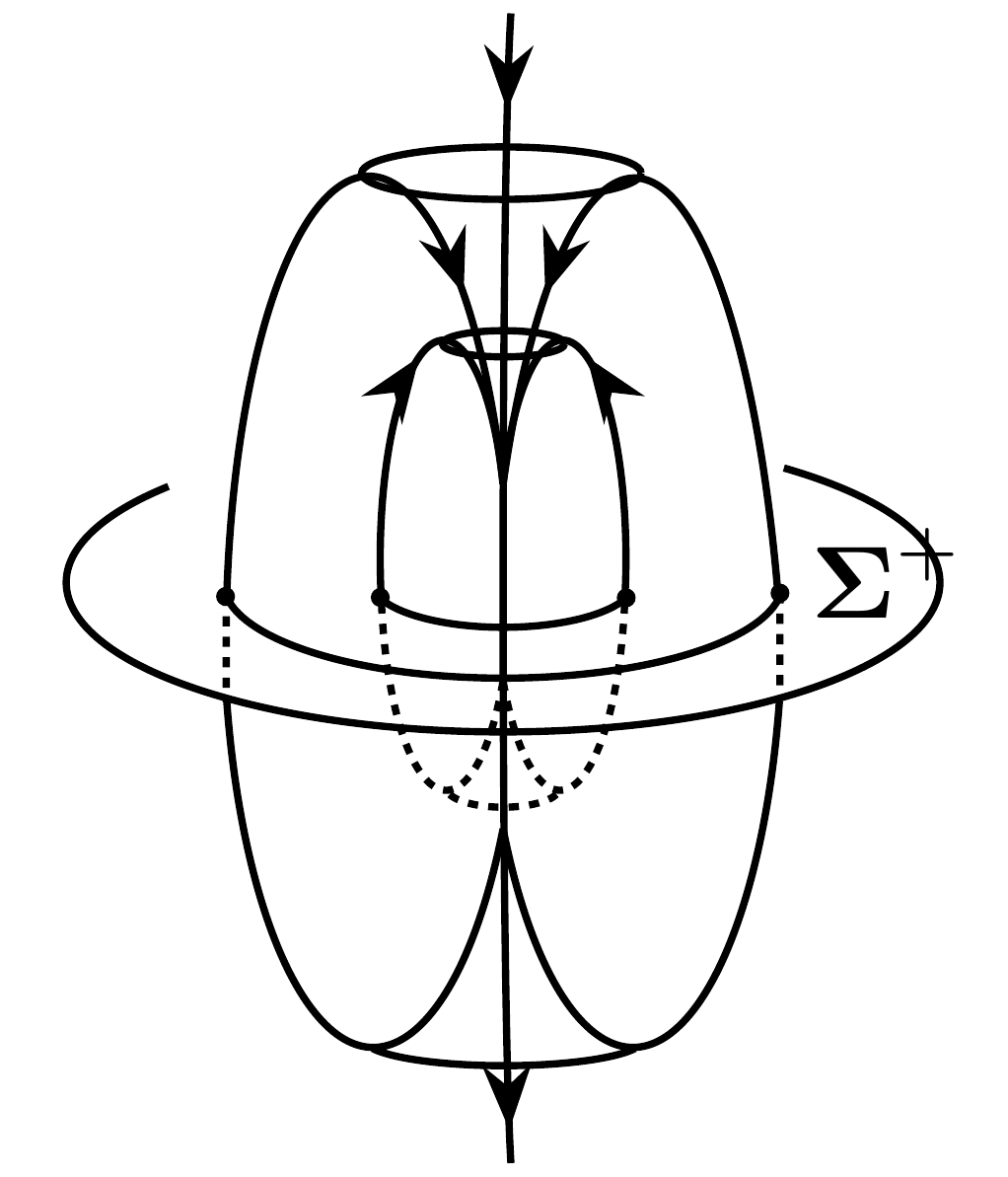}}
\subfigure[]{\includegraphics[width=2.0in]{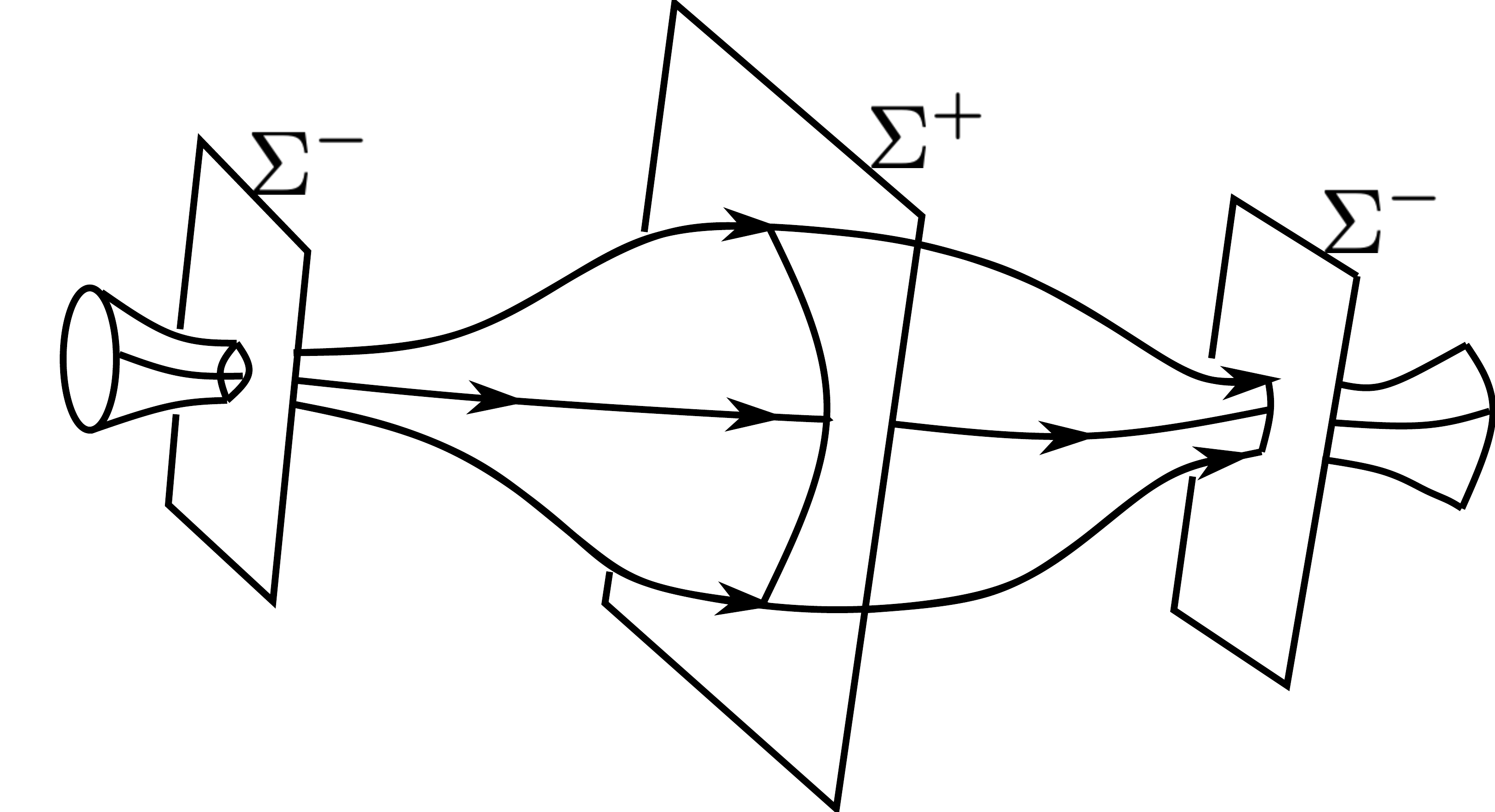}}
\subfigure[]{\includegraphics[width=2.0in]{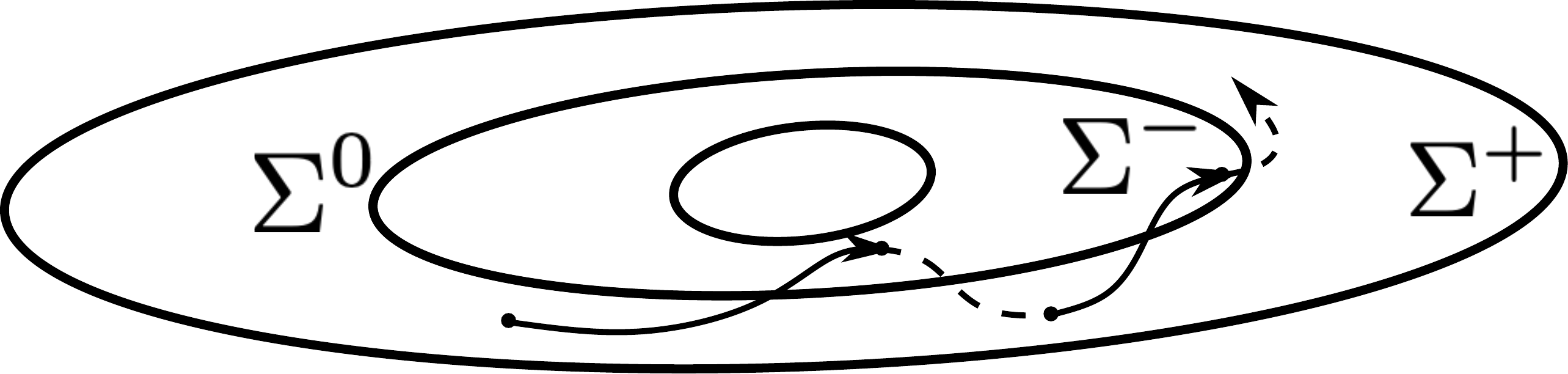}}
\caption{The set $\Sigma$ of critical points of $|B|$ along fieldlines for (a) a dipole, (b) a mirror machine, and (c) a tokamak.}
\label{fig:exSigma}
\end{figure}

The analysis presented in this Article demonstrates that the identification of $\Sigma$ in a magnetic field, including its decomposition into $\Sigma^\pm$ and $\Sigma^0$, is a key step to understanding the confinement properties of the field.

There are various ways to label a segment of fieldline.  All include specifying the fieldline $\gamma$ and the (common) value $h$ of $|B|$ at its endpoints. %{\color{red}Is there a risk of confusing the $h$ here with $h$ defined elsewhere as $|B|$ restricted to $\Sigma$?}  
To complete the label, we have to say which wells of $|B|$ it visits along the fieldline.  One way to do this is to specify the set of all local minima it visits, or indeed the set of all its intersections with $\Sigma$.  Given $\gamma$ and $h$ this is more than is required to specify the segment, because the segment is connected, so one could instead select one of the local minima.  The only problem with any of these specifications is that as the segment moves to nearby fieldlines the set of local minima may change size, because one might be annihilated by a local maximum or a new one created from a horizontal inflection;  a jump will be required if the selected local minimum is annihilated.  The set also changes size when there is a real change in the class, i.e.~when one end becomes a local maximum so the segment can lengthen over some new wells, or when an interior local maximum rises to height $h$, which splits the segment into two, but we consider that natural.
An alternative is to specify the first intersection with $\Sigma^-$ outside the segment when going in one direction along the fieldline (say the positive one).  This choice does not suffer from the previous issue, but the local maximum might be annihilated with a local minimum a bit further away and then the first local maximum would jump to one further away, or a horizontal inflection might be born between the local maximum and the end of the segment, leading to a jump in the other direction.
Thus there is no good solution.

The best way is to use equivalence classes of such labels.  We denote such an equivalence class by $\widetilde{M}$. As mentioned above, $\widetilde{M} = (M,h)$ comprises $M$, a field line together with visited wells along that line, and $h$, the common value of $|B|$ at the segment's endpoints. {Note that the set of possible $\widetilde{M}$ is in one-to-one correspondence with (images of) ZGCM bouncing orbits.}

There is likewise not a good concept of class of segment.  One might be tempted to say two segments are in the same class if one can be obtained from the other by continuous change of segment, but the example of Figure~\ref{fig:exclass} shows that this can lead to two different segments of the same fieldline with the same $h$ being in the same class, which is not what we want (it can even be done conserving $j$).  On the other hand, discontinuous change in a segment is well defined, being a jump in segment at given $h$ as the fieldline varies locally continuously, and we refer to this somewhat loosely as a transition between bouncing classes.
\begin{figure}[htbp] %  figure placement: here, top, bottom, or page
   \centering
     \includegraphics[width=2in]{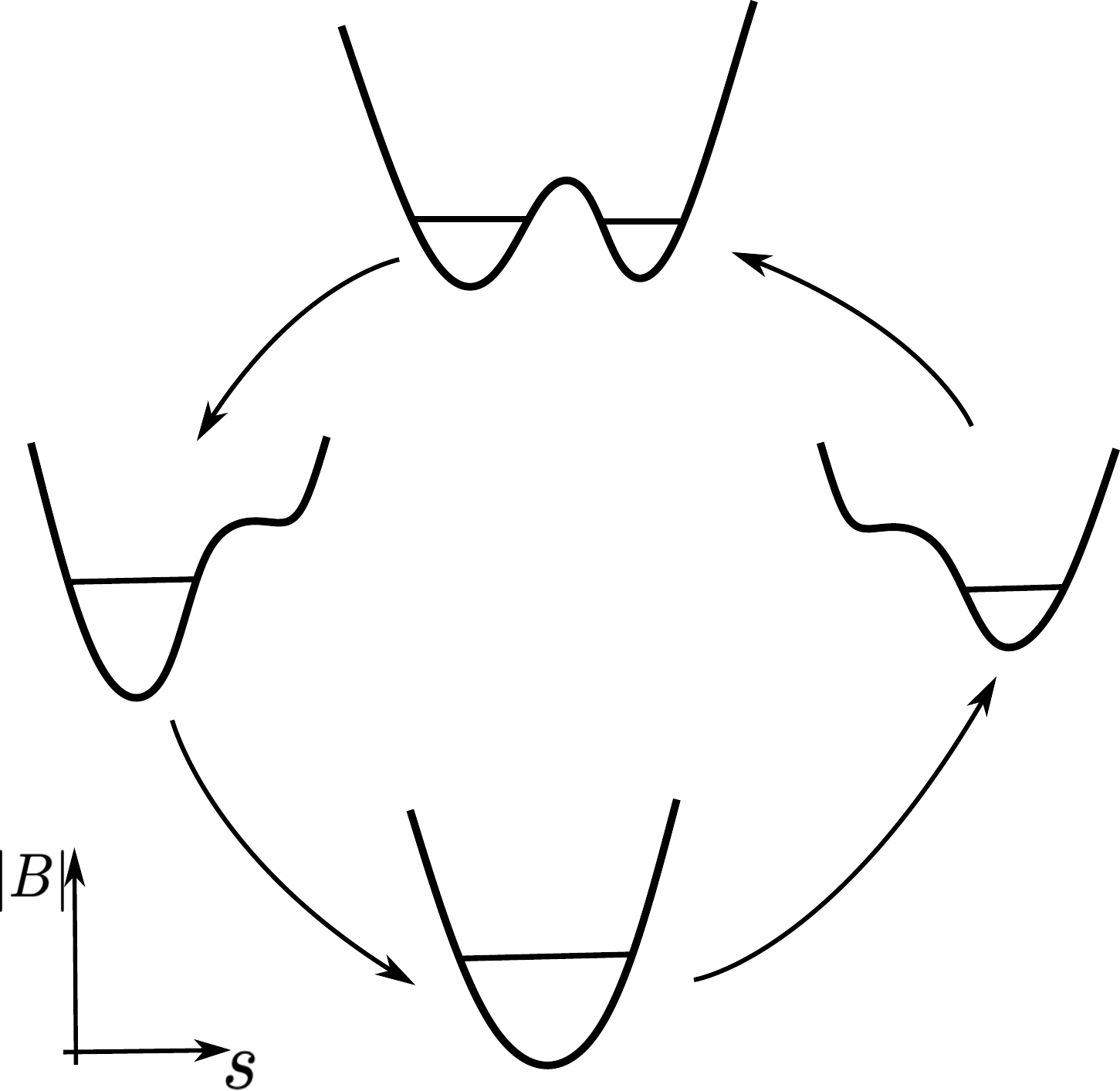}%{OneClass} 
   \caption{Example to show that two different segments of one fieldline and shared energy can be in the same class ($|B|$ plotted against arc length $s$ along fieldlines).}
   \label{fig:exclass}
\end{figure}

Given a fieldline and the set $M$ of intersections of a segment $\gamma$ with $\Sigma$, there is an interval $[h_{\min},h_{\max}]$ of compatible values of $h$.  $h_{\min}$ is the maximum of $|B|$ over $M$ and $h_{\max}$ is the minimum of $|B|$ over the first local maxima outside $M$ in the two directions. The action $j$ of (\ref{eq:j}) is a continuous increasing function of $h$ for this interval of segments (and differentiable in the interior).  It goes from some $j_{\min}$ to a $j_{\max}$.  If $M$ is a single local minimum then $j_{\min}=0$.  

For given $j$ in this interval ($j$ is now a scalar rather than a function), we obtain the reduced and scaled Hamiltonian $H_j(M)$ defined to be the unique value such that 
\begin{equation}
\int_{\gamma(M,H_j(M))} \sqrt{2(H_j(M)-|B|)}\ ds = j,
\label{eq:Hj}
\end{equation}
where $\gamma(M,h)$ is the segment of the fieldline with label $M$ and endpoints with $|B|=h$.  $H_j(M)$ is the energy divided by $\mu$, equivalently, $|B|$ at the bounce points, expressed as a function of $M$ for given $j$.  See Figure~\ref{fig:hjclip}.
\begin{figure}[htbp] %  figure placement: here, top, bottom, or page
   \centering
   \includegraphics[width=4in]{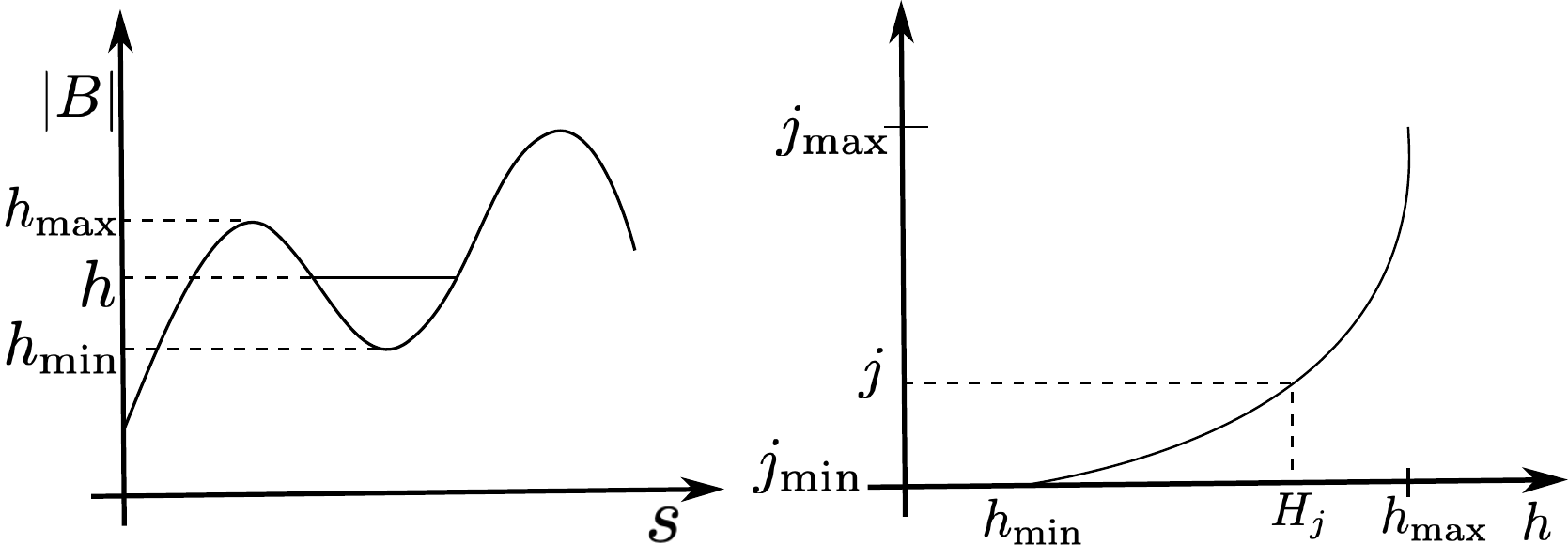} 
   \caption{Illustration of the definitions of $h_{\min}, h_{\max}$ and $H_j$ for $j \in [j_{\min},j_{\max}]$. Here $M$ consists of a single point (the local minimum of $|B|$), implying $h_{\text{min}}$ is the value of $|B|$ there.  $h_{\text{max}}$ is the value of $|B|$ at the hill top to the left of the minimum.  }
   \label{fig:hjclip}
\end{figure}

The phase space $F_j$ for the reduced system is 
$$F_j = \{M: j_{\min}(M) \le j \le j_{\max}(M) \}.$$  
It is easiest to think of the case of a single well, then $F_j$ is just a subset of $\Sigma^+$, but in general it can be mapped to any transverse section of the field.  The symplectic form on $F_j$ is $e\beta$, which is invariant along the fieldlines between any choices of transverse section.  

%\rsm{think about how we want to establish differentiability of $H_j$;} 
To address the reduced dynamics on $F_j$ we first remark that on $\mathring{F}_j = \{M: j_{\min}(M) < j < j_{\max}(M) \}$ (the subset of $F_j$ for which the segments are non-marginal), $H_j$ is differentiable.  Indeed, 
\begin{prop}\label{prop:dh}
%OLD: Let $X$ be any vector field on physical space taking fieldlines to fieldlines with $i_X d|B|=0$ at the ends of a periodically bouncing segment $\gamma$.  Then 
{Let $X$ be any vector field on guiding-centre phase-space whose flow commutes with that of ZGC dynamics and that satisfies $dj(X) = 0$. Then}
\begin{equation}
i_XdH_j=\frac{1}{T_s} \int_\gamma \bigg(\frac{X_\perp \cdot \nabla |B|}{\sqrt{2(H_j-|B|)}}b^\flat - \sqrt{2(H_j-|B|)}\, i_X i_c \Omega\bigg),
  \label{eq:dH_j}
\end{equation}
where $c = \curl\, b$ and 
$$T_s = \int_\gamma \frac{1}{\sqrt{2(H_{j}-|B|)}} b^\flat$$
is related to the true bounce period $T$ by $T_s = \sqrt{\mu/m}\,T/2$.
\end{prop}

\begin{proof}
Let $Y$ be any vector field on guiding center phase space whose flow commutes with that of ZGC dynamics (we do not yet require that $Y$ preserve level sets of $j$). The derivative of $j$ along $Y$ is given by
\begin{align*}
    i_Ydj &=  \frac{dj}{dh}\,L_Yh + \int_\gamma L_Y(\sqrt{2(h-|B|)}\, b^\flat  \\
    &=  \frac{dj}{dh}\,L_Yh +\int_\gamma i_Yd(\sqrt{2(h-|B|)}b^\flat) + d(i_Y \sqrt{2(h-|B|)}b^\flat).
\end{align*}
% First consider how $j$ changes as the fieldline is changed for fixed value of $h$.
% $$i_Xdj = \int_\gamma L_X(\sqrt{2(h-|B|)}\, b^\flat = \int_\gamma i_Xd(\sqrt{2(h-|B|)}b^\flat) + d(i_X \sqrt{2(h-|B|)}b^\flat).$$
The second term of the integrand integrates to $0$ because the square root is zero at the ends of $\gamma$.
The first term expands to
$$\int_\gamma i_Y \left( \frac{-d|B|\wedge b^\flat}{\sqrt{2(h-|B|)}} + \sqrt{2(h-|B|)}\, db^\flat\right).$$
Now $db^\flat = i_c\Omega$ and $b$ is tangent to $\gamma$, so we end up with 
$$i_Ydj = \frac{dj}{dh}\,L_Yh + \int_\gamma \frac{-Y_\perp \cdot \nabla |B|}{\sqrt{2(h-|B|)}} b^\flat + \sqrt{2(h-|B|)}\, i_Yi_c\Omega.$$
Next we use (\ref{eq:djdh}) which gives the change in $j$ for a change in $h$ without change of fieldline (one way to derive that equation is a similar but simpler use of Lie derivative). As discussed in the text, it implies $dj/dh = (T/2)\sqrt{\mu/m}$.
%It implies that to achieve no change in $j$ for a simultaneous change in fieldline and $h$ we require
%$i_Xdj + T i_X dH_j = 0$.  
On taking $Y$ to be a field $X$ that in addition satisfies $dj(X)=0$, this leads to (\ref{eq:dH_j}).
\end{proof}

We refer interested readers to Appendix \ref{app:dh} for an alternative proof of this result that casts it in a more general light.

Applied to $\gamma$, $b^\flat/\sqrt{2(h-|B|)}$ is the time interval $dt$ for ZGCM.  So, modulo the correction term involving $i_c\Omega$, the formula (\ref{eq:dH_j}) says that $dH_{j}$ is the time-average of $d|B|$ along the segment of fieldline. An important consequence for our analysis is that taking the limit as a bouncing segment goes to a homoclinic one, we see that $i_XdH_j$ goes to the value of $i_{X_\perp}d|B|$ at the critical point, because the bouncing segment spends all but a vanishing fraction of its time near there.  At a critical point, $|B|'=0$ so $i_{X_\perp} d|B| = i_X d|B|$ there.  Thus, using $\Sigma^-$ as transverse section and defining $h=|B|$ on $\Sigma^-$ (being the value of $h$ for homoclinics to $\Sigma^-$), 
\begin{equation}
  dH_j \to dh  
  \label{eq:limit}
\end{equation}
as any homoclinic case is approached.
%\rsm{perhaps we need to add an argument that $H_j$ is indeed differentiable there, with this derivative? And it seems to me $dH_j$ and hence the dynamics are undefined at heteroclinic cases.}

The reduced dynamics on $F_j$ in a scaled time $\tau = \frac{\mu}{e}t$ is $dM/d\tau = u$ defined by 
\begin{equation}
 i_u\beta = -dH_j
 \label{eq:reduceddyn}
\end{equation} 
($B\times u = -\nabla H_j$ in $F_j$), at points where $H_j$ is differentiable. At least in $\mathring{F}_j$, $dH_j$ is $C^1$ so the vector field $u$ induces a flow. This dynamics conserves $H_j$ and so trajectories move along level curves of $H_j=h$.  %\rsm{we need to be sure the flow extends to the boundary because we want to use uniqueness of trajectories!}

For closed level curves of $H_j$, let $\Phi$ be the magnetic flux enclosed, then the precession period (in real time) is 
\begin{equation}
T=\frac{e}{\mu} \frac{\partial \Phi}{\partial h}
\label{eq:precession}
\end{equation}
with $j$ fixed.  This comes from the standard formula that the period $T$ of a periodic orbit $\gamma$ of an autonomous Hamiltonian system is $\frac{\partial S}{
\partial E}$ where $S$ is its action ($\int_\gamma \alpha$ for a primitive $\alpha$ of the symplectic form $\omega$) and $E$ its energy (using that periodic orbits come in 1-parameter families).

The case of short bounces can be treated explicitly (we mean short in length, not in time; these are usually called ``deeply trapped'' but again they might not be trapped).
The limiting case of zero length is $j=0$.  For these, $M$ is a singleton.  The motion is on $\Sigma^+$ with Hamiltonian $H_0=|B|$.  So the trajectories follow level curves of $|B|$ on $\Sigma^+$ at rate $\frac{\mu}{e}u$ in real time, with $i_u\beta = -d|B|$.  To keep these guiding centres in a region $R$ one must put them on level sets of $|B|$ on $\Sigma^+$ lying within $R$.  
For fields with flux surfaces, an ideal is that for each flux surface the minima of $|B|$ along fieldlines on it have the same value of $|B|$, because then the short bouncers remain on that flux surface (pointed out already by \cite{MCB}).
In contrast, it goes wrong for ripple-trapped particles in tokamaks. They are a class of bouncing trajectories in a poloidally confined region (not around the equatorial plane) for which the $|B|$ contours in $\Sigma^+$ are approximately vertical, so for one sign of $e/m$ the guiding centres leave the desired solid torus.  Any particles that make the transition into this class are lost.  See Fig.5 of \cite{GT} for an example, and \cite{P+} for more.

One can also treat the linear approximation to short bounces with $j>0$.  It gives $j = {\pi(h-h_{\min})}/\sqrt{|B|''}$, where $|B|''$ is evaluated on $\Sigma^+$ (and is the bounce frequency in a scaled time).  So to this order in $j$, $$H_j(M) = |B(M)| + \sqrt{|B(M)|''}\,\tfrac{j}{\pi},$$
where $M \in \Sigma^+$.

Although our focus in this paper is on transitions between classes, the above behaviour of short bouncers is of independent interest. Yet, it will be seen in Section~\ref{sec:tok} to be of relevance to transitions for perturbations of a tokamak field.

\subsection{Weak isodrasticity}
The adiabatic invariant $j$ is well conserved if the segment changes relatively little during one bounce period.  The condition can be written as $\frac{2}{\ell^2}\sqrt{|B|} \frac{\partial j}{\partial h} \ll 1/\rho$, where $\ell$ is a lengthscale for variation of $B$ and $\rho$ is the gyroradius (\ref{eq:gyror}).  This fails when the period becomes large, in particular as a segment approaches a marginal case.  Figure~\ref{fig:marginal} shows the three principal ways marginal cases can be approached.
\begin{figure}[htbp] %  figure placement: here, top, bottom, or page
   \centering
\includegraphics[width=2in]{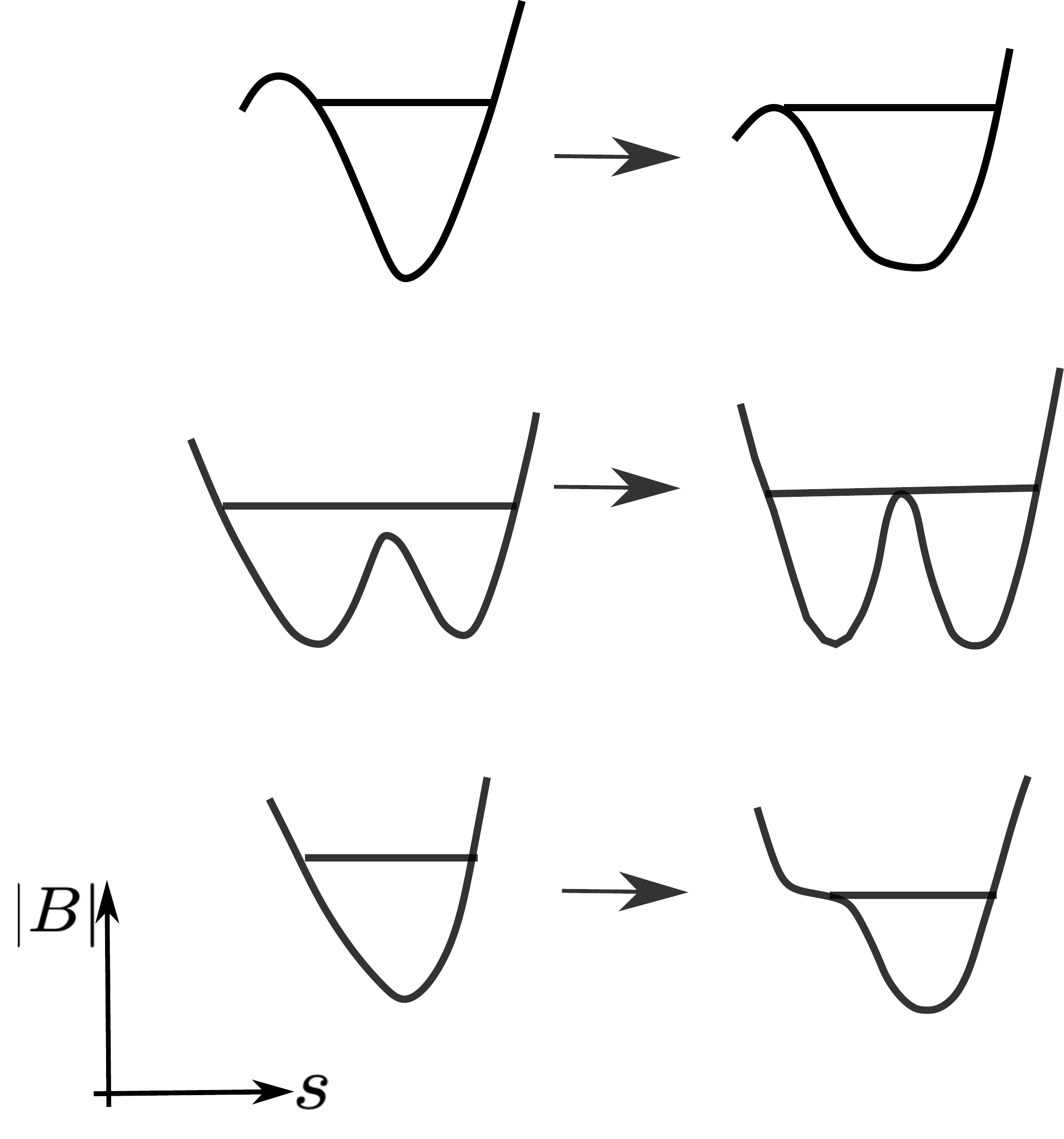} 
   \caption{Examples of approach to marginal cases for bouncing segments ($|B|$ against arc length $s$ along field lines).}
   \label{fig:marginal}
\end{figure}

In the weak version of isodrasticity, to determine whether marginal cases can be approached we make the approximation that $j$ continues to be conserved up to and including when a marginal case is reached and the dynamics is given by the reduced dynamics (\ref{eq:reduceddyn}) on $F_j$.

\begin{defn}
A magnetic field $B$ is {\em weakly isodrastic} if the marginal cases are never reached  from non-marginal ones by the first-order reduced dynamics. 
%{\color{red}Alt: A magnetic field $B$ is {\em weakly isodrastic} if, for all but a measure-zero set of initial conditions, marginal cases are never reached from non-marginal ones by first-order reduced dynamics.}  
\end{defn}
%\noindent[XXX perhaps we would like to add ``nor approached arbitrarily closely''? but then I can't make a theorem]

We will develop some necessary and some sufficient conditions for weak isodrasticity.  
On
$$\Sigma^{-0} = \Sigma^- \cup \Sigma^0,$$ 
define $h = |B|$ (this agrees with $E/\mu$ for guiding centres with $v_\pl=0$ on $\Sigma^{-0}$), and for  direction $\sigma \in \{\pm\}$ along the field define $\text{\text{\j}}^\sigma$ to be the value of $j$ for a segment $\gamma^\sigma$ starting at the given point of $\Sigma^{-0}$ and going in direction $\sigma$ (assuming a turning point is reached; if not, $\text{\text{\j}}^\sigma$ is undefined). When it is clear which direction is under consideration, we drop $\sigma$.  So
$$\j = \int_\gamma \sqrt{2(h-|B|)}\, ds.$$
Equivalently, given $x\in \Sigma^{-0}$ the value of $\text{\text{\j}}(x)$ is  the area of the zeroth-order guiding-centre separatrix lobe attached to $(X,v_\parallel) = (x,0)$ in the appropriate direction along the field. 
Note the relation
$$H_{\j(x)}(x) = h(x)$$
for $x \in \Sigma^-$, which follows from the definition (\ref{eq:Hj}) of $H_j$.

%\rsm{perhaps we can get $dH_j = dh$ on the boundary of $F_j$ from this, using $dh/dj = T^{-1} \to 0$?}

Loosely speaking, our results are that a magnetic field is weakly isodrastic iff the contours of $h$ and $\j$ coincide on $\Sigma^{-0}$ for both directions $\sigma$.  See Figure~\ref{fig:weakiso} for an illustration of failure of weak isodrasticity.  
\begin{figure}[htbp] %  figure placement: here, top, bottom, or page
   \centering
  \includegraphics[width=2in]{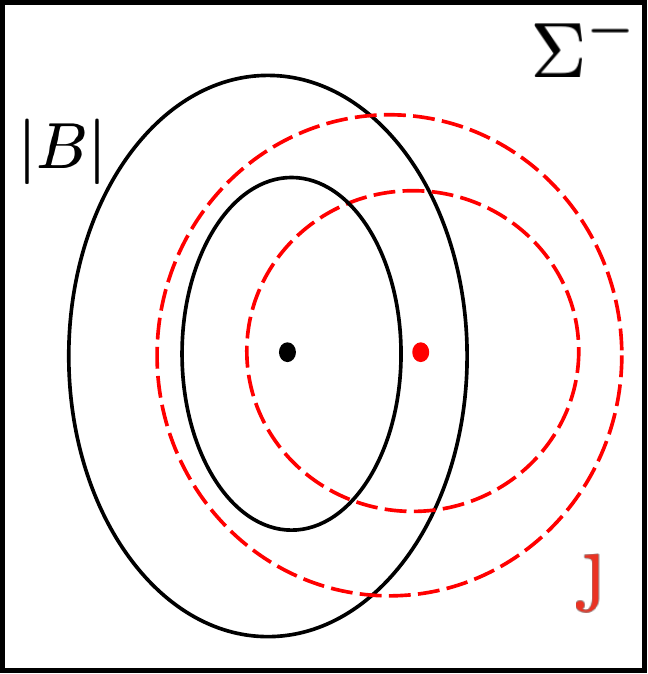} 
   \caption{Illustration of failure of weak isodrasticity; the level curves of $h$ and $\text{\text{\j}}$ on $\Sigma^-$ do not coincide.}
   \label{fig:weakiso}
\end{figure}
This explains the 
etymology of our definition.  In Greek, ``iso'' means equal and ``drasis'' means action.  We're asking for marginally bouncing trajectories of given energy (which up to scaling by $\mu$ is $h$ for particles on segments with an endpoint on $\Sigma^-$) to have the same action $\j$.  Use of the term ``isodrastic'' goes back to \cite{W}, but in a different context. 

We formulate precise statements of two necessary and one sufficient condition for weak isodrasticity in (Theorem~\ref{thm:1}).
First, the function $h$ is smooth on $\Sigma^-$ because $B$ is assumed to be smooth and $\Sigma^-$ was proved to be smooth.
For given direction along the field, the separatrix area $\j$ is also a smooth function on $\Sigma^-$ except at points with a heteroclinic connection, where the derivative generically becomes infinite and $\j$ jumps or ceases to be defined (if there is no longer any bounce point in that direction).  We define ${\Sigma^-}'$ to be this subset of $\Sigma^-$.
As already remarked, $H_j$ is differentiable for non-marginal segments and $dH_j \to dh$ as a homoclinic case is approached (\ref{eq:limit}).

% We will assume the generic situation that $\Sigma^0$ consists of smooth curves (see Appendix~\ref{app:generic}). \rsm{perhaps not necessary}

\begin{thm} 
\label{thm:1}
(a) If magnetic field $B$ is weakly isodrastic then for both directions along the field, $dh$ and $d\j$ are linearly dependent at every point of ${\Sigma^-}'$;\\
(b) If $B$ is weakly isodrastic and $\Sigma^0$ is a smooth curve without heteroclinic cases, then $h$ and $\j$ are constant on connected components of $\Sigma^0$;\\
(c) If for both directions, $\j$ is constant on components of level sets of $h$
then $B$ is weakly isodrastic.
\end{thm}
%\rsm{In thm mode, $d\text{\text{\j}}$ doesn't come out right}
%\rsm{I've kept some old text in the proofs so we can compare}

%The first condition is that $h$ and $\text{\text{\j}}$ have the same contours on $\Sigma^-$.  See Figure~\ref{fig:weakiso} for an illustration of failure of weak isodrasticity.  This explains the etymology of our definition.  In Greek, ``iso'' means equal and ``drasis'' means action.  We're asking for marginally bouncing trajectories of given energy (which up to scaling by $\mu$ is $h$ for particles on segments with an endpoint on $\Sigma^-$) to have the same action $\j$.  Use of the term ``isodrastic'' goes back to \cite{W}, but in a different context.  The second condition is that $\Sigma^0$ be simultaneous contours of $h$ and $\text{\text{\j}}$; thus it includes $\Sigma^0$ as a special contour on $\Sigma^{-0}$.

\begin{proof}
(a) If for some direction along the field, $d\j$ and $dh$ are linearly independent at some $x_0 \in {\Sigma^-}'$, let $j_0 = \j(x_0)$.  
%\rsm{perhaps I should write $x_0$ instead of $p$ so we can reuse the Figure}
In particular, $d\j \ne 0$ there so the level set $\j^{-1}(j_0)$ is locally a smooth curve. See Figure \ref{fig:pf1}.  Locally, it is the boundary of $F_{j_0}$ and the motion is periodic in the interior of $F_{j_0}$.  By independence of $dh$ and $d\j$ at $x_0$, $dh \ne 0$ there and the tangent to the boundary is not in $\ker dh$.  $dH_{j_0} \to dh$ as $x_0$ is approached from the interior.  So $\ker dH_{j_0}$ is at a non-zero angle to the boundary near $x_0$. The reduced dynamics is the Hamiltonian dynamics of $H_{j_0}$ with respect to the flux form $\beta$.  So trajectories of the reduced dynamics reach the boundary in finite time for one sign of scaled time, i.e.~in finite positive time for one sign of charge.  Thus $B$ is not weakly isodrastic.
\begin{figure}[htbp] %  figure placement: here, top, bottom, or page
 \centering
 \includegraphics[width=3in]{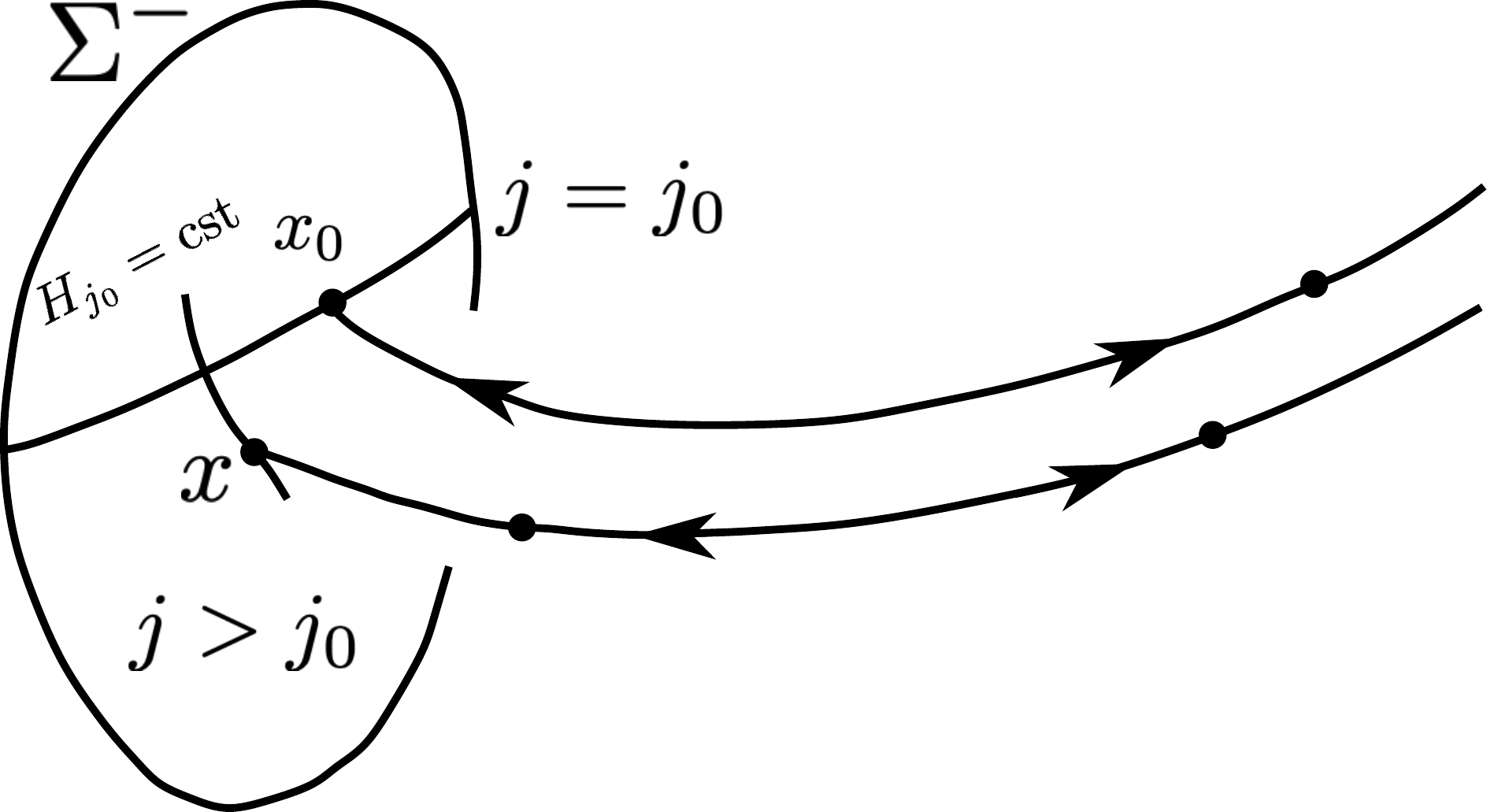} \caption{Illustration of part (a) of the proof of Theorem 1.}
 \label{fig:pf1}
\end{figure}

% OLD (a) Suppose the field is weakly isodrastic and that $dh$ and $d\text{\text{\j}}$ are independent at some point $x_0$ of $\Sigma^-$ for one direction along the field.  We will obtain a contradiction. 
% Take $j_0=\text{\text{\j}}(x_0)$.  Consider the function $H_{j_0}$ of (\ref{eq:Hj}) as being defined on the part of $\Sigma^-$ near $x_0$ where $\text{\text{\j}} \ge j_0$ by representing segments with $j=j_0$ on the given side of $\Sigma^-$ by the intersection of their fieldline with $\Sigma^-$ near $x_0$.  By hypothesis $d\text{\text{\j}} \ne 0$ at $x_0$ so $\text{\text{\text{\j}}}=j_0$ is locally a $C^1$ curve.
% See Figure~\ref{fig:pf1}.
% It represents the boundary of the bouncing class for this value of $j$.  On this boundary, $H_{j_0} = h$, because the segment has endpoint on $\Sigma^-$.  By independence of $dh$ and $d\text{\text{\j}}$ there, $\ker dH_j$ is transverse to the boundary and hence for one sign of $e$ there are trajectories of the reduced dynamics $\dot{x}=\frac{\mu}{e}u(x)$ defined by $i_u\beta = -dH_j$ that reach the boundary from non-marginal segments -- a contradiction. Therefore weakly isodrastic fields must have $dh$ and $d\text{\text{\j}}$ linearly-dependent on $\Sigma^-$.

(b) If $\Sigma^0$ is a smooth curve and $h$ is not constant along a component of $\Sigma^0$ then there exists a point $x_0 \in \Sigma^0$ where the derivative of $h$ along $\Sigma^0$ is non-zero.  Thus the tangent to $\Sigma^0$ there is not in $\ker dh$.  Let $j_0 = \j(x_0)$.  By assumption the segment from $x_0$ is not heteroclinic.  Then $dH_{j_0} \to dh$ as $x_0$ is approached from $\Sigma^-$.  Thus $\ker dH_{j_0}$ is at a non-zero angle to $\Sigma^0$ near $x_0$.  So trajectories of the reduced dynamics for this value $j_0$ reach $\Sigma^0$ in finite time for one sign of scaled time $\tau$, i.e.~in finite positive time for one sign of charge.  Thus $B$ is not weakly isodrastic.  We deduce that weak isodrasticity implies $h$ is constant along smooth components of $\Sigma^0$.
Since weak isodrasticity also implies $dh\wedge d\j=0$ up to the boundary in reduced space by (a), $h=\text{const.}$ along the boundary implies also that $\j=\text{const.}$ along the boundary.

%Also, weak isodrasticity implies that $\j$ is constant along generic $\Sigma^0$ for a neighbouring well.  This follows from linear dependence of $dh$ and $d\j$ on $\Sigma^-$ and $h$ being constant along $\Sigma^0$.

% OLD (b) Next, suppose $h$ or $\text{\j}$ is not constant on a component of $\Sigma^0$. Recall that we assume the generic situation that $\Sigma^0$ is a smooth curve bounding $\Sigma^-$ (Appendix~\ref{app:generic}).  If a bouncing segment ends on $\Sigma^0$ and its drift trajectory remains marginal on $\Sigma^0$ then both $h$ and $\text{\j}$ must remain constant along it.  So either the drift happens to be zero or the drift trajectory leaves $\Sigma^0$.  The former is a special case that we ignore.  In the latter case, by considering the opposite sign of charge, we deduce that there are drift trajectories that reach $\Sigma^0$ from outside $\Sigma^0$, so the field is not weakly isodrastic.

(c) Assume $\j$ is constant on level sets of $h$. Let $\gamma(t)$ denote a trajectory for first-order reduced dynamics such that $\gamma(0)$ is not marginal. If $\gamma$ were to become marginal in (say) forward time there would be a $t_0>0$ such that $\gamma(t_0)$ is marginal and $\gamma(t)$ is not marginal for $0\leq t < t_0$. But this is impossible because we claim:
\begin{itemize}
    \item[$\star$] $\gamma(t_0)$ marginal implies $\gamma(t)$ marginal for $t$ in an open neighborhood of $t_0$.
\end{itemize}
We prove $\star$ as follows. Suppose $x_0=\gamma(t_0)$ is marginal. The initial value problem for $x_0$ enjoys uniqueness because dynamics in the reduced space arise as a quotient of dynamics in the full guiding center phase space, in which uniqueness holds. Therefore if $x_0$ is a fixed point then we conclude $\star$. So assume $x_0$ is not a fixed point. Hamilton's equations imply that $dH_j(x_0)\neq 0$. Since $dh$ agrees with $dH_j$ at marginal points it follows that $dh(x_0)\neq 0$. By the implicit function theorem, the level set of $h$ containing $x_0$ is a $1$-manifold $\Gamma$ near $x_0$. And since $\j$ is constant along $\Gamma$ the reduced Hamiltonian $H_j$ is constant along $\Gamma$. The pullback of Hamilton's equations to $\Gamma$ now implies $\Gamma$ is locally invariant, which establishes the claim.  Thus the field is weakly isodrastic.
The same argument applies for the case of splitting of a segment by an interior local maximum.  Simply, $\j$ needs replacing by the sum $\j^+ + \j^-$.  If both $\j^\pm$ are constant on level sets of $h$ then so is this sum.
\end{proof}
One might ask why in (c) we did not need a hypothesis like $h$ constant on components of $\Sigma^0$.  We think this follows from $\j$ constant on $h$-levels, at least under some generic assumptions.  For example, in Appendix~\ref{app:generic} we show that $h$ constant on generic $\Sigma^0$ follows from linear dependence of $dh$ and $d\j$ for the short bouncing class on the neighbouring part of $\Sigma^-$.

%\rsm{added} 
%Also, weak isodrasticity implies that $\j$ is constant along generic $\Sigma^0$ for a neighbouring well.  This follows from linear dependence of $dh$ and $d\j$ on $\Sigma^-$ and $h$ being constant along $\Sigma^0$.

\subsection{Quantification of failure of weak isodrasticity and transition flux}
Theorem 1 leads to a quantification of failure to be isodrastic.  The failure of contours of $h$ and $\text{\j}$ to coincide on $\Sigma^-$ can be measured by the 2-form $dh \wedge d\text{\j}$.  This is most simply described by comparing it to the magnetic flux-form $\beta$, which is a nondegenerate top-form on $\Sigma^-$. Thus there is a function $\mathcal{M}$  on $\Sigma^-$ such that $$dh \wedge d\text{\j} = \mathcal{M}\beta.$$  
In Section~\ref{sec:Mel}, $\mathcal{M}$ will be identified as a ``Melnikov function'' for the FGCM dynamics.
But for now, to compute $\mathcal{M}$, if $\Sigma^-$ is given locally as the graph $z=Z(x,y)$ of a function in Cartesian coordinates then 
\begin{equation}
\mathcal{M} = \frac{h_{,x} \text{\j}_{,y} - h_{,y}\text{\j}_{,x}}{B_z-B_x Z_{,x}-B_y Z_{,y}},
\label{eq:Mfromhj}
\end{equation}
where subscripts after a comma indicate partial derivatives.

The function $\mathcal{M}$ has units of square root of field strength divided by length, but it is natural to multiply $\mathcal{M}$ by the factor $\sqrt{m\mu}$ to turn $\text{\j}$ into $L$. The quantity $\sqrt{m\mu}\, \mathcal{M}$ is an inverse time, so represents the rate of transition. {
%We can justify this dimensional analysis as follows.
Indeed, the flux of reduced orbits between classes is given precisely by the following Theorem~\ref{thm:2}.

We need first to introduce the Liouville volume-form $\Lambda$ on guiding-centre phase-space. It is defined by $\Lambda = \frac12 \omega \wedge \omega$, where $\omega$ is in (\ref{eq:omega}).  This can be computed to be
\begin{equation}
\Lambda = em \tilde{B}_\pl \Omega \wedge dv_\pl,
\label{eq:Lambda}
\end{equation}
using the relations $\beta \wedge b^\flat = |B|\, \Omega$ and $b^\flat \wedge db^\flat = -b\cdot c\ \Omega$.  In the following, we use the symbol $\rho$ for a density in GC phase space (as opposed to the gyroradius).

\begin{thm}
\label{thm:2}
Let $\Lambda$ denote the Liouville volume form in guiding center phase space. For a distribution $\varrho$ of guiding centers with density $\rho$ with respect to $\Lambda$ {in $(X,v_\pl)$}, the flux of reduced orbits between classes is given by the $2$-form on $\Sigma^-$
\begin{align*}
4\pi\,m^{1/2}\,\mu^{3/2}\,\overline{\rho}\,dh\wedge d\j,
\end{align*}
where $\overline{\rho}$ is the ZGC bounce-average of $\rho$.
%\rsm{where $\overline{\rho}$ is the average of $\rho$ along trajeectories of ZGCM?.}
\end{thm}

\begin{proof}
In the $3$-dimensional space of ZGCM bouncing orbits, particle flux is quantified by a $2$-form $\Gamma$. The flux of reduced orbits between classes is given by restricting $\Gamma$ to the $2$-manifold of marginal bouncers. So we will find a formula for $\Gamma$ and then analyze its restriction to the marginal orbits.

To determine an expression for $\Gamma$ we first argue generally in the context of symplectic Hamiltonian systems with $U(1)$ symmetry. This is relevant to bouncing particles because Kruskal's theory \cite{Kr} of nearly periodic systems implies guiding-centre dynamics formally comprises such a system when restricted to bouncing orbits. Let $V$ denote a Hamiltonian vector field with Hamiltonian $\mathcal{H}$ on the symplectic $2n$-manifold $(M,\omega)$. Assume there is a symplectic $U(1)$-action $\Phi_\theta:M\rightarrow M$ with momentum map $\mathcal{P}:M\rightarrow \mathbb{R}$ and that $\mathcal{H}$ is $U(1)$-invariant. We will also suppose that the quotient map $\pi:M\rightarrow M/U(1)$ sending phase points to their $U(1)$ orbits is a smooth mapping between smooth manifolds. 

Let $\Lambda= \tfrac{1}{n!}\omega^{\wedge n}$ denote the Liouville volume on $M$. Given a particle density $\varrho = \rho\,\Lambda$, regarded as a top-form on $M$, we will derive a formula for the flux of particles in the orbit space $M/U(1)$. At the level of measures, the density of particles in $M/U(1)$ is simply the measure-theoretic pushforward along $\pi$ of the phase space measure defined by $\varrho$. At the level of differential forms, this means the density on $M/U(1)$ is given by the fibre integral $\pi_*\varrho$ --- a volume form on $M/U(1)$. If $v$ denotes the vector field on $M/U(1)$ induced by the $U(1)$-invariant vector field $V$, the particle flux form on $M/U(1)$ is therefore $\Gamma = \iota_v\pi_*\varrho$. We can simplify this abstract formula as follows. Let $P,H$ denote the unique functions on $M/U(1)$ such that $\pi^*P = \mathcal{P}$ and $\pi^*H = \mathcal{H}$. A direct calculation shows that the fibre integral is given by
\begin{align*}
\pi_*\varrho = \frac{2\pi}{(n-1)!}\,\overline{\rho}\,dP\wedge \overline{\omega}^{\wedge (n-1)},
\end{align*}
where $\overline{\omega}$ is any $2$-form on $M/U(1)$ that restricts to the Marsden-Weinstein reduced symplectic form on level sets of $P$, and $\overline{\rho}$ is the $U(1)$-average of $\rho$. It follows that the flux form $\Gamma$ is given by
\begin{align*}
    \Gamma = \iota_v\,\pi_*\varrho = \frac{2\pi}{(n-2)!}\,\overline{\rho}\,dH\wedge dP\wedge\overline{\omega}^{\wedge (n-2)}.
\end{align*}

Specializing now to the guiding-centre case ($n=2$, $P = 2L$), we deduce that the flux form in the space of ZGCM orbits is given by $\Gamma = 4\pi\,\overline{\rho}\,dH\wedge dL$, where $H$ is the true (i.e.~unscaled) Hamiltonian. To restrict this $2$-form to the class boundary, we parameterize the set of marginal bouncers using the map $i$ that sends points in $\Sigma^-$ to the corresponding ZGCM separatrix orbit attached to $\Sigma^-$ in the appropriate direction along the field. Since we have the pullback identities
\begin{align*}
    i^*H = \mu\,h,\quad i^*L = \sqrt{m\,\mu}\, \j,
\end{align*}
the particle flux through the class boundary is given by
\begin{align*}
i^*\Gamma = 4\pi\,m^{1/2}\,\mu^{3/2}\,\overline{\rho}\,dh\wedge d\j,
\end{align*}
as desired.
\end{proof}

}

The flux-form of Theorem~\ref{thm:2} represents the transition fluxes out of and into a class as positive and negative contributions with respect to an orientation on $\Sigma^-$.  To obtain the flux in one direction, one has to integrate it over the subset with the appropriate sign.

The size of the function $\mathcal{M}$ can be quantified in various ways, for example, one could take the maximum of $|\mathcal{M}|$ over a relevant piece of $\Sigma^-$ (say between two contours of $h$), or the integral of $|\mathcal{M}|$ with respect to the flux form $\beta$ over a relevant piece of $\Sigma^-$. Note that the size of $\mathcal{M}$ can be used as an objective function for a magnetic field optimizer that encourages the optimized field to be weakly isodrastic.

The function
$\mathcal{M}$ can be computed from (\ref{eq:Mfromhj}) by numerical differentiation of $h$, $\text{\j}$ and $Z$, as will be illustrated in the next section, but a more direct method will be given in Section~\ref{sec:Mel}.

As a special case, a segment can approach marginality simultaneously at two (or more) points.  It leads to the reduced phase space $F_j$ having corners as well as edges.  This is discussed in Appendix~\ref{app:double}.  In particular, $\Sigma^-$ can contain curves along which there is heteroclinic connection:~the condition is just that there is a segment with both ends on local maxima, so basically a Maxwell equal-height condition.

%It is (almost) equivalent to say that the set of marginal cases is invariant under the reduced dynamics.  The marginal cases are given by an endpoint $x$ of a segment being in $\Sigma^{-0} = \Sigma^- \cup \Sigma^0$.  They automatically have $H_j(x) = |B|(x)$.  Suppose $x \in \Sigma^-$ and use the intersection of fieldlines with a neighbourhood of $x$ in $\Sigma^-$ as label.
%If $\nabla H_j(x) \ne 0$ then the intersection $x(t)$ of the fieldline with $\Sigma^-$ will move along the level curve of $H_j$.  Thus marginality is preserved iff $|B|$ is constant along this level curve.  But $j$ is preserved in this adiabatic approximation so if we define $\text{\j}$ on $\Sigma^-$ to be the value of $j$ for a segment with an endpoint on $\Sigma^-$ and going in the given direction, then we see that invariance of marginality implies $\text{\j}$ is constant on components of curves of $|B|$ constant with $\nabla |B| \ne 0$.  Conversely, $\text{\j}$ constant on such curves of $|B|$ constant implies marginality is invariant.  Points $x$ of $\Sigma^-$ where $\nabla H_j(x)=0$ are equilibria, so are automatically invariant.

%To complete the analysis of invariance of marginality, we have to address the case of a point $x \in \Sigma^0$.  $\Sigma^0$ being generically the boundary of $\Sigma^-$, invariance of marginality requires that the motion on $\Sigma^{-0}$ have its boundary invariant.  This requires $|B|$ to be constant along its components, and conversely. [improve]

\section{Examples of deviations from weak isodrasticity}
To help understand the constructions of the surfaces $\Sigma$ in a magnetic field and the reduced Hamiltonian $H_j$ for bouncing trajectories with given (scaled) value $j$ of longitudinal invariant, we treat  examples of the three types of field presented in Figure~\ref{fig:exSigma}.  For each we start from an axisymmetric field and then consider the effects of breaking axisymmetry.  The case of a dipole field has only $\Sigma^+$, so is trivially isodrastic; it is treated in Appendix~\ref{app:dipole}.  In the cases with $\Sigma^{-0}$, one in general loses weak isodrasticity.  The big question is whether there are special perturbations that keep weak isodrasticity.  For mirror fields, that will be answered positively in Section~\ref{sec:construct}.

See Appendix~\ref{app:sqrts} for some practical tips for computing expressions involving $b$ and $|B|$ and for computing $\text{\j}$.

\subsection{Mirror machine}
\label{sec:mir}
Here we consider a mirror machine.  As a simple axisymmetric and vacuum version, we take the field of two circular coils centred on a common axis that we orient vertically, with magnetic dipole moments in the same direction along the axis and with separation significantly larger than their radii so that the field is significantly weaker  between the coils than at their centres (in contrast to the Helmholtz case).  We focus attention on the region within less than the coil radii of the axis (outside the coils there are additional parts of $\Sigma^+$ and further away there can be nulls at which $\Sigma^+$ branches; the full picture will be presented in a future publication).  Then $\Sigma^-$ consists of two surfaces, one spanning each coil, and $\Sigma^+$ is an intermediate surface cutting the axis (recall Figure~\ref{fig:exSigma}(b)).  There is no $\Sigma^0$.  

Segments can bounce between the stronger fields near the coils, or if they have enough energy compared to magnetic moment they can escape through one or other coil.  
The bouncing segments can be labelled by the radius $r$ of their intersection with $\Sigma^+$.  For a segment $\gamma$ labelled by $r$ and with $|B|=h$ at the bounce points, $$j = J(r,h) = \int_\gamma \sqrt{2(h-|B|)}\, ds.$$  The function $J$ increases with $h$ and it is plausible that it increases with $r$ too.
Given $j \ge 0$, the Hamiltonian $H_j$ is defined on the part of $\Sigma^+$ for which there is a bouncing segment with that $j$, by $J(r,H_j(r))=j$.  The domain on $\Sigma^+$ is a disc for $j\le j^*$ and an annulus for $j>j^*$, where $j^*$ is the value for the segment bouncing along the axis between the saddle points of $|B|$ at the centres of the coils.
Under the assumption that $J$ increases with $r$, the derivative $dH_j/dr = - \frac{\partial J/\partial r}{\partial J /\partial h}< 0$, except 0 for $r=0$.  Its level sets are axisymmetric.  So the bouncing segments precess round the axis.

If axisymmetry is broken by a smooth perturbation but not too strongly then the surfaces $\Sigma^\pm$ deform smoothly and the function $H_j$ on $\Sigma^+$ deforms smoothly.  It remains a Morse function (i.e.~all its critical points are non-degenerate), so (under the assumption that $\frac{\partial J}{\partial r} > 0$) the low level sets remain closed curves around a central point and the bouncing segments precess around them, but level sets for higher values of $H_j$ may reach the boundary of definition. 

To consider transitions, it is simplest to break up-down symmetry by making, say, the lower coil produce a stronger field than the upper coil. Then there is a range of segments from the upper part of $\Sigma^-$ that bounce before the lower part of $\Sigma^-$. Let $\text{\j}$ be the function on the upper part of $\Sigma^-$ giving $j$ for the segment that starts at the given point on $\Sigma^-$ and goes into the mirror machine.
Let $j_1$ be the minimum of $\text{\j}$ on the upper part of $\Sigma^-$.  In the axisymmetric case this is $j^*$.
For $j<j_1$, the segments precess forever.  But for $j>j_1$, motion along a level set of $H_j$ could take the segment to the boundary of its domain of definition, where it becomes marginal.
Then one has to examine $h$ ($ = |B|$) and $\text{\j}$ on the upper piece of $\Sigma^-$.  In general, their level sets do not coincide, which corresponds to motion of segments leading to marginal cases.

The result of breaking axisymmetry for a mirror machine is illustrated in Figure~\ref{fig:nonaximm}(a,b).
\begin{figure}[htbp] %  figure placement: here, top, bottom, or page
   \centering
   \subfigure[]{\includegraphics[width=1.9in]{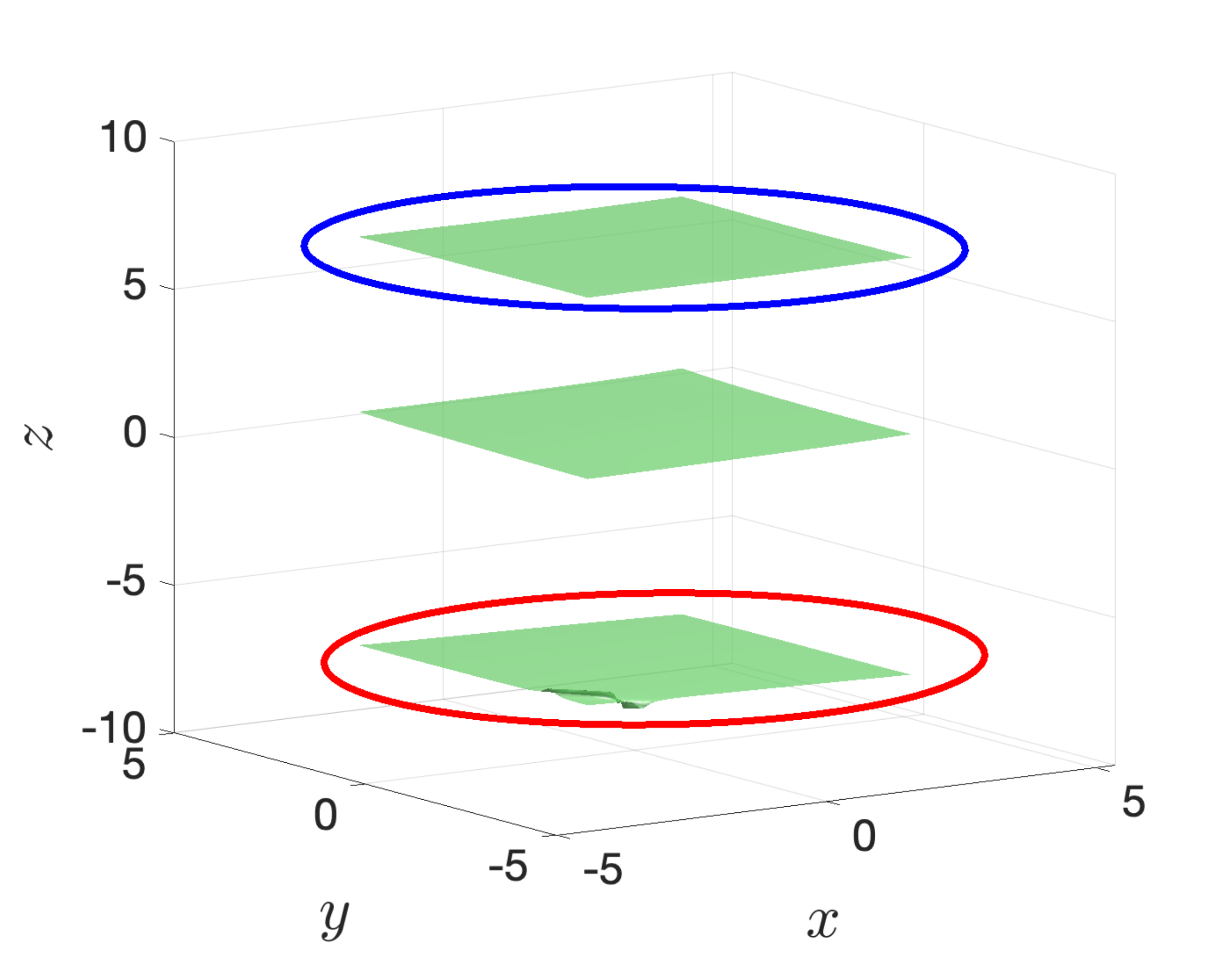}}
   \subfigure[]{\includegraphics[width=1.65in]{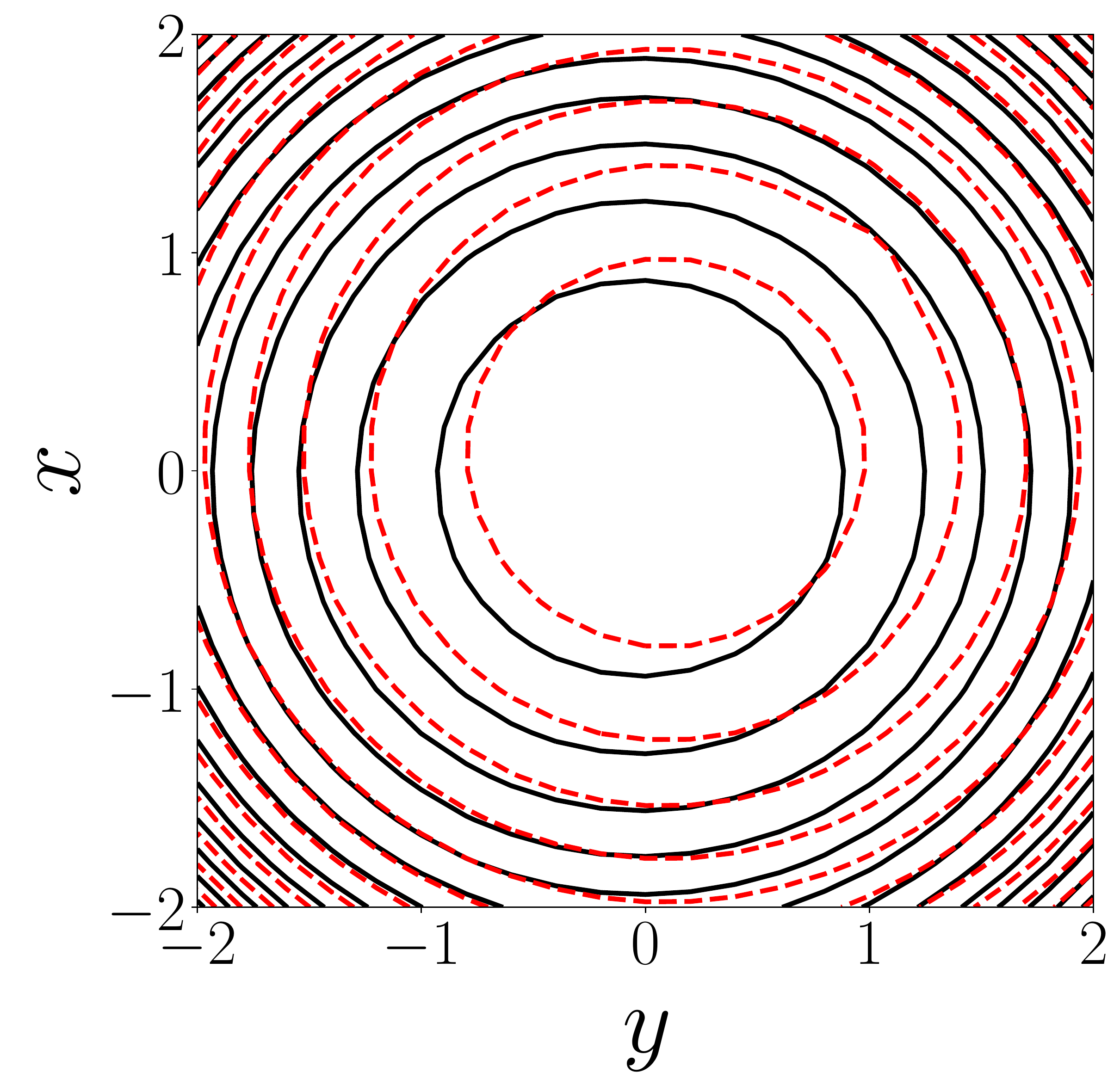}}
   \subfigure[]{\includegraphics[width=1.95in]{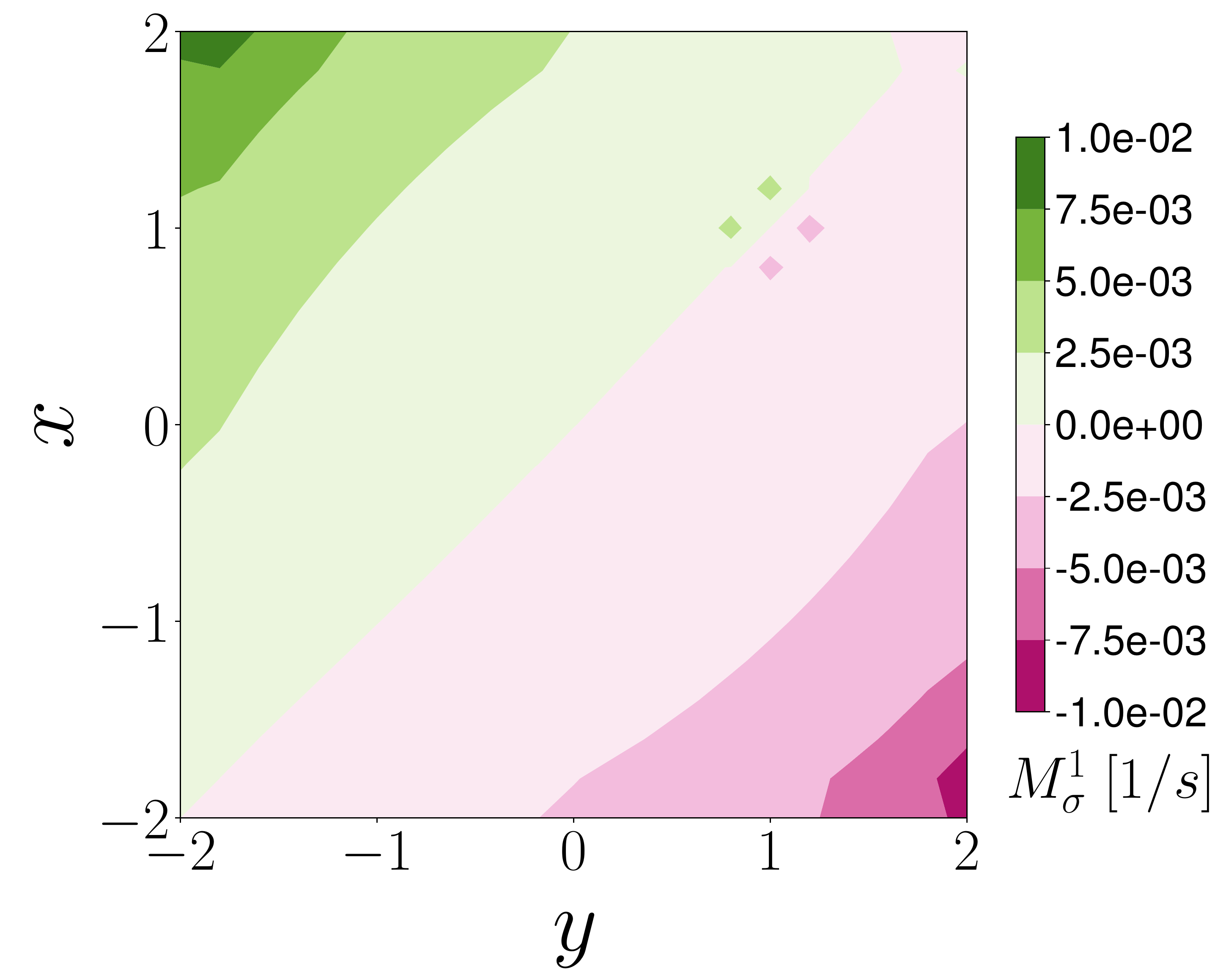}}
\caption{The effects of breaking axisymmetry in a mirror machine. The field is oriented vertically, with the stronger (red) coil at the bottom, having twice the current. Axisymmetry is broken by rotating the stronger coil by $\pi/60$ about the vector $(-\sin \pi/60,\cos \pi/60,0)$ at the center of the coil and then centering it at $(1,1,-7)$ while the weaker (blue) coil is at $(0,0,7)$.
(a) the parts of $\Sigma$ above the square $|x|,|y|\le 2$; the pieces surrounded by the two coils are in $\Sigma^-$, the piece in the middle is in $\Sigma^+$; (b) level sets of $h$ and $\text{\j}$ on the upper part of $\Sigma^-$ (solid black for $h$, dashed red for $\text{\j}$); (c) the function $\cM$ on the upper part of $\Sigma^-$.}
   \label{fig:nonaximm}
\end{figure}
As described in the previous section, one can quantify the failure of contours of $h$ and $\text{\j}$ to coincide by computing $dh \wedge d\text{\j}$.  It is a multiple $\cM$ of the flux 2-form $\beta$, so the convenient way to describe it is the function $\cM$.  Figure~\ref{fig:nonaximm}(c) shows $\cM$ for the example, computed by numerical differentiation.  
%As $M$ is a function there are various ways to quantify its size. One is to take the maximum of $|M|$ over some specified subset $S$ of $\Sigma^-$.  A better one is to take $\int_S M^+ \beta$, where $M^+ = \max(0,M)$. [XXX specify what to compare it with; plus this is repetition]
%More about $M$ will come in Section~\ref{sec:Mel}.

Analysis of the motion would be completed by computing $H_j$ on $\Sigma^+$, but for isodrasticity, it is enough to look at $h$ and $\text{\j}$ on $\Sigma^-$, as we have explained.

A more complicated case is perturbations of an up-down symmetric mirror machine.  In the unperturbed case all the marginal cases are doubly marginal, because both ends of the segment are zeroes of $|B|'$.  On breaking axisymmetry (and up-down symmetry if desired) a subset of initial conditions on each piece of $\Sigma^-$ bounce before reaching the other piece of $\Sigma^-$ but another subset cross the latter at a lower value of $|B|$ and thus escape.  For these points, $\text{\j}$ is undefined. Actually, some of the escaping fieldlines might encircle a coil and come back for another approach to $\Sigma^-$ and  bounce, thereby giving a defined but larger value of $\text{\j}$, so $\text{\j}$ would have jump discontinuity lines at heteroclinic cases.
We do not discuss this further here, but the same issue is unavoidable for tokamaks, to be addressed in the next subsection.

It would be interesting to compute $\Sigma$, $h$ and $\text{\j}$ on $\Sigma^-$ for some other mirror machines, for example, a baseball coil \cite{P}.  Indeed, we treat a different mirror field in Section~\ref{sec:mirror}.

\subsection{Tokamak}
\label{sec:tok}
A simple example of an axisymmetric tokamak field, in physical components for cylindrical polar coordinates $(R,\phi,z)$, is
$$B_R = -z/R, \ B_\phi = C/R, \ B_z = r/R,$$ in a solid torus $r^2+z^2 \le r_0^2$ for some $r_0<1$, with 
$$r = R-1$$ (not minor radius) and $C > 1$ (this is a simple case of Balescu's class of standard axisymmetric magnetic fields \cite{Ba}).  It has a closed fieldline $z=0,R=1$, called the ``magnetic axis'', whose radius has been scaled to $1$. 
The field is divergence-free but we made no effort to make it magnetohydrostatic; perhaps it would be better to use a Solov'ev equilibrium \cite{S,Ce}, but we chose this one because we could do more calculations explicitly. The parameter $C$ is chosen to exceed 1 for stability to kinks, as will be discussed shortly, though not being an equilibrium, this stability condition is not really applicable.  

The fieldlines preserve 
$$\psi = \tfrac12(r^2+z^2).$$
With poloidal angle $\theta$ around the magnetic axis, one obtains $$d\phi/d\theta = C/R = C/(1+\sqrt{2\psi} \cos\theta).$$  This can be integrated to show that the change in $\phi$ for one revolution in $\theta$ is $2\pi C/\sqrt{1-2\psi}$, hence the rotational transform $\iota = \sqrt{1-2\psi}/C$. In particular, if $C>1$ then the ``safety factor'' $q = 1/\iota$ exceeds 1 on all flux surfaces, satisfying the Kruskal-Shafranov condition for stability to kinks, as claimed.

The field strength is $$|B|=\tfrac{1}{R}\sqrt{C^2+2\psi}.$$  It follows that $$|B|' = b\cdot \nabla |B| = z/R^2.$$  Thus
$\Sigma$ is the annulus $z=0$.  
Furthermore
$$|B|'' = \frac{Rr+2z^2}{|B|R^4},$$ so on $\Sigma$ $$ |B|^{\prime \prime} = \frac{r}{\sqrt{C^2+r^2}R^2}.$$
So $\Sigma^0$ is the magnetic axis $r=z=0$ and separates $\Sigma$ into $\Sigma^+$ for $r>0$ and $\Sigma^-$ for $r<0$ (see Figure~\ref{fig:exSigma}(c)).  There are passing trajectories that circulate in the same direction forever.  There are bouncing segments that cross $\Sigma^+$ repeatedly, bouncing at the stronger field where $R$ is smaller; they give the bananas.

As for the axisymmetric mirror machine, $j = J(h,r)$ for some function of $h=|B|$ at the bounce points and $r$ the value at which the segment crosses $\Sigma^+$.  It is defined for $$\frac{\sqrt{C^2+r^2}}{1+r} \le h \le \frac{\sqrt{C^2+r^2}}{1-r},$$ 
the limits corresponding to the field strengths on $\Sigma^+$ and $\Sigma^-$ respectively.
The function
$J$ increases with $h$, to a maximum of $j^*(r)$ corresponding to the upper limit of $h$.  It is plausible also that it increases with $r$.  The Hamiltonian $H_j$ on $\Sigma^+$ for motion with given $j$ is again defined by $J(H_j,r)=j$, on the subset for which bouncing motion with the given $j$ is possible.  This is the set with $r \ge j^{*-1}(j)$.
Hence
$dH_j/dr = -\frac{\partial J/\partial r}{\partial J /\partial h}$ is negative (except on the magnetic axis, but that is relevant for only $j=0$).  Its level sets are axisymmetric, and the bananas precess at constant rate.

If axisymmetry is broken then $\Sigma$ deforms to a nearby surface, because $\partial_z |B|' = R^{-2} \ne 0$, so the implicit function theorem applies.  By the same argument, $\Sigma^0$ deforms to a nearby curve (on the surface) because $\partial_R |B|'' = \frac{1}{C} \ne 0$.  Furthermore, the magnetic axis deforms to a nearby closed curve if its rotational transform $\iota_0 \notin \Z$, by persistence of non-degenerate fixed points of the return map of fieldline flow to a transverse section.  Recall that for this example, $\iota_0 = 1/C \in (0,1)$.  If $\iota_0 \notin \Z/2$ then the perturbed magnetic axis remains elliptic.  But typically the new magnetic axis is not contained in $\Sigma$.  All fieldlines intersect $\Sigma^+$ and $\Sigma^-$ alternately, except for intersections with $\Sigma^0$.  The intersections with $\Sigma^0$ are tangential to $\Sigma$ so produce no crossing except at special points where the contact is of odd order.

The level sets of $H_j$ deform smoothly except near $\Sigma^0$ and near the cases of heteroclinic orbits from $\Sigma^-$ to $\Sigma^-$.  Thus a lot of the motion remains similar to the axisymmetric case, with precessing bananas, but the motion is qualitatively different near $\Sigma^0$ and near the former heteroclinic cases; the result in both cases is transitions between bouncing and passing.  

% \sn{Note to self: Double checking the following result in Fig. 10}
Figure~\ref{fig:tok} shows an example of $\Sigma$ for a perturbed tokamak, including the effect on $\Sigma^0$ and level curves of $|B|$ on $\Sigma$.
For this, we consider the physical components of a vector potential $A_{\varepsilon} = (0,\varepsilon z R \cos\phi,0)$ in cylindrical polar coordinates and perturb the axisymmetric field by adding $B_{\varepsilon} = \nabla \times A_{\varepsilon}$. We give expressions for $|B|$ and $ |B|^{\prime}$ in  Appendix~\ref{appsect:pert_tokamak} for sake of completeness. %\rsm{Actually, $|B|''$ is not there}

\begin{figure}[htbp] %  figure placement: here, top, bottom, or page
\centering
% \subfigure[]{\includegraphics[width=2.8in]{sigma_surface_formI_C5.50epsilon0.02_z-modB}}
\subfigure[]{\includegraphics[width=5in]{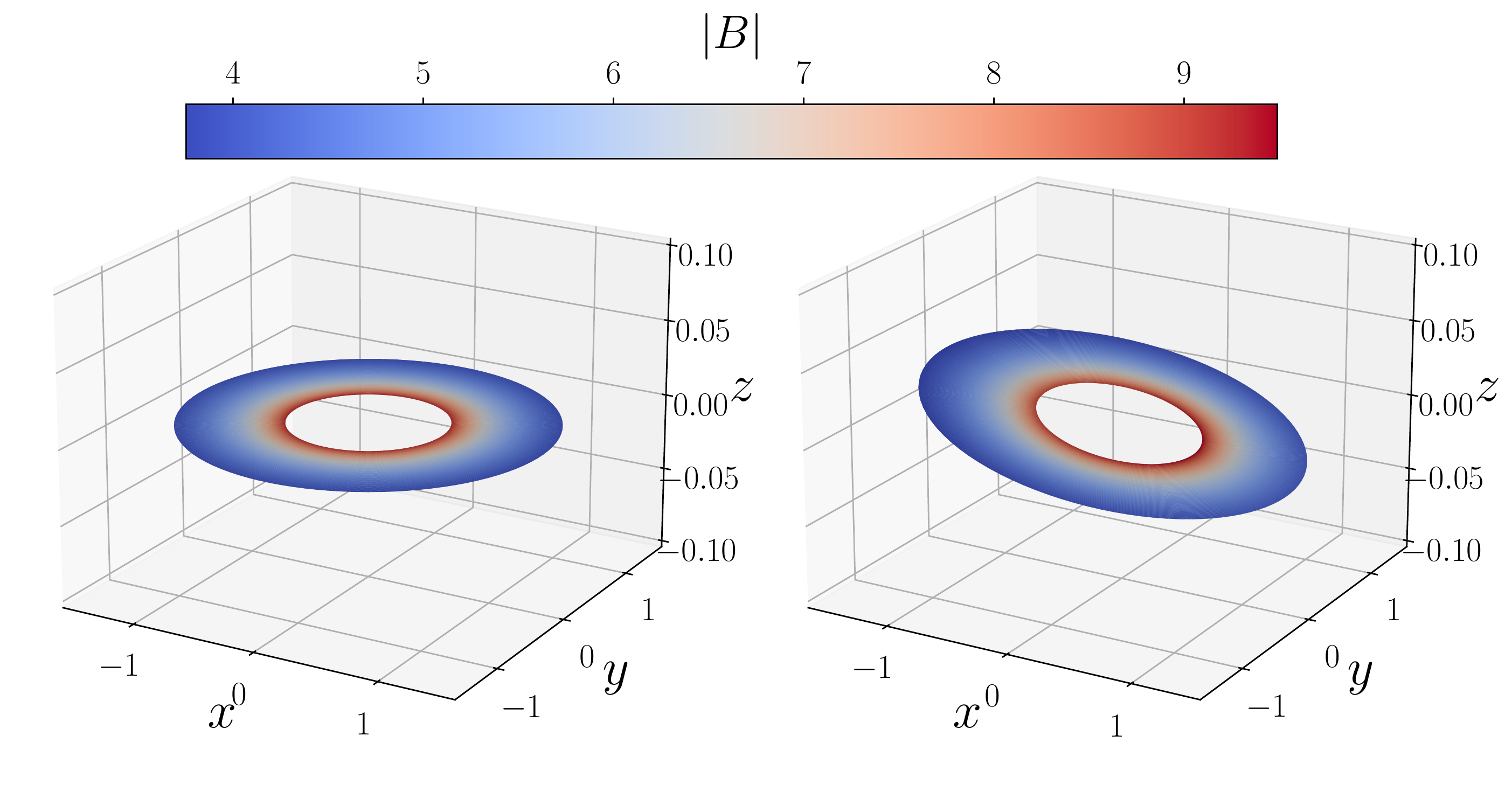}}
\subfigure[]{\includegraphics[width=2.5in]{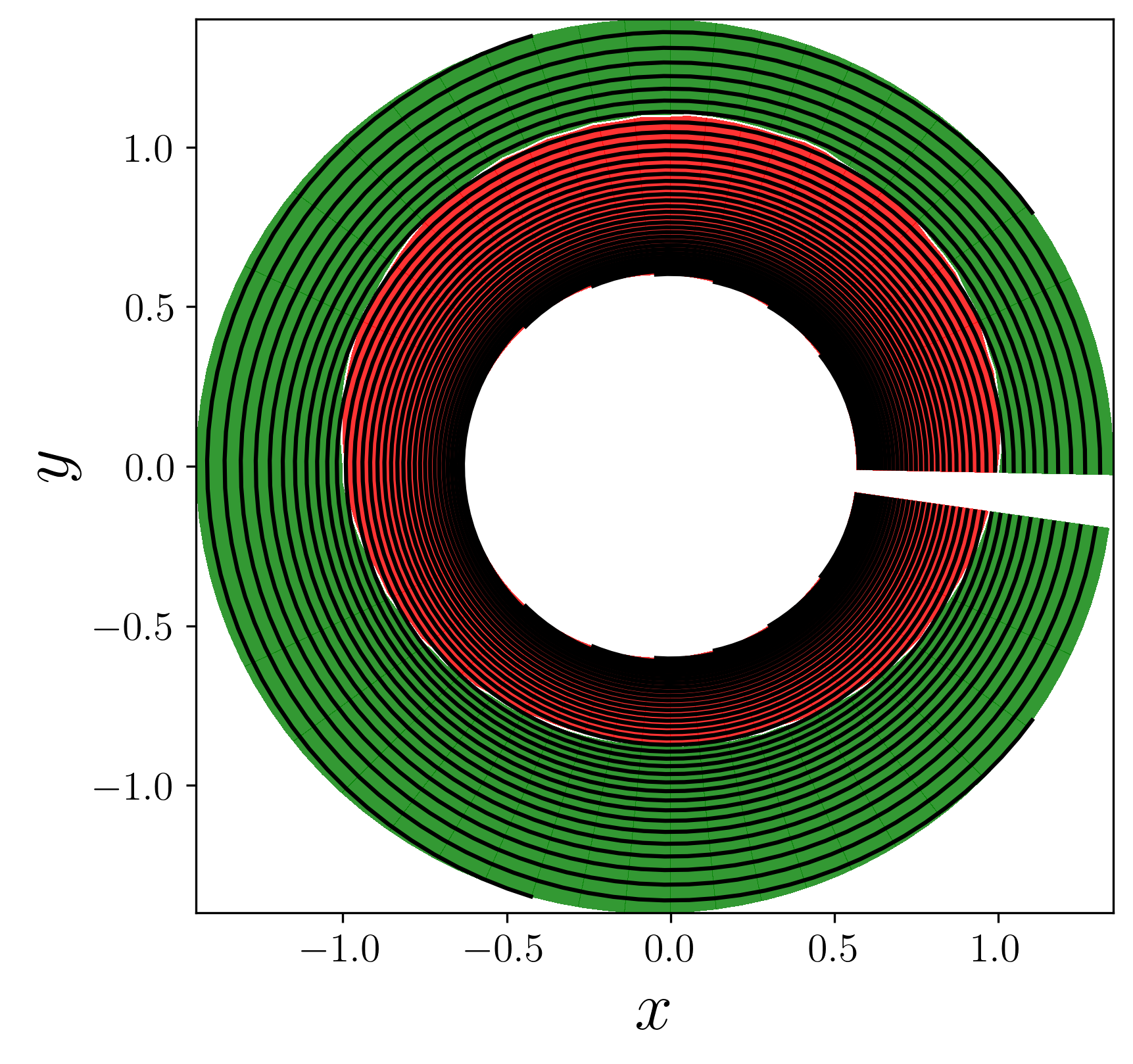}}
\caption{(a) $\Sigma$ for the (left) axisymmetric tokamak with $\eps=0$ and the (right) perturbed tokamak with $\eps=0.02$, for $C=5.5$.    
(b) Partition of $\Sigma$ for $C=5.5$, $\eps=0.02$ into $\Sigma^-$ (red), $\Sigma^+$ (green), separated by $\Sigma^0$, and level sets of $h = |B|$ on $\Sigma$.}
%\sn{Color the $|B|$ on the left instead of $z$. Think about 3d view to show the effect of perturbation.}\rsm{yes, I think starting with a 3D view would be good after all; just need to scale z differently from x,y.}\sn{ok, let me remake the plot and share with you both a few attempts.}}
%XXX Add level sets of $\text{\j}$ on $\Sigma^-$?} and (c) $j$ contours on the $\Sigma^-$ surface for $C=5.5$, $\eps=0.01$.
\label{fig:tok}
\end{figure}

Note from the righthand panel that some of the level curves of $|B|$ on $\Sigma$ cross from $\Sigma^+$ to $\Sigma^-$.  From this we deduce (following the discussion in the previous section) that some short bouncers (deeply trapped) drift into $\Sigma^-$ where they become unstable. 
%\rsm{Is this a good interpretation?}

Computation of $\text{\j}$ on $\Sigma^-$ is complicated by the fact that the unperturbed ZGC trajectories from $\Sigma^-$ are mostly heteroclinic (reaching critical points at both ends) rather than homoclinic (returning to the same critical point).  Under small perturbation, the values of $|B|$ at the two crossings of a fieldline with $\Sigma^-$ in general become different.  This implies that under small perturbation they may bounce just before reaching $\Sigma^-$ again or they may cross $\Sigma^-$ and make another poloidal revolution before approaching $\Sigma^-$ again.  They may bounce there or cross again, etc.  There may even be trajectories that never bounce.  Thus $\Sigma^-$ is divided into many components labelled by the number $N \in \{0,1,2\ldots\}$ of crossings with $\Sigma^-$ before bouncing. They are separated by the curves on which the fieldline from the given point of $\Sigma^-$ crosses $\Sigma^-$ some number $M$ of times at lower values of $|B|$ than it started and then reaches $\Sigma^-$ at a point with exactly the same value of $|B|$ as it started.  On moving the starting point on $\Sigma^-$ across the curve labelled by $M$, the number $N$ changes from something less than $M$ to something at least $M$, but they are not necessarily neighbours.  The function $\text{\j}$ has jump discontinuities at these curves.
Interpreting this picture becomes challenging, though the simplest option is just to keep the component $0$ and consider the rest of $\Sigma^-$ to go to the ``circulating'' class.

In addition, one should remember that different functions $\text{\j}^\pm$ are defined for each direction along the field, so one should plot two pictures of level curves of $\text{\j}$, and that the circulating classes for the two are in opposite directions.

Near $\Sigma^0$, the picture is particularly intricate.
Analysis of what happens near $C^4$-generic $\Sigma^0$ is carried out in Appendix~\ref{app:generic}.  
In particular, near a generic point of $\Sigma^0$ there are short bouncing segments in the well of a cubic (see Figure~\ref{fig:genSigma0}) and they make transitions to longer ones and vice versa.  
Necessary and sufficient conditions for the short bouncing class to make no transitions are derived in Appendix~\ref{app:generic}.
Also treated there are the longer bouncing classes that come close to $\Sigma^0$.
Breaking of axisymmetry in a tokamak induces many transitions between classes and hence some form of mixing in the core. 
Mixing in the core is not necessarily a bad thing, however; \cite{Boozer} makes the case for ``annular confinement''.

The picture will become clearer when we go to the exact version of the theory in the second half of this paper.
The effect of the drifts in FGCM is in general to replace the equilibria of ZGCM by periodic orbits.
We will analyse the effects of this after introducing strong isodrasticity, but for now, we conclude that for typical perturbation of the tokamak example, at any value of $\psi$ there are in general repeated transitions between different types of bouncing trajectory and passing trajectories.  They can produce relatively rapid diffusion in $\psi$, which is bad for confinement.  Hence the desire to make the field isodrastic.

The same issue about critical bouncers being critical at both ends rather than just one end holds for all quasi-symmetric fields.  This is because for a quasisymmetric field $|B|$ is constant along the lines of the symmetry field.

\section{Realisation of weak isodrasticity}
\label{sec:construct}
Can isodrastic fields be realised, outside omnigenity?  
Here we show how one can construct many weakly isodrastic mirror fields that are not omnigenous.  

Firstly, we construct such examples in the form of field strength as a function of fieldline coordinates.  Then we prove under some additional conditions that such a function can be realised by a divergence-free field in Euclidean $\R^3$.

\subsection{Construction in fieldline coordinates}

By fieldline coordinates for a magnetic field, we mean a pair of fieldline labels $u,v$ with independent derivatives, and (signed) arclength $s$ along the fieldline from a transverse reference surface.

Choose a positive $C^2$ function $h$ on a disk with coordinates $(u,v)$, with a non-degenerate minimum at $(0,0)$ and no other critical points, so its level sets are nested closed curves around the origin.  It will represent $|B|$ on $\Sigma^-$.  Extend $h$ to a $C^2$ function $\cB$ of $(s,u,v)$ for an interval of $s$ with $\cB(0,u,v)=h(u,v)$, such that along each line of constant $(u,v)$, $\cB$ has a non-degenerate local maximum at $s=0$, a minimum at some $s_m(u,v)>0$ and a first point $s_b(u,v)> s_m(u,v)$ at which $\cB=h$ with positive $s$-derivative.  $\cB$ will represent $|B|$ along the fieldlines.  The functions $s_m$ and $s_b$ are to be chosen $C^2$.  Most importantly, we require also that $$\j(u,v) = \int_{0}^{s_b(u,v)} \sqrt{2(h(u,v)-\cB(s,u,v))}\, ds$$ be a function of $h$, call it $\j = J(h)$.
There is a lot of freedom in these choices.  

Any magnetic field realising such a function $\cB$ is weakly isodrastic, but in general it is not omnigenous. Fields realising $\cB$ automatically have a flux function, namely the value of $h$ as a function of fieldline labels $(u,v)$.  By the isodrastic condition, any other flux-function has to have the same flux surfaces.  Thus it is omnigenous iff in addition, $j(E,u,v)=\int_\gamma \sqrt{2(E-\cB(s,u,v))}\, ds$ is a function of just $E$ and $h(u,v)$, where $\gamma$ is the segment of fieldline with $\cB(s,u,v)\le E$.  In particular, if it is omnigenous then $\cB(s_m(u,v),u,v)$ has to have the same value along each line with the same value of $h$.  It is easy to make counterexamples.

%Any magnetic field without recurrence has many flux functions:~they are given by choosing any function of fieldline labels that has non-zero derivative almost everywhere. But we can make examples such that for no choice of flux function $\psi$ is $j(x,y,E)=\int_\gamma \sqrt{2(E-k)}\, dz$ a function of just $E$ and $\psi$, where $\gamma$ is the segment of fieldline between the points with $k(x,y,z)=E$.  For starters, the flux surfaces would have to be $h=$ constant in order to satisfy this for the homoclinics.  Fix a value of $h$ and let $\theta$ be an angle around the curve $h(x,y)=$ constant.  The function $j$ is thus a function of $E$ and $\theta$.  As a function of $E$ it is equivalent to the function $\ell$ specifying the length of $\gamma$, by
%$j(E) = \int^E \sqrt{2(E-v)}\, d\ell(v)$.
%We have the weak isodrastic constraint that $j(h)$ has to be the same for all $\theta$, but otherwise we are free to make $j$ vary with $\theta$ for fixed $E<h$, by varying the function $\ell$.  For example, there is no need for $k$ to have the same minimum value for each $\theta$. 

Concretely, let us take $\cB$ to be a cubic function of arclength along each fieldline: $$\cB(s) = c+ r^2 -a(r,\theta)s^2+b(r,\theta)s^3$$ along the fieldline starting from $(r,\theta)$ in polar coordinates on $\Sigma^-$.  One can take $a$ and $b$ in the form of polynomials $a(r,\theta) = \Re \sum_{n=0}^{N} a_n r^n e^{i n\theta}$, for example ($N=1$ suffices).
Then $h=c+r^2$ has level curves $r=$ constant, and 
$$\j = \int_0^{a/b} \sqrt{2(as^2-bs^3)}\, ds = \frac{4\sqrt{2} a^{5/2}}{15 b^2}.$$
Thus we have an isodrastic field iff $b^2 \propto a^{5/2}$ on $r=$ constant.  But the well depth (difference in $|B|$ between the local maximum and minimum) is $\frac{4 a^3}{27 b^2}$, which can be varied independently of keeping $\j$ constant on $r=$ constant.

Let $f(r^2)  = 1 + \frac{r^2}{1+r^2}$ and
\begin{align*}
     a(u,v)  = \frac{\bigg(1 + \alpha \left(\tfrac{1}{3}u^2 + \tfrac{2}{3}v^2\right)\bigg)^2}{f(u^2+v^2)^2},\quad b(u,v)  = \frac{\bigg(1 + \alpha \left(\tfrac{1}{3}u^2 + \tfrac{2}{3}v^2\right)\bigg)^{5/2}}{f(u^2+v^2)^3},
\end{align*}
where $\alpha\geq 0$ is a real parameter ($a$ and $b$ are expressed here in Cartesian rather than polar coordinates, to facilitate seeing that the result is smooth). For the field $\mathcal{B} = 1 + u^2 + v^2 -a s^2 + b s^3$, $h(u,v) = 1+u^2 + v^2$ and $s_m(u,v) = \tfrac{2}{3}\tfrac{a(u,v)}{b(u,v)}$, consistent with the above discussion. The separatrix action $\j(u,v)$ is given by
\begin{align*}
    \j(u,v) = \frac{4\sqrt{2} a^{5/2}}{15 b^2} = \frac{4\sqrt{2}}{15}f(u^2 + v^2).
\end{align*}
This field is therefore weakly isodrastic for each value of $\alpha$, as illustrated in Fig. \ref{fig:isodras_demo}.
\begin{figure}
    \centering
    \includegraphics[scale=.5]{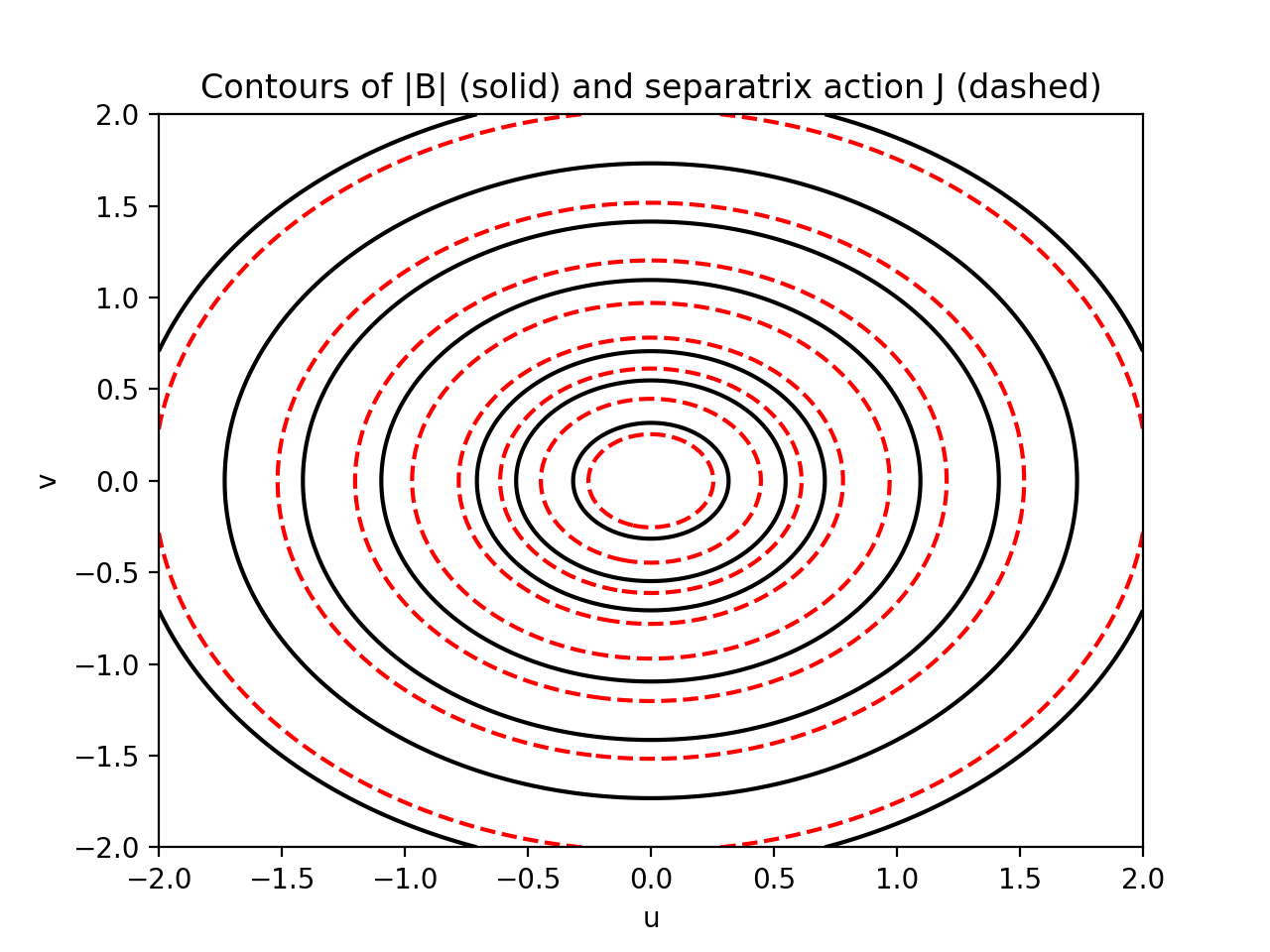}
    \caption{Contours of $h=|B|$ and $\j$ on $\Sigma^-$ for the isodrastic mirror example with $\alpha =0.3$.}
    \label{fig:isodras_demo}
\end{figure}
On the other hand,
the field strength at the local minimum $s_m$ is given by
\begin{align*}
    \mathcal{B}(s_m(u,v),u,v) & = 1 + u^2 + v^2 - \tfrac{4}{27}\bigg(1 + \alpha \left(\tfrac{1}{3}u^2 + \tfrac{2}{3}v^2\right)\bigg),
\end{align*}
which is not constant along the circles $h = \text{const.}$ when $\alpha\neq 0$, as shown in Fig. \ref{fig:modB_contours_flux}.
\begin{figure}
    \centering
    \includegraphics[scale=.7]{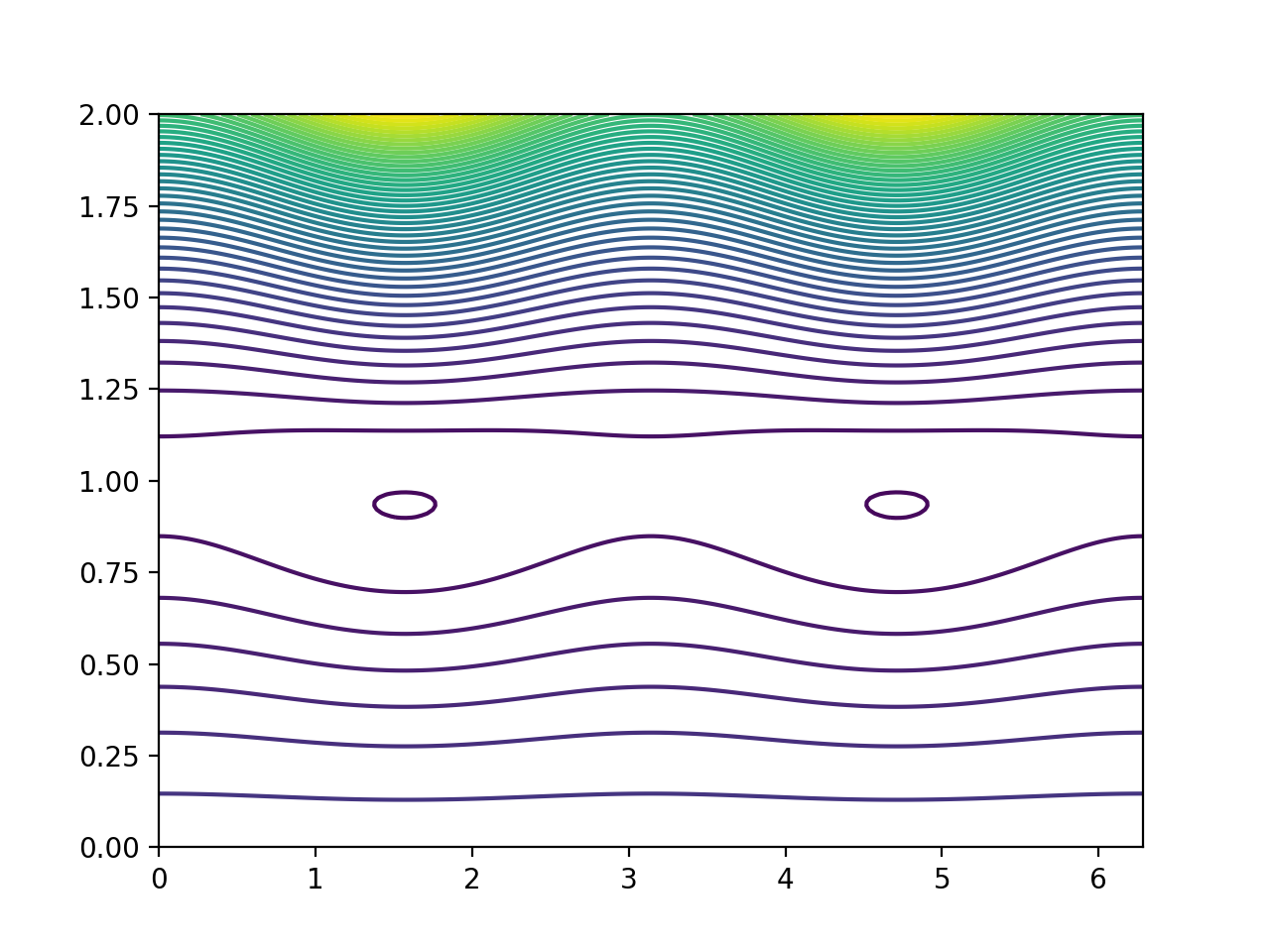}
    \caption{$\mathcal{B}$-contours on the flux surface $\sqrt{u^2 + v^2} = 1.2$ when $\alpha = 0.3$.}
    \label{fig:modB_contours_flux}
\end{figure}
It follows that the field is not omnigenous for nonzero $\alpha$. In fact the field is omnigenous iff $\alpha = 0$. Also, Figure \ref{fig:modB_contours_flux}  shows that this field is not pseudosymmetric; see Appendix \ref{app:ps}. Finally, note that a field that is not omnigenous cannot admit a rotation symmetry. Therefore our construction produces a truly $3$-dimensional weakly-isodrastic field.

\subsection{Realisability by divergence-free field in Euclidean space}
\label{sec:real}
%%%
%\section{Realizability of $B(s,u,v)$}

In this subsection we address the general question whether a given function $\cB$ of fieldline coordinates can be realised as the field strength along the fieldlines of a divergence-free field in $\R^3$.  We obtain a positive answer locally if the function is analytic.

We suspect that analyticity is not really necessary, but it suffices for realising examples, such as the cubic examples of the previous subsection.  We believe also that although the result is local, if $c$ is large enough then it applies at least as far as $s = a/b$, thus containing all the bouncing trajectories of those examples.

Our construction makes a closed two-form $\beta$ representing the magnetic flux-form for a magnetic field $B$. Given a volume-form $\Omega$, it is immediate from this to construct $B$ as the unique vector field such that $i_B\Omega = \beta$, and $\beta$ being closed is equivalent to $B$ being divergence-free.

To make sure that $\beta$ is closed we will find a diffeomorphism $\phi$ from a neighbourhood $V$ of $(0,0,0)$ in $(s,u,v)$ space to a neighbourhood $U$ of a point $p$ in physical space such that the pullback $\phi^*\beta = du \wedge dv$.
This makes $u$ and $v$ locally into Clebsch coordinates for the field.  Then $\beta$ is closed because $du \wedge dv$ is closed.  $\phi^*\beta = du \wedge dv$ looks like a large constraint.  It could be relaxed to $\phi^*\beta = f(u,v)du \wedge dv$ for any positive function $f$, but we will see that there is still an immense amount of freedom for the construction.

The two remaining constraints are that $\partial_s \phi$ should be a unit vector ($s$ is supposed to represent arclength) and that $|B|\circ \phi = \cB$ (which can be written as $\phi^*|B|=\cB$).

Although the case of interest is $\R^3$ with the standard Euclidean metric, our construction can be done in an arbitrary 3D Riemannian manifold, so we start with the more general formulation.

\begin{defn}
Let $\beta$ be a nowhere-vanishing closed $2$-form on a Riemannian $3$-manifold $(M,g)$. Use $(s,u,v)$ to denote the standard Euclidean coordinate system on $\mathbb{R}^3$. A system of \textbf{Clebsch coordinates} for $\beta$ near $p\in M$ comprises a pair of open sets, $U\subset M$ containing $p$, and $V\subset \mathbb{R}^3$ containing $0$, together with a diffeomorphism $\varphi:V\rightarrow U$ such that
\begin{align*}
    \varphi^*\beta = du \wedge dv,
\end{align*}
and
\begin{align*}
    g_{\varphi(s,u,v)}(\partial_s\varphi(s,u,v),\partial_s\varphi(s,u,v)) = 1,\quad (s,u,v)\in V.
\end{align*}
\end{defn}

\begin{defn}
Let $\cB:V\subset \mathbb{R}^3\rightarrow \mathbb{R}$ be a positive smooth function. We say $\cB$ is \textbf{realizable} on a Riemannian $3$-manifold $(M,g)$ if there exists a nowhere-vanishing closed $2$-form $\beta$ on $M$ and a system of Clebsch coordinates $(p,U,V)$ for $\beta$ such that $\varphi^*|\beta| = \cB$. Here $|\beta|$ denotes the pointwise norm of $\beta$ defined by the inner product of $2$-forms induced by $g$.
\end{defn}

\begin{thm}
Let $g = dx^2 + dy^2 + dz^2$ denote the standard flat Riemannian metric on $\mathbb{R}^3\ni (x,y,z)$, and let $M\subset \mathbb{R}^3$ be an open set. A positive function $\cB:V\subset \mathbb{R}^3\rightarrow\mathbb{R}$ is realizable on $(M,g)$ if and only if there exist smooth functions $X,Y,Z:V\rightarrow \mathbb{R}^3$ that satisfy the system of partial differential equations
\begin{gather}
    \text{det}\begin{pmatrix}
    \partial_sX & \partial_sY & \partial_sZ\\
    \partial_uX & \partial_uY & \partial_uZ\\
    \partial_vX & \partial_vY & \partial_vZ\\
    \end{pmatrix} = \frac{1}{\cB(s,u,v)}\label{Clebsch_1}\\
    (\partial_sX)^2 + (\partial_sY)^2 + (\partial_sZ)^2 =1.\label{Clebsch_2}
\end{gather}
\end{thm}
\begin{proof}
First suppose that $\cB$ is realizable. Then we have a nowhere-vanishing closed $2$-form $\beta$ on $M$ and a diffeomorphism $\varphi: V\rightarrow U:(s,u,v)\mapsto (X,Y,Z)$ such that $\varphi^*\beta = du\wedge dv$, $g(\partial_s\varphi,\partial_s\varphi) = 1$, and $\varphi^*|\beta| = \cB$. Let $\Omega = dx\wedge dy\wedge dz$ denote the Euclidean volume form on $U$ and let $e_s$ denote the standard basis vector along the $s$-axis in $V$. There is a unique nowhere-vanishing vector field $\bm{B}$ on $U$ such that $\iota_{\bm{B}}\Omega = \beta$. We also know that the vector field $\partial_s = \varphi_*e_s$ on $U$ satisfies $\iota_{\partial_s}\beta = 0$, since 
\begin{align*}
    \iota_{\partial_s}\beta = \iota_{\partial_s}\varphi_*(du\wedge dv) = \varphi_*(\iota_{e_s}du\wedge dv) = 0.
\end{align*}
Since the null space for $\beta$ is one-dimensional by hypothesis there must therefore be a smooth function $\alpha$ with $\partial_s = \alpha\,\bm{B}$. But since
\begin{align*}
    1 = g_{\varphi(s,u,v)}(\partial_s\varphi(s,u,v),\partial_s\varphi(s,u,v)) = g(\partial_s,\partial_s)\circ \varphi,
\end{align*}
$\partial_s$ must be a unit vector. Therefore the function $\alpha$ is given by $\alpha = 1/|\bm{B}| = 1/|\beta|$ (using the standard result that $|\beta| = |\bm{B}|$). We arrive then at the useful identity
\begin{align*}
    \iota_{|\beta|\partial_s}\Omega = \iota_{\bm{B}}\Omega = \beta.
\end{align*}
Upon introducing the Jacobian determinant
\begin{align*}
    J = \text{det}\begin{pmatrix}
    \partial_sX & \partial_sY & \partial_sZ\\
    \partial_uX & \partial_uY & \partial_uZ\\
    \partial_vX & \partial_vY & \partial_vZ\\
    \end{pmatrix},
\end{align*}
we may express the pullback of the previous identity along $\varphi$ as 
\begin{align*}
    \iota_{\cB\,e_s}\,J\,ds\wedge du\wedge dv = du\wedge dv,
\end{align*}
which implies $J\,\cB = 1$. This formula, together with the condition 
$$1 = g_{\varphi(s,u,v)}(\partial_s\varphi(s,u,v),\partial_s\varphi(s,u,v)),$$ recovers the desired system of PDEs \eqref{Clebsch_1} and \eqref{Clebsch_2}.

Conversely, suppose that $\varphi=(X,Y,Z)$ satisfies \eqref{Clebsch_1} and \eqref{Clebsch_2}. Since $\cB$ is nowhere-vanishing, $\varphi$ is a local diffeomorphism. By restricting to an appropriate open set we may therefore assume it is a diffeomorphism onto its image. The diffeomorphism $\varphi$ defines a system of Clebsch coordinates for the $2$-form $\beta = \varphi_*(du\wedge dv)$ by \eqref{Clebsch_2}. But since $J\,\cB = 1$ by \eqref{Clebsch_1}, we may also write $\beta$ as
\begin{align*}
\beta = \varphi_*(\cB\,J\,\iota_{e_s}\,ds\wedge du\wedge dv) = \iota_{\varphi_*(\cB\,e_s)}\Omega.
\end{align*}
Using $|\varphi_*e_s| = 1$, we therefore have
\begin{align*}
    |\beta| = |\varphi_*(\cB\,e_s)| = \varphi_*\cB\,|\varphi_*e_s| = \varphi_{*}\cB,
\end{align*}
which says that $\cB$ is realizable.
\end{proof}

\begin{thm}
For each real analytic positive $\cB$ there is an open set $V\subset \mathbb{R}^3\ni (s,u,v)$ containing the origin and real analytic functions $X,Y,Z$ defined on $V$ such that
\begin{gather}
    \text{det}\begin{pmatrix}
    \partial_sX & \partial_sY & \partial_sZ\\
    \partial_uX & \partial_uY & \partial_uZ\\
    \partial_vX & \partial_vY & \partial_vZ\\
    \end{pmatrix} = \frac{1}{\cB(s,u,v)}\label{Clebsch_1_thm}\\
    (\partial_sX)^2 + (\partial_sY)^2 + (\partial_sZ)^2 =1.\label{Clebsch_2_thm}
\end{gather}
\end{thm}
\begin{proof}
The proof is an application of the Cauchy-Kowalevski theorem \cite{RR}. There is still a lot of freedom, so we will establish existence of a solution with $Z = v$. For such solutions the PDE system reduces to
\begin{gather*}
    \partial_sX\,\partial_uY - \partial_sY\,\partial_uX = \frac{1}{\cB(s,u,v)}\\
    (\partial_sX)^2 + (\partial_s Y)^2 = 1.
\end{gather*}
Let $B_0 = \cB(0,0,0)$. Observe that $X_0 = s + u/B_0$, $Y_0 = u/B_0$ solves the system at the origin since
\begin{align*}
    &\partial_s X_0\,\partial_u Y_0 - \partial_s Y_0\,\partial_u X_0  = B_0^{-1}\\
    &(\partial_sX_0)^2 + (\partial_sY_0)^2 = 1.
\end{align*}
We will therefore also restrict our search to solutions near $(X_0,Y_0)$ (this means $B$ is principally along the $X$-direction).

Assuming $(X,Y)$ is close to $(X_0,Y_0)$, we may reformulate the reduced system by solving for $\partial_s X$ and $\partial_s Y$ according to
\begin{align*}
    \partial_s Y  & = \frac{\partial_u X/\partial_u Y}{\cB\,\partial_uY}\left(\frac{-1 + \sqrt{1 - [1+(\partial_u X/\partial_u Y)^2][1-(\cB\,\partial_u Y)^2]/(\partial_uX/\partial_u Y)^2}}{1+(\partial_u X/\partial_u Y)^2}\right)\\
    \partial_s X & = \frac{1}{\cB\,\partial_u Y} + \frac{(\partial_u X/\partial_u Y)^2}{\cB\,\partial_uY}\left(\frac{-1 + \sqrt{1 - [1+(\partial_u X/\partial_u Y)^2][1-(\cB\,\partial_u Y)^2]/(\partial_uX/\partial_u Y)^2}}{1+(\partial_u X/\partial_u Y)^2}\right).
\end{align*}
Since the right-hand-sides of these formulae are real analytic near 
\[
\begin{pmatrix}
s\\
u\\
v\\
X\\
Y\\
\partial_sX\\
\partial_uX\\
\partial_vX\\
\partial_sY\\
\partial_uY\\
\partial_vY
\end{pmatrix}=\begin{pmatrix}
0\\
0\\
0\\
0\\
0\\
1\\
B_0^{-1}\\
0\\
0\\
B_0^{-1}\\
0
\end{pmatrix},
\]
the Cauchy-Kowalevski theorem implies that the initial value problem
\begin{align*}
    X(0,u,v) &= B_0^{-1}u\\
    Y(0,u,v) & = B_0^{-1}u,
\end{align*}
has a unique analytic solution in some open neighborhood of $(s,u,v) = 0$ in $\mathbb{R}^3$. This $X$ and $Y$, together with $Z=v$, comprise the desired solution of the original PDE system \eqref{Clebsch_1_thm}-\eqref{Clebsch_2_thm}.
\end{proof}

%%%

Estimating a neighbourhood in which the Cauchy-Kowalevski theorem applies requires some work.  An example where this has been done is \cite{GRR}, in which a given analytic magnetic field on a given 2D region with analytic boundary is proved to have a vacuum field extension to a neighbourhood.

Perhaps there are alternative proofs not requiring analyticity, but the results are likely to still be local in character.

The real challenge is to make a non-trivial isodrastic stellarator field.  That will have to wait for a future publication.  The cases of heteroclinic and homoclinic connections have to be addressed.

We close this section by commenting that, despite claims in the literature, it is not clear whether omnigenity can be realised outside axisymmetry.  References like \cite{CS,PCHL} construct $|B|$ as a function of Boozer coordinates, but it is not evident that one can realise an arbitrary such function as the strength of a divergence-free field in $\R^3$.

\section{Strong Isodrasticity}
\label{sec:exact}
The treatment of weak isodrasticity rests on assuming conservation of the adiabatic invariant $j$, but that assumption fails near the transitions.
%, so the treatment is physically inconsistent.

So now we develop a version that does not assume conservation of $j$.  We derive an exact condition for absence of transitions, which we call ``strong isodrasticity''.  We illustrate it in Section~\ref{sec:illust} and elaborate on the theory in Section~\ref{sec:splitting}.  In Section~\ref{sec:Mel} we recover the results for weak isodrasticity as a first-order approximation.

%\subsection{Persistence of $\Sigma^-$}
The key idea is that $\Sigma^- \times \{v_\pl=0\}$ is an approximate normally hyperbolic submanifold for guiding-centre motion.  A {\em normally hyperbolic submanifold (NHS)} for a dynamical system is an invariant submanifold such that any tangential contraction or expansion is weaker than normal contraction or expansion, respectively (precise specification of this property is technical, see \cite{F,HPS}, or \cite{K} for a tutorial).  An approximate NHS is a submanifold that is close to being tangent to the vector field and similar tangential versus normal contraction and expansion comparisons hold.  

It follows from the theory of NHS that there is a locally unique true NHS $N^-$ near $\Sigma^-\times\{0\}$ for the guiding-centre dynamics.  In general, computing the true NHS near an approximate one is hard, but in this context the approximate NHS consists of equilibria so is a ``slow manifold'' and there is an algorithm to compute higher-order slow manifolds to arbitrary order (see \cite{M04} for an indication of how to get started, though higher than first order is less straightforward than that reference would lead one to believe, and \cite{Bu2} for more).  Furthermore, in our context, for $\mu>0$ the system is Hamiltonian and the initial slow manifold is symplectic (meaning that the symplectic form is non-degenerate on it)  and there is a streamlined procedure to compute arbitrarily high-order symplectic slow manifolds (\cite{M04} with the same caveat).  Even more, in our context, the resulting NHS has only 1DoF so consists principally of periodic orbits plus some equilibrium points and homoclinic or heteroclinic orbits between them.

For this discussion, assuming $\mu \ne 0$, it is simplest to treat FGCM in a scaled time $\tau = \sqrt{\frac{\mu}{m}} t$, scaling velocity to $u = \sqrt{\frac{m}{\mu}} v$ and magnetic moment to $\tilde{\mu} = \frac{m}{e^2}\mu$.  
Then FGCM becomes
\begin{eqnarray}
\frac{dX}{d\tau} &=& \frac{1}{\tilde{B}_\pl} (u\tilde{B} + \sqrt{\tilde{\mu}}\, b\times\nabla|B|) \\
\frac{du}{d\tau} &=& -\frac{\tilde{B}\cdot\nabla |B|}{\tilde{B}_\pl} \\
\tilde{B} &=& B + \sqrt{\tilde{\mu}}\, u\ \curl\, b,
\end{eqnarray}
and is the Hamiltonian dynamics of the scaled Hamiltonian and symplectic form
\begin{eqnarray}
\tilde{H} = \tfrac{1}{\mu} H &=& \tfrac12 u^2 + |B| \\
\tilde{\omega} = \tfrac{1}{\sqrt{m\mu}}\omega &= & \tfrac{1}{\sqrt{\tilde{\mu}}} \beta +\, d(u b^\flat). \label{eq:scaledomega}
\end{eqnarray}
This scaling reduces the set of parameters to just $\sqrt{\tilde{\mu}}$.  For the excluded limiting case $\tilde{\mu}=0$, the inverse square root in $\tilde{\omega}$  looks singular, but recall that it is the inverse of the symplectic form (the Poisson bracket) that gives the dynamics; the Poisson bracket is degenerate at $\tilde{\mu}=0$ leading to the conservation of fieldline (this is a case of Casimirs for degenerate Poisson brackets, e.g.~\cite{MR}). Thus the dynamic for $\tilde{\mu}=0$ is a well defined case, namely the motion of a unit mass in potential $|B|$ along each fieldline.  
This scaling also allows one to extend to higher-order guiding-centre approximations (with corrections to $\tilde{H}$ and $\tilde{\omega}$) that are relevant for high energy, in particular for the $\alpha$-particles produced by $D-T$ fusion.
One could also non-dimensionalise arclength $s$ by a typical lengthscale $\ell$ for variation of $B$, $B$ by a typical field-strength $B_0$, and $\tilde{\mu}$ by $1/B_0$, but little is gained by this.

Applying the symplectic slow manifold method of \cite{M04} to leading order in $\sqrt{\tilde{\mu}}$ produces $N^-$ as a graph over $\Sigma^-$ (see Appendix~\ref{app:slow}). There is a displacement tangent to $\Sigma$, which plays negligible role, and a scaled parallel velocity
$$u = U(x) =  \sqrt{\tilde{\mu}} \frac{\Omega(b,\nabla|B|',\nabla |B|)}{|B||B|''}.$$
Approximations to $N^-$ can alternatively be computed by expanding and solving the PDE expressing invariance of a graph to desired order (this may appear in a separate paper).

The dynamics on $N^-$ is given by the restrictions of the guiding-centre Hamiltonian and symplectic form to it (the restriction of the symplectic form is non-degenerate).  Being 2D, the bounded trajectories are mostly periodic, the exceptions being equilibria and trajectories connecting them. 

NHS have forward and backward contracting submanifolds $W^\pm$ (usually called stable and unstable manifolds respectively, but that terminology is inconsistent with the concepts of stable and unstable sets), consisting of the set of points whose trajectory in the stated direction of time converges to the NHS.  They are made up of sub-submanifolds $W^\pm(x)$ for each point $x$ of the NHS (Arnol'd's ingoing and outgoing ``whiskers''\cite{Ar}), consisting of the set of points whose trajectory in the stated direction of time converges together with the trajectory of $x$.

To prevent transitions, the main part of our strong isodrastic condition is that the relevant branches of $W^\pm$ coincide, forming ``separatrices'':~invariant submanifolds that separate motions of different types.  A familiar example is the separatrices $\tfrac12 p^2 = 1-\cos \theta$ for the pendulum, which separate librating motion from rotating motion.  There the NHS is just a saddle point in 2D, but the same idea extends to higher dimensions (2D NHS in 4D in our case). 

Perfect separatrices are achieved by integrable systems. In 2DoF, integrability corresponds to a continuous symmetry.  In the GCM context,  integrability is implied by quasisymmetry \cite{BKM}, but perfect quasisymmetry is perhaps not achievable outside of axisymmetry.  Even if one allows a velocity-dependent symmetry (as in weak quasisymmetry) \cite{BKM2}, we are not aware of any exact examples.

Integrable systems are not the only way to obtain perfect separatrices in Hamiltonian systems, however; there are constructions with perfect separatrices that are not integrable (see Appendix~\ref{app:sep}).  Thus there is hope that one might be able to achieve this for GCM.

A little care is required in the above construction of $N^-$, however, because the theory of NHS requires $|B|'' \le c$ for some constant $c < 0$ (depending on the perturbation size), so it fails near the boundary $\Sigma^0$ of $\Sigma^-$ (if it has boundary).  For cases with no $\Sigma^0$, like the mirror machine, nothing needs doing, but for cases like the tokamak, one has potentially to exclude a neighbourhood of $\Sigma^0$ in the construction of $N^-$.
Indeed, as will be described in the next section, when $\mu$ is turned on, $N$ for this example truly develops a gap around $\Sigma^0$.  Nonetheless, we will see that a good understanding of $N^-$ can be obtained.

%\rsm{new formulation; what do you think?} 
To complete the strong isodrastic condition, we have to deal with the issue that if $N^-$ has a poorly defined edge then GCM trajectories on it might fall off its edge.  So we require that the above mentioned potential failure of continuation of $\Sigma^-$ to $N^-$ near $\Sigma^0$ does not occur.  Specifically, we ask for $\Sigma^{-0}\times\{0\}$ to continue to an invariant submanifold $N^{-0}$ with boundary consisting of a NHS $N^{-}$ and its boundary $N^0$.  For $N^-$ to be invariant, $N^0$ has to also be invariant. Being 1D 
%\rsm{an assumption}
, the invariance condition for $N^0$ is just that $\tilde{H}=|B|+\frac12 u^2$ is constant on it.

We suspect that the above problem does not occur if (i) $\Sigma^0$ is generic, as per Appendix~\ref{app:generic}, and (ii) $|B|$ is constant on it, as for weak isodrasticity, but have not established this (see Appendix~\ref{app:pers+} for some discussion). 
%It is likely to be a question of unfolding the degenerate Poisson bracket from $\mu=0$
So for present purposes we make the following definition.

%We choose to require that $\Sigma^0$ have a continuation to an invariant \rsm{given our discussion, I'm not sure now why we'd want invariant?} curve $N^0$, forming the common boundary between $N^\pm$ on an invariant symplectic submanifold $N$.  This means that the potential failures of continuation of $\Sigma^\pm$ near $\Sigma^0$ mentioned above are required not to occur and the invariance condition for $N^0$ holds, namely that $H=|B|+\frac{m}{2\mu}v_\pl^2$ is constant on it.   

\begin{defn} A magnetic field is {\em strongly isodrastic} if $\Sigma^{-0} \times \{0\}$ continues to a maximal invariant submanifold $N^{-0}$ with boundary for guiding-centre motion for a range of $\sqrt{\tilde{\mu}}>0$, which can be decomposed into normally hyperbolic $N^-$ and its boundary $N^0$, the relevant branches of the contracting submanifolds of $N^-$ coincide, and $\tilde{H}$ is constant along $N^0$.
\end{defn}

In the case that there is no $\Sigma^0$ and hence no $N^0$, the continuation is guaranteed, so strong isodrasticity is just the coincidence of $W^\pm$.  This applies to many mirror fields.  But in tokamak and quasisymmetric stellarators, $\Sigma^0$ is an essential feature and thus its continuation to an invariant $N^0$ and the continuation of $N^-$ right up to $N^0$ is an additional consideration for isodrasticity.

In the next section, we will illustrate how strong isodrasticity is in general lost for perturbations of axisymmetric mirror and tokamak fields.  This will lead to a quantification of failure of isodrasticity that could be useful for reducing it.

The definition allows also for use of higher-order guiding-centre approximations, relevant to the alpha-particles generated by fusion of $D$ and $T$ for example.

\section{Illustrations of the exact picture}
\label{sec:illust}
We illustrate the ideas of the previous section (construction of normally hyperbolic submanifolds for FGCM and their contracting manifolds) by a mirror machine and tokamak again.

%---two-coil version moved to text-backup.tex---%
\subsection{Mirror machine}
\label{sec:mirror}
For a mirror machine of the type described in Section~\ref{sec:mir} (not restricted to axisymmetry), there is a non-degenerate saddle point of $|B|$ near the centre of each coil.  For the gradient field $\nabla |B|$, each of them has one-dimensional downhill subspace and two-dimensional uphill subspace.  They give unstable equilibrium points of guiding-centre dynamics with $v=0$.  They are each surrounded by a family of periodic orbits of guiding-centre motion, called {\em Lyapunov orbits}, which form the 2D centre manifold of the equilibrium point.  The periodic orbits are hyperbolic and the centre manifold is normally hyperbolic.  This is a case of a general phenomenon for Hamiltonian systems with an index-one saddle, understood by Conley in the context of celestial mechanics \cite{Co}.  The forward contracting submanifold of the periodic orbit at given energy separates trajectories that bounce from those that pass over the saddle.  The flux of energy-surface volume passing over the saddle at given energy is the action of the corresponding periodic orbit \cite{M90}. 
This ``flux over a saddle'' picture is the basis for the current subsection.
\cite{R18} validated the flux formula of~\cite{M90} on a 2DoF four-well potential energy surface, using a numerical method for computing hyperbolic periodic orbits and their forward and backward contracting submanifolds that we shall use again here. %Thus, the hyperbolic periodic orbit is the dynamical object required for computing the flux from~\cite{M90}. The numerical method in \cite{R18} can be used to find the contracting submanifolds that separate trajectories passing over the saddle. 
%give an alternate approach to computing the flux from their intersection with $v_{\pl} = 0$. 
%\rsm{Thanks!  I'd make it shorter, e.g.~(validated numerically on an example by \cite{R18})}
%\sn{Here is an attempt at adding an explanation.} \rsm{Naik, do you want to add an explanation for the citation of your paper, e.g.~experimentally confirmed in \cite{R18}?}

The field for the two-coil example used earlier involves elliptic integrals, which turned out to be tedious to deal with for the exact approach.  So we switched to a mirror field based on \cite{G+}.  After scaling the field strength to 1 at $r=0,z=\frac{\pi}{2k}$, theirs is an axisymmetric vacuum field
$$B^z = 1 - a \cos kz\, I_0(kr), \quad B^r = -a \sin kz\, I_1(kr),\quad B^\phi = 0,
$$
with $I_j$ being modified Bessel functions. A corresponding vector potential, written as a 1-form, is $A_\phi\,d\phi$, with $A_\phi = \frac{1}{2}r^2 - r\,a\, k^{-1}\,\cos(k\,z)\,I_1(k\,z)$.  We take $a \in (0,1)$ to avoid introducing nulls on the axis.  The field is periodic in $z$ but they consider one period to be a model for a mirror field.  $|B|$ has saddles at $z=\pm \pi/k$.  
%We take $k=\pi$.
%[XXX We could scale lengths so that $k=\pi$, \sn{I've used $k = 2$ by choosing $L=2$ in my computations.}] [RSM: do you mean $L=\pi$?] 

For realism one should concentrate on just a core $r<R(z)$ for some relatively small function $R$.  In particular, the field has a ring of nulls on $z=0$ at the radius $r_1$ such that $I_0(kr_1)=1/a$ and $R(0)$ should be taken less than $r_1$ (the same issue of nulls arises for the two-coil example).

To simplify treatment of transitions, we add a similar vacuum field of twice the period to break the reflection symmetry about $z=0$, so
$$
B^z = 1 - a \cos kz\, I_0(kr)- a \eta \sin \tfrac{kz}{2} I_0(\tfrac{kr}{2}), \quad B^r = -a \sin kz\, I_1(kr)+a \eta \cos \tfrac{kz}{2} I_1(\tfrac{kr}{2}),
$$
with $\eta \in (0,4)$, thereby making the upper saddle weaker than the lower one so that we can study transitions involving passing through the top alone, as for the two-coil example.  We restrict to $a< (1+\eta^2/8)^{-1}$ so as not to introduce zeroes on the axis. A potential for this field is given by $A_\phi\,d\phi$, with
\begin{align*}
    A_\phi & =\tfrac{1}{2}r^2 - r\,a\, k^{-1}\,\cos(k\,z)\,I_1(k\,z) -2\,k^{-1}\,r\,a \eta\,\sin(k\,z/2)\,I_1(k\,r/2).
\end{align*}
The saddles remain at $z=\pm \pi/k$; indeed there is still reflection symmetry about $z=\pm \pi/k$.  In particular, $\Sigma^-$ is the two planes $z = \pm \pi/k$. The parameter $k$ can be scaled to any desired value; we will take $k = 2.0$ in Figures~\ref{fig:bottle}--\ref{fig:mmmodel1958-axiasym-manifolds}.
% We will take $k=\pi$.

The saddles have $|B| = B_\pm = 1+a (1\mp \eta)$.
The saddle at $z=+\pi/k$ is the weaker one and it is the one on which we will focus attention.  It is surrounded by a family of periodic orbits of guiding-centre motion.  Indeed the plane $z=\pi/k, v_\pl=0$, is invariant and the dynamics is a drift around the axis.  Thus at energy $E > \mu B_+$ there is a periodic orbit at $z=\pi/k$ with radius such that $|B|=E/\mu$.  If also $E<\mu B_-$ then the region accessible to the guiding centre has the form of a bottle (Figure~\ref{fig:bottle}(a)).  
\begin{figure}[htbp] %  figure placement: here, top, bottom, or page
\centering
\subfigure[]{\includegraphics[width=0.45\textwidth]{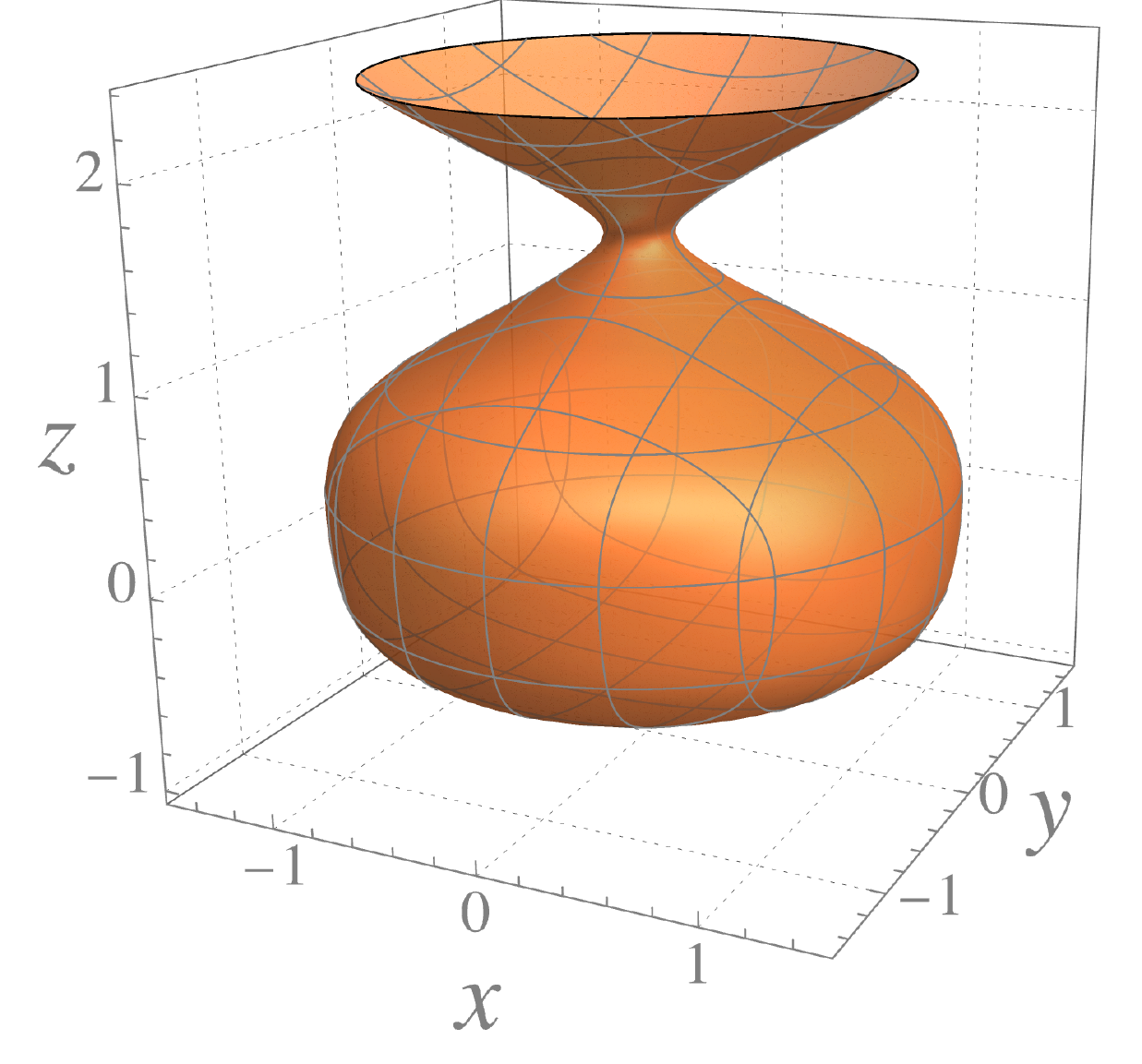}}\quad
\subfigure[]{\includegraphics[width=0.245\textwidth]{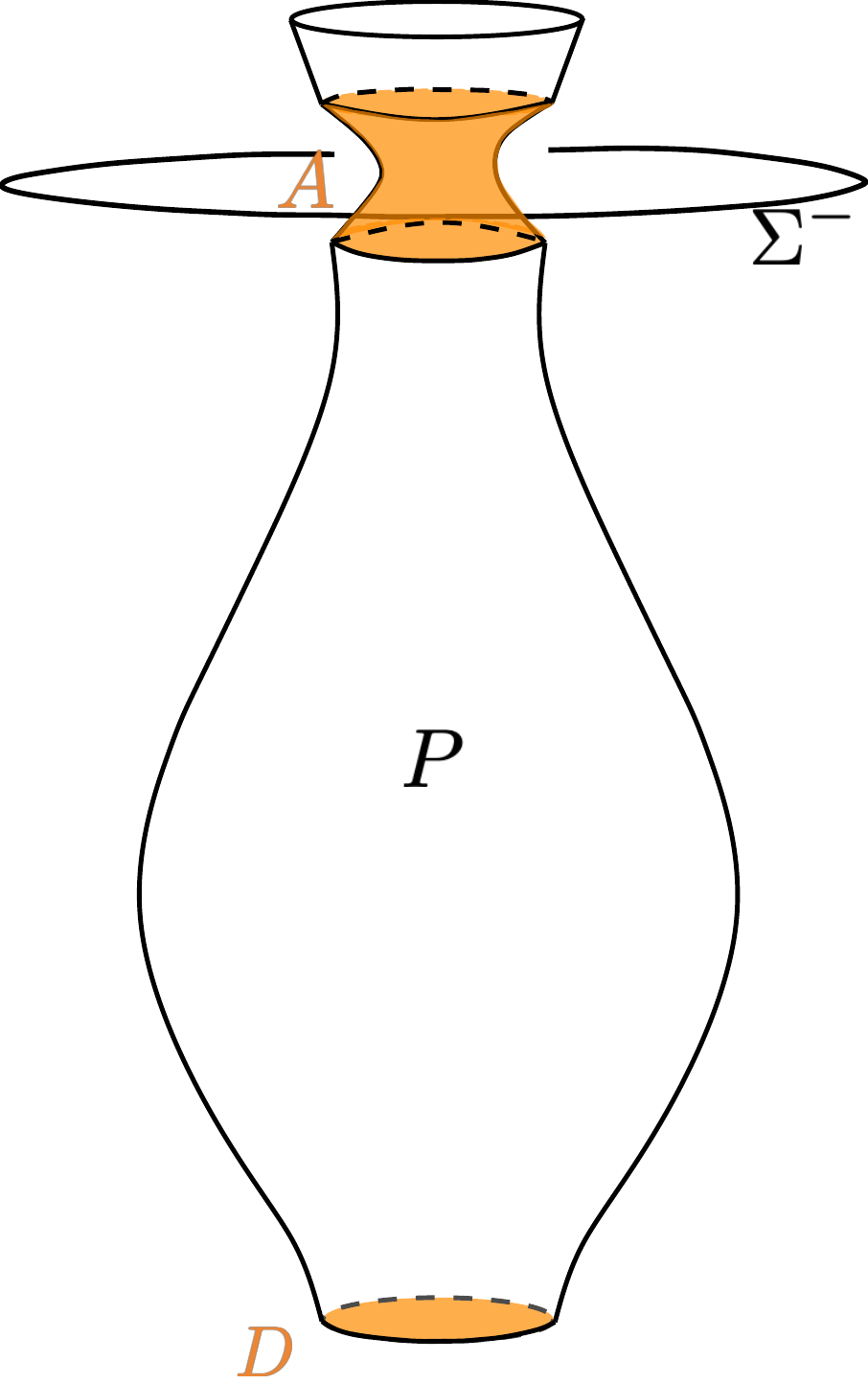}}
\caption{(a) The region accessible to guiding centres for the axisymmetric case $a=0.5, \eta=1.5, E=1.135, (B_+ = 1.125), \tilde{\mu} = 10^{-2}$;
(b) Restriction $P$ of the accessible region to $r \le R(z)$, showing also the relevant part of $\Sigma^-$. The neck $A$ and disk $D$ are pieces of the level set of $|B|$ shown on the left.}
\label{fig:bottle}
\end{figure}
% \rsm{Naik, maybe make it more axisymmetric; my drawing was not particularly good! Also can you colour the annulus A and disk D in brown to match (a)?}
This is somewhat irrelevant though, because the bottle contains the above-mentioned ring of nulls, whereas a realistic mirror machine would look like only a smaller core $r\le R(z)$ of the field.  So we should retain only the features that the level set of $|B|$ has an annular neck $A$ and a roughly horizontal disk $D$ at the bottom (Figure~\ref{fig:bottle}(b)).

Around the neck of the bottle is a periodic orbit of FGCM, given by the intersection of the bottle with $\Sigma^-$.  It is hyperbolic.
We plot it in Figure~\ref{fig:mmmodel1958-z0asym-manifolds}(a) together with its forwards and backwards contracting submanifolds $W^\pm$ in projection to physical space, up to the first bounce.  
\begin{figure}[htbp] %  figure placement: here, top, bottom, or page
\centering
\includegraphics[width=5.2in]{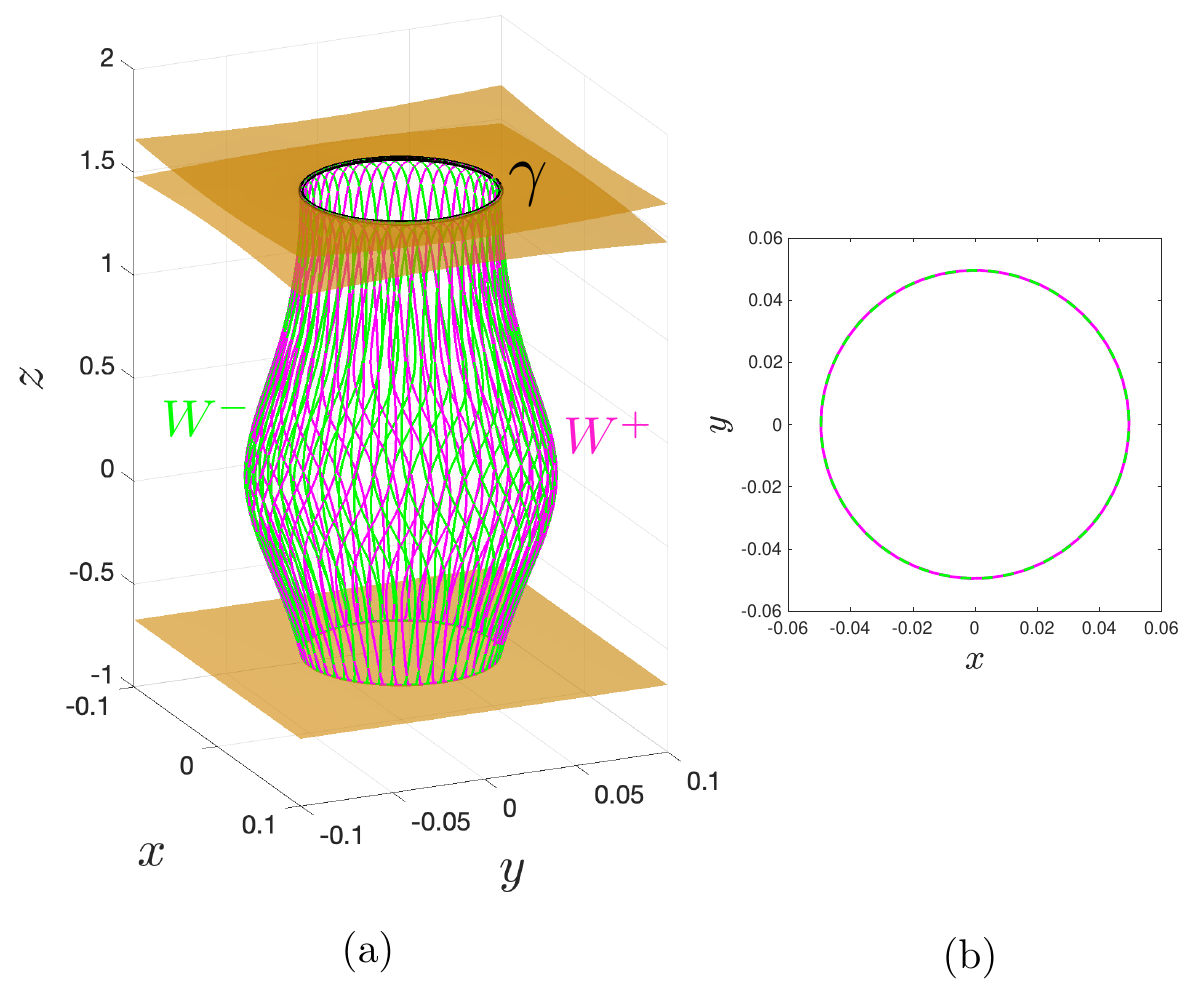}
\caption{(a) Projection to physical space of the hyperbolic periodic orbit (blue) for an axisymmetric case with $a=0.5, \eta=1.5, E = 1.126, (|B|_+ = 1.125), \tilde{\mu} = 10^{-2}$, and some trajectories on its backwards (magenta) and forwards (green) contracting manifolds up to the first bounce; (b) the traces of the first bounces for the contracting manifolds where dashed green is for the forward and continuous magenta is for the backward contracting manifolds.}
%[XXX make more visible, e.g.~make the orange surface fainter]\sn{I've changed few things: decreased opaqueness of the surface, zoomed-in and made the lines thicker}.}
\label{fig:mmmodel1958-z0asym-manifolds}
\end{figure}
$W^-$ is plotted by releasing initial conditions with the same energy slightly below  the periodic orbit and integrating forwards in time.  The system has time-reversal symmetry under simultaneous change of sign of $t, v_\pl$ and $\phi$.  Thus, $W^+$ is just the time-reverse of $W^-$.  Consequently, in this projection, $W^\pm$ coincide.  In phase space they are distinct, having opposite signs of $v_\pl$, but at the first bounce they merge and thus form a perfect separatrix.  It separates trajectories that bounce periodically in the mirror machine from those that enter downwards via the neck, make one bounce and then leave via the neck.

Then we break axisymmetry by adding (in contravariant components) 
$$B^\phi=-\eps x k \sin kz, B^z = \eps y \cos kz.$$  
This can be generated from 1-form $A_r\, dr$ with $A_r = \eps r^2 \cos\phi \cos kz$, by $B^\phi = \tfrac{1}{r}\partial_z A_r, B^r = -\tfrac{1}{r}\partial_\phi A_r$.
The hyperbolic periodic orbit has a locally unique continuation, which no longer has $v_\pl$ exactly zero, but as the average velocity in the field direction is zero, it crosses $v_\pl=0$ at least twice in a period.  Thus it forms a closed loop slightly inside the neck but touching it at at least two points.  We compute it by finding a fixed point of a return map.  Then we compute its contracting manifolds, as for the axisymmetric case.  This time they do not coincide.  In particular, their curves of first bounce do not coincide, as indicated in Figure~\ref{fig:mmmodel1958-axiasym-manifolds}.  Transitions between bouncing and escape are now possible.
\begin{figure}[htbp] %  figure placement: here, top, bottom, or page
\centering
% \subfigure[]{\includegraphics[width=3in]{mmmodel1958-axiasym-WuWs-tophypPO-set15.pdf}}\quad
% \subfigure[]{\includegraphics[width=1.8in]{mmmodel1958-axiasym-WuWs-tophypPO-intersectv0-2d-set15.pdf}}
\includegraphics[width=5.2in]{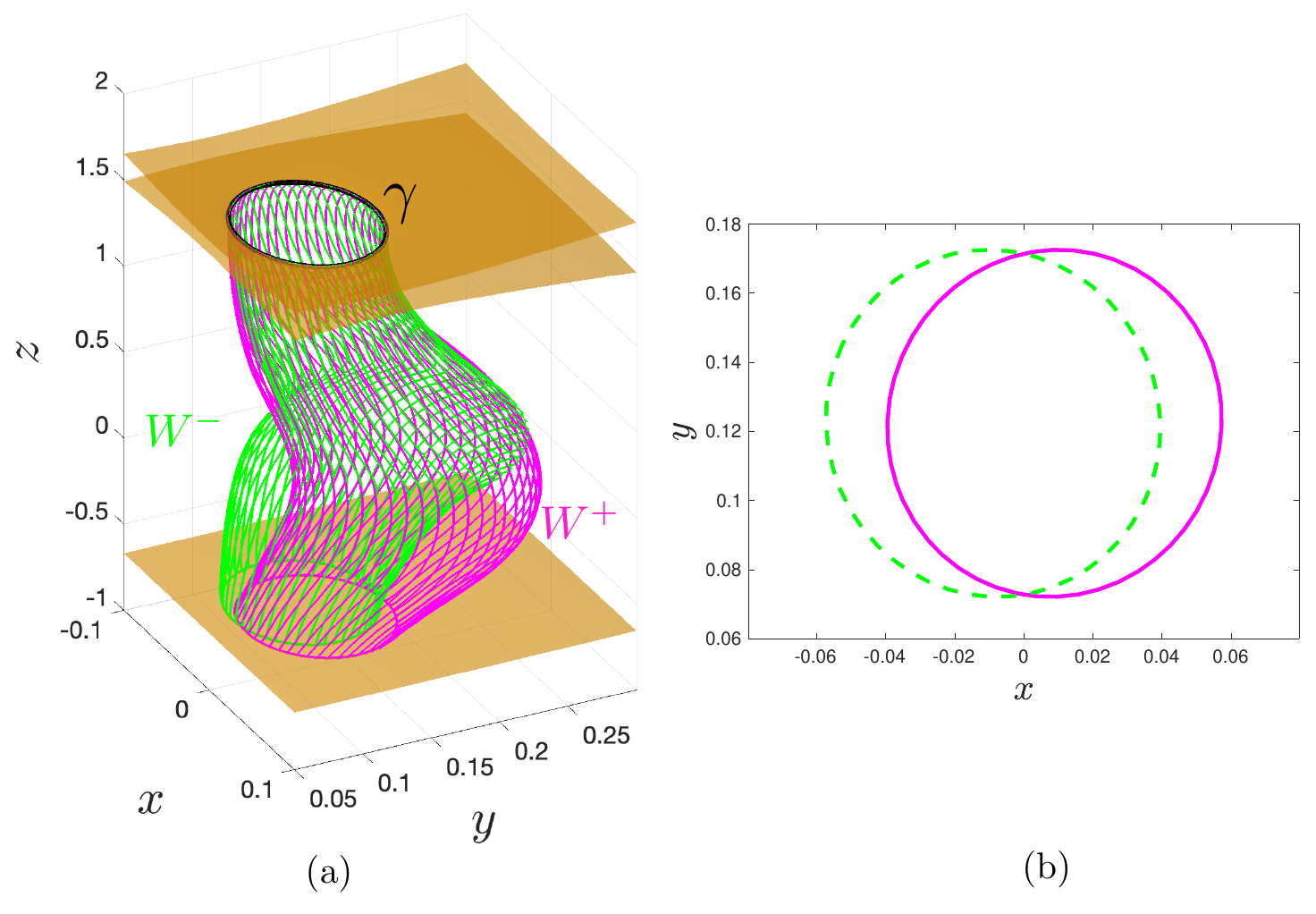}
\caption{(a) Projection to physical space of the hyperbolic periodic orbit for a non-axisymmetric case with $a=0.5, \eta=1.5, \eps=0.1, E = 1.11987, (|B|_+ = 1.11887), \tilde{\mu} = 10^{-2}$, and some trajectories on its backwards (magenta) and forwards (green) contracting manifolds up to the first bounce; (b) the traces of the first bounces for the contracting manifolds.}
\label{fig:mmmodel1958-axiasym-manifolds}
\end{figure}
% \sn{Set 14 has $B_0 = 2.0$ but as per our expressions for B, we keep $B_0 = 1.0$ fixed. So this set of parameter is an attempt at keeping values close to the axisymmetric case in Fig. 12. I need to set the aspect ratio same for both Fig. 12 and 13 same.}

In a range of energies a small amount above the saddle compared to the breaking of axi-symmetry, it is possible for $W^\pm$ to completely miss each other at the first bounce, as illustrated in Figure~\ref{fig:miss}.
\begin{figure}[htbp] %  figure placement: here, top, bottom, or page
\centering
% \subfigure[]{\includegraphics[width=2.5in]{mmmodel1958-axiasym-WuWs-tophypPO-set15b.pdf}}\quad
% \subfigure[]{\includegraphics[width=2.0in]{mmmodel1958-axiasym-WuWs-tophypPO-intersectv0-2d-set15b.pdf}}
\includegraphics[width=5.2in]{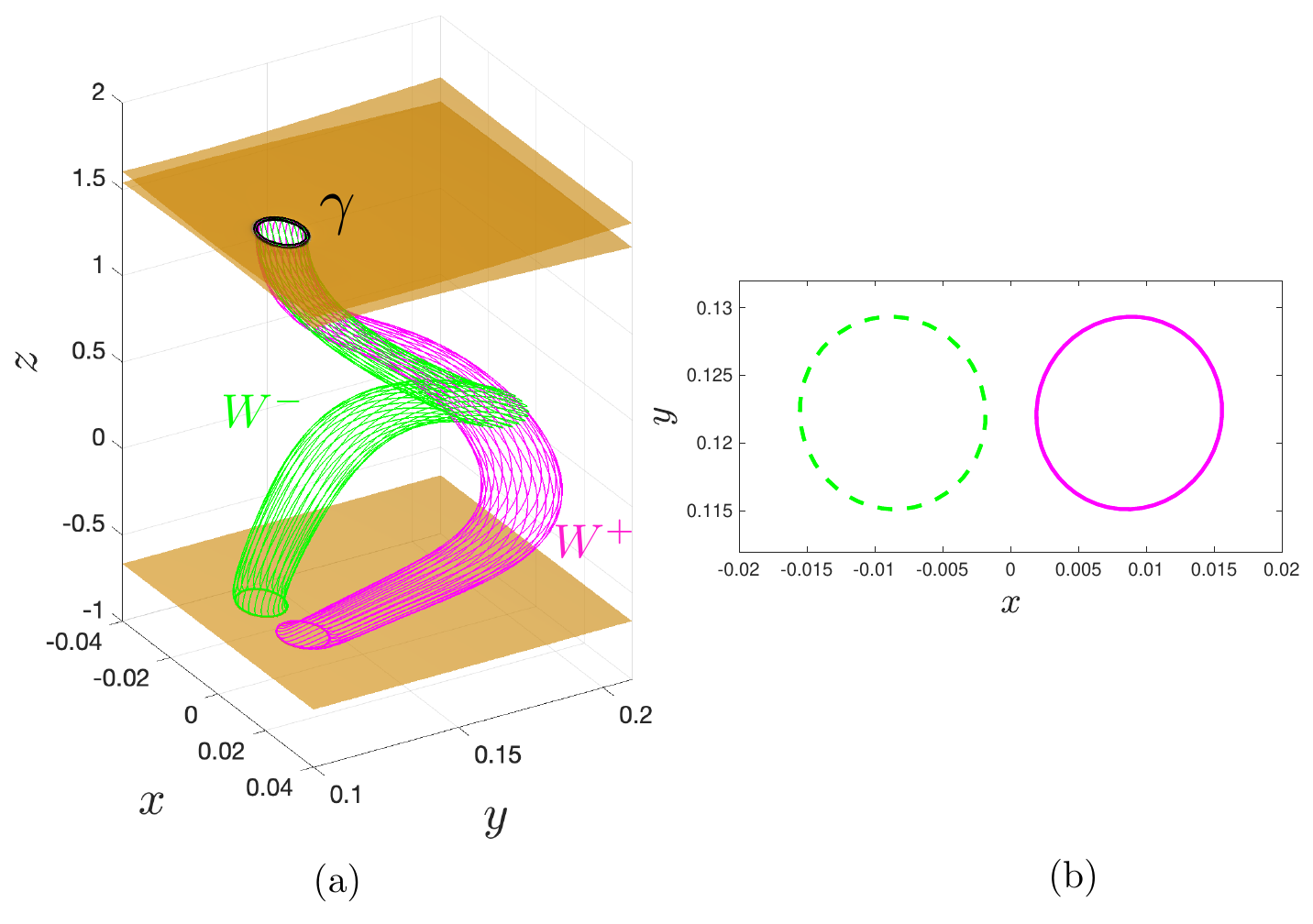}
% \rsm{insert new figures and modify the caption} 
\caption{(a) Projection to physical space of the hyperbolic periodic orbit for a non-axisymmetric case with $a=0.5, \eta=1.5, \eps=0.1, E = |B|_+ + 2 \times 10^{-5}, (|B|_+ = 1.11887), \tilde{\mu} = 10^{-2}$, and some trajectories on its backwards (magenta) and forwards (green) contracting manifolds up to the first bounce; (b) the traces of the first bounces for the contracting manifolds.}
\label{fig:miss}
\end{figure}

To understand Figures~\ref{fig:mmmodel1958-z0asym-manifolds}--\ref{fig:miss} better in phase space, it is essential to resolve the two-to-one nature of the projection of an energy level to physical space.
Let $P$ be the region $|B|\le E/\mu$ of physical space, restricted to a core of the form $r \le R(z)$ and to $z$ between the disk $D$ and a little above the neck (recall Figure~\ref{fig:bottle}(b)).
To each point of $P$ except on $D$ and $A$ there are two velocities $\pm v$ with the given energy.  On $D$ and $A$, these merge into a single value $v=0$.  In a neighbourhood of $D\times\{0\}$ the energy surface is diffeomorphic to $D\times I$, with $I$ being an interval of velocities containing $0$. This is because $\nabla |B| \ne 0$ on $D$ so we can label nearby points of the energy surface by pairs $(x,v) \in D \times I$, by flowing along $-\nabla |B|$ from $x$ to the locally unique point where $|B| = h-\frac{mv^2}{2\mu}$. 
Similarly, in a neighbourhood of $A\times \{0\}$ (again cf. Fig. \ref{fig:bottle}), the energy surface is diffeomorphic to $A\times I$. {Since $\partial_r|B|\neq 0$ on $A$, one way to explicitly realize such a chart is simply by projecting into $(\phi,z,v_\parallel)$ coordinates; the coordinates $(\phi,z)$ parameterize $A$ and $v_\parallel$ parameterizes $I$. To visualize the local stable and unstable manifolds with this chart we introduce the  variables $Z=z$, $X=(2+v_\parallel)\cos\phi$, $Y = (2+v_\parallel)\sin\phi$ for convenience and display the results in Figure \ref{fig:local_stable_unstable}. }  
\begin{figure}[htbp] %  figure placement: here, top, bottom, or page
\centering
\includegraphics[width=3in]{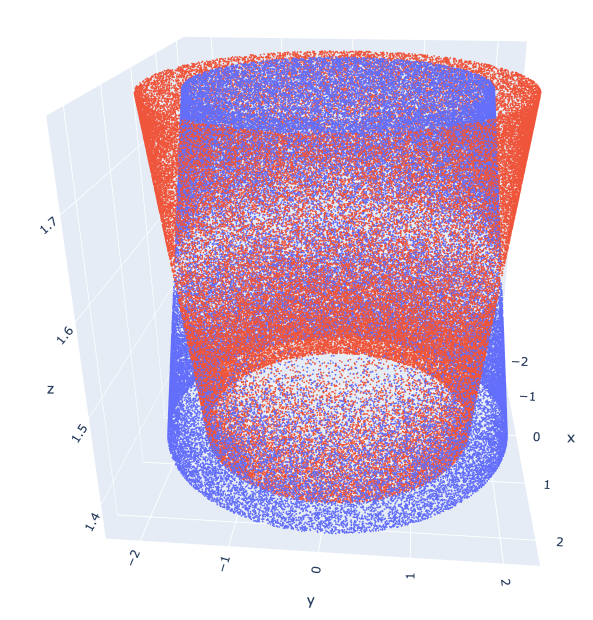}
\caption{Samples of stable (blue) and unstable (red) manifolds attached to the hyperbolic periodic orbit at energy level $E=1.061155$ in $(X,Y,Z)$ coordinates described in the text. Other parameters include $e^{-1}=0.1$, $\epsilon = 0.01$, $m=1$, $\mu = 1.1$, $a=0.5$, %$a2=.55$,
$\eta=1.1$,
$k=2$. Samples of stable manifold were generated using a rejection-sampling technique. The energy level is sampled randomly from a distribution whose spatial projection is uniform. Each sample is evolved forward in time using a well-resolved Runge-Kutta integration. If a sample stays in the window $\pi/k-0.2 < z<\pi/k+0.2$ for a fairly large time ($\Delta t = 6$ in this case) it is kept. Samples on the unstable manifold are generated in the same manner, but with backward-in-time integration.  }
\label{fig:local_stable_unstable}
\end{figure}
Away from $A\cup D$, we  take two copies of physical space, one for each sign of $v_\pl$.
The whole energy surface (restricted to $r\le R(z)$) is obtained by overlapping these charts. 
%Denote by $\pi$ the map from the energy surface to physical space given by ignoring $v$. \rsm{The previous sentence is out of place, probably because I inserted lots more after it before the place where it is used; probably should move it}

In principle, one can choose a global coordinate system for the energy surface (at least for the part projecting to $P$). For example, choose coordinates $(x,y)$ on $D\cup C \cup A$, where $C$ is the cylinder $r=R(z)$ connecting $D$ to $A$ (one can make them smooth by suitable choice of the function $R$); choose a vector field $U$ transverse to $D\cup C \cup A$ inwards and nowhere zero, so that every point of $P$ is reached uniquely by flowing along $U$ for some time $t>0$ from a point $(x,y)$; then given $v_\pl$ at this point, let $s = \sign(v_\pl) \sqrt{t}$.  Then $(x,y,s)$ are global coordinates for this part of the energy level, but there is a lot of choice and no obvious one to settle on.

One way to do something like this is to write $(x,y) = t(z)(X,Y)$, $v_\pl = s(z) W$, for some positive scaling functions $s$ and $t$ to be chosen as functions of $z$.  The origin of $(x,y)$ should be chosen to be the lowest point on the surface $|B|=E/\mu$ (or the scaling extended to include a shift).  Then for suitable $s$ and $t$, the energy surface is a graph $z = Z(X,Y,W)$.  To see how and when this can be done, consider the equation for the energy surface in the scaled variables:
$\tfrac12 s(z)^2 W^2 + |B|(t(z)X,t(z)Y,z) = E/\mu$.  The $z-$ derivative of this expression is
$$s' s W^2 + t' (X |B|_{,x}+Y |B|_{,y}) + |B|_{,z} = \tfrac{s'}{s}v_\pl^2 + \tfrac{t'}{t} (x|B|_{,x}+y|B|_{,y}) + |B|_{,z}.$$
At the bottom of the accessible region, $|B|_{,z}<0$ and the other two terms are zero.  We can make it negative on the whole of the boundary of the accessible region by choosing $\frac{t'}{t}$ sufficiently negative in the parts where $|B|_{,z}\ge 0$; this assumes $x|B|_{,x}+y|B|_{,y} >0$ away from the bottom, which is true for mild perturbations from axisymmetry.  Lastly, we can make it negative for $v_\pl \ne 0$ by choosing $\frac{s'}{s}$ sufficiently negative in places where what was constructed before is not already negative.  Then by the implicit function theorem, the energy surface is a graph $z = Z(X,Y,W)$. The function $Z$ has a minimum at $(0,0,0)$ and its other level sets are topologically two-spheres.

We did not yet implement such a choice of global coordinates, but following this idea, it is convenient in sketches to represent $v=0$ in the energy surface as a horizontal plane, with $v>0$ above the plane, $v<0$ below it.  We put $D\times \{0\}$ at the centre of this plane and $A\times \{0\}$ as a concentric annulus.  Between them is a torus $T$ representing points on $r=R(z)$ with $z$ between the two components of $|B|=h$, and outside $A$ is a cylinder $C$ representing points on $r=R(z)$ above the torus component.
This is illustrated in Figure~\ref{fig:ZZZ}.
\begin{figure}[htbp] %  figure placement: here, top, bottom, or page
   \centering
    \includegraphics[width=4in]{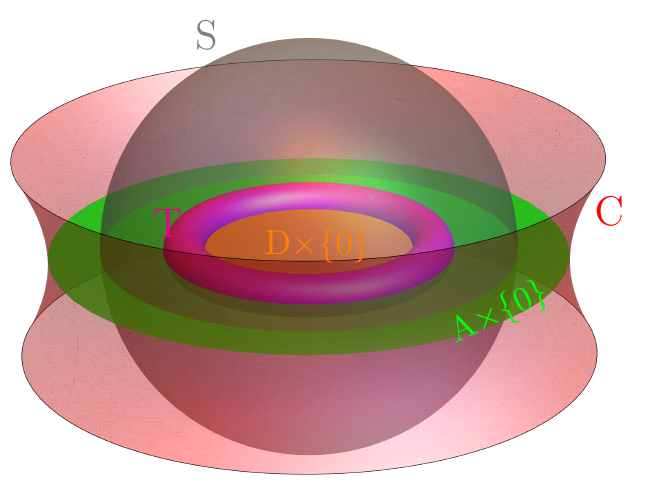}
   \caption{Double cover of the region $P$, glued along $A$ and $D$, showing the torus $T$, the cylinder $C$, the sphere $S$, the disk $D\times \{0\}$ and the annulus $A\times \{0\}$.}
   \label{fig:ZZZ}
\end{figure}
Also, denoting by $\pi$ the map from $H^{-1}(E)$ to physical space, $S= \pi^{-1}(\Sigma^-)$ is a double cover of the disk $\Sigma^- \cap P$ glued along its boundary $\Sigma^- \cap A$; so it is a sphere.  A neighbourhood of $S$ can be obtained in the form $S\times I$ for an interval $I$ representing height $z$ relative to $\Sigma^-$; this is known as the Conley-McGehee representation (see for example the sketch in \cite{M90}, and \cite{KW}).

As a first pass, 
the region of the energy surface between $S$ and $T$ might be considered to be the states that are inside the machine.  The upper hemisphere has $v>0$ so one might say it consists of states that are just exiting the machine; but this is not quite right because guiding-centre drifts could compensate for small $v$, so we will give a dynamical construction shortly.  Similarly, the lower hemisphere has $v<0$ so consists of states that are just entering the machine, modulo the drift corrections. On $T$, states may enter or exit, but this is an aspect that we do not address here, as our focus is on transitions between different classes of guiding-centre motion.

Now we come to the dynamics in the energy surface.  FGCM has a periodic orbit $\gamma$ which is the continuation of the circle $\pi^{-1}(\Sigma^- \cap A)$ of equilibria for ZGCM.  It is the Lyapunov orbit with energy $E$ for the upper saddle.  In general it does not have $v$ identically zero, but $v$ oscillates about $0$ on it.  It is hyperbolic,
%Given $\mu$, FGCM at energy $E$ has two periodic orbits $\gamma^\pm$. They correspond to perturbation by $\mu$ of the curves $|B|=h$ on $\Sigma^\pm$, i.e.~the intersections of the amphora with $\Sigma^\pm$ as shown in Figure~\ref{fig:M1}(a).  They have small $v$, oscillating about zero, and except where $v=0$ they lie slightly inside the amphora.  Their direction of motion (precession from the drift equations) is the same.
%Although in physical space $\gamma^-$ is a smaller periodic orbit than $\gamma^+$, in our $(X,Y,Z)$ coordinates it is larger.  So we sketch them as in Figure~\ref{fig:M3}.  
%Because $|B|''>0$ on $\Sigma^+$, $\gamma^+$ is elliptic so nearby points perform helical motion around it, corresponding to bouncing.  Because $|B|''<0$ on $\Sigma^-$, $\gamma^-$ is hyperbolic.  
so it has forward and backward contracting manifolds $W^\pm$.  It is possible to span $\gamma$ by a surface diffeomorphic to a sphere that is transverse to the dynamical vector field except on $\gamma$.  It is a perturbation of the sphere $S$, so we denote it by the same symbol.  It is non-unique but an essential feature is that it lies in the sectors between $W^\pm$ indicated in Figure~\ref{fig:M3}.
\begin{figure}[htbp] %  figure placement: here, top, bottom, or page
    \centering
    \includegraphics[width=3.5in]{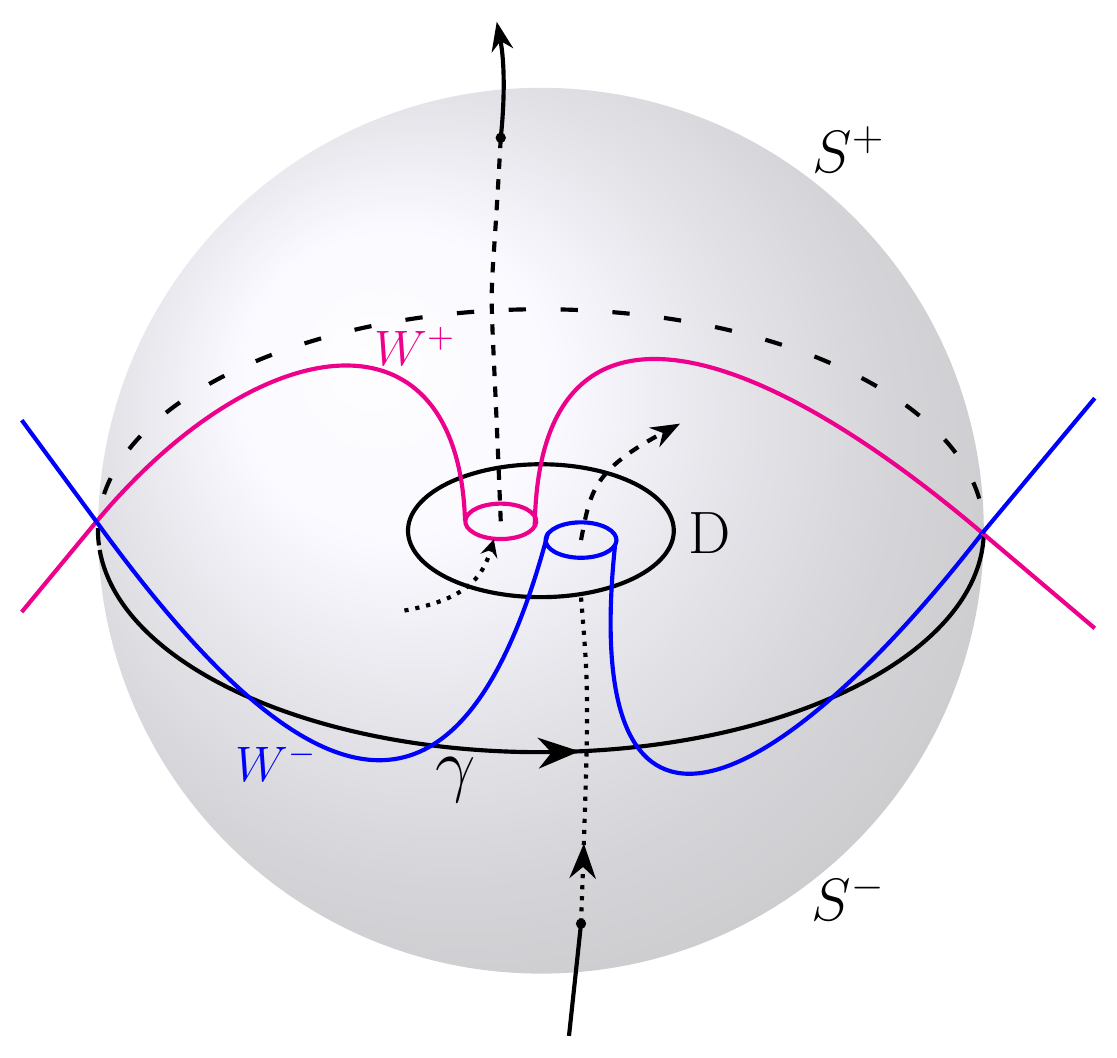} 
    \caption{Dynamics in the energy surface $H^{-1}(E)$, showing the periodic orbit $\gamma$, the contracting manifolds $W^\pm$ and the transverse hemispheres $S^\pm$, in the case when the first intersections of $W^\pm$ with $v=0$ miss each other. Inside the sphere S, we show the trajectories below the horizontal as dotted and above as dashed while escaping out of the sphere is shown as continuous.}
   \label{fig:M3}
\end{figure}
It is called a dividing surface, because the upper hemisphere $S^+$ has unidirectional flux from inside to outside, representing guiding centres leaving the machine, and the lower hemisphere $S^-$ has unidirectional flux from outside to inside, representing guiding centres entering the machine.
The manifolds $W^\pm$ form tubes that in one direction go inside the sphere.  They necessarily cross $v=0$ transversely, representing bounce on the disk $D$.  Their first intersections with $v=0$ are diffeomorphic to two circles.  In the axisymmetric case, they coincide, but if axisymmetry is broken then they need not coincide.  Generically for $h$ only slightly above $h_0$ they miss each other entirely, as shown in Figure~\ref{fig:M3}.  This will be justified at the end of the subsection.
In this case, we see that all the flux entering the sphere transitions to bouncing trajectories (in fact, making at least three bounces) and all the flux leaving the sphere came from bouncing trajectories (making at least three bounces).  We call them ``trapping'' and ``detrapping'' fluxes, respectively, to align with standard terminology.  The fluxes of energy-surface volume are equal and can be expressed as the action integral of $\gamma$:
$$S = \int_{\gamma} \alpha, \quad \alpha = eA^\flat + mv b^\flat,$$
or $S = \int_{\gamma} (eA + mv b)\cdot dx$, where $A$ is any vector potential for $B$.

\begin{figure}[htbp] %  figure placement: here, top, bottom, or page
\centering
\includegraphics[width=3.5in]{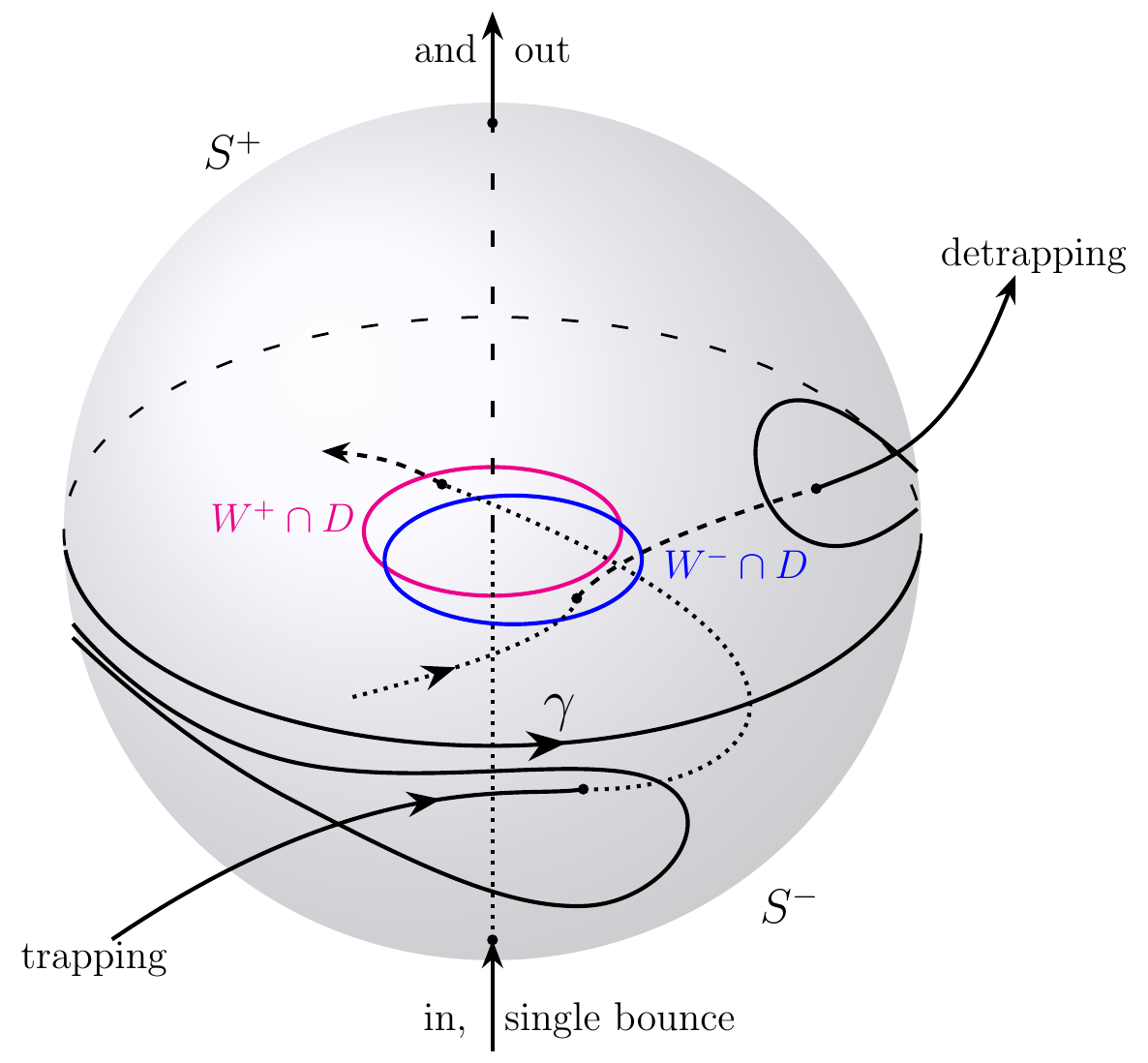}
\caption{Case when the first intersections of $W^\pm (\gamma)$ with the bounce surface $\{v=0\}$ intersect. $W^\pm$ not drawn, except for their first intersections with $D$, to give priority to the lobes on the dividing surface. Same convention for the trajectories as in Fig.~\ref{fig:M3}.} 
\label{fig:M4}
\end{figure}
% \sn{Linestyle convention for trapping and detrapping: inside the sphere below horizontal, use dotted, inside the sphere above the horizontal, use dashed, outside the sphere, use solid. Southern hemisphere lobe make the convergence to the equator look the same.}
For larger $h$, the two circles of first intersection of $W^\pm$ with the bounce surface $\{v=0\}$ might intersect, as in Figure~\ref{fig:M4}.  We have drawn the simplest case of two intersections, but of course there could be more.  The trajectories of the intersections are homoclinic to $\gamma$.  They wrap onto $\gamma$ in both directions in time.  
%Rather than drawing them, we draw the contracting manifolds of the points of $\gamma^-$ to whose trajectories the trajectory of the intersection converges as $t \to \pm \infty$.  They are in general distinct points of $\gamma^-$, because there can be a phase shift between the homoclinic trajectory and $\gamma^-$.
Now there are three types of flux.  Firstly there is the trapping flux across a lobe on the upper hemisphere that crosses $v=0$ in the indicated lobe and turns into a bouncing trajectory (at least three bounces).  Secondly, there is the detrapping flux that comes from bouncing trajectories that cross $v=0$ in the other indicated lobe and exit the lower hemisphere via the indicated lobe.  Thirdly, there is the remaining flux across the upper hemisphere, that passes through the intersection of the disks bounded by the circles on $v=0$ (where they perform one bounce) and then exit via the rest of the lower hemisphere.  These are single-bounce trajectories.

The fluxes of energy-surface volume are related to the actions of the homoclinic orbits and of $\gamma$.  Namely, the trapping flux is the difference in action between the homoclinic orbits at the ends of the corresponding lobe.  The action of a homoclinic orbit does not converge, but the difference in actions of two homoclinic orbits to the same periodic one does, as long as one takes the end points to converge together, and that is what is assumed in this statement.  The detrapping flux is equal to the trapping flux because they are both given by the difference in action of the same pair of homoclinic orbits.  The single-bounce flux is the difference between the action of $\gamma$ and the trapping flux.

To complete this discussion, we explain why generically the circles miss each other for $h$ only slightly above $h_0$.  In the limiting case $h=h_0$, the periodic orbit $\gamma$ shrinks to an equilibrium point, namely the saddle of $|B|$ near the centre of the top coil.  The accessible region can be considered as an inside and an outside that are pinched together in a conical point at the saddle.  The part of the energy surface corresponding to the inner part is a double cover of a sphere with a conical singularity at the saddle, hence it is a 3-sphere with a conical singularity. The ``shrinking'' of $\gamma$ to this conical point corresponds in our representation to $\gamma$ growing to a circle at infinity:~think of how a small circle around the north pole maps under stereographic projection to a large circle in the plane tangent to the south pole.  
The saddle has 1D contracting manifolds.  In one direction they pierce $v=0$.  Generically they pierce it in distinct points.  Now deform the picture to $h>h_0$:~the points expand to small circles, but still miss each other for small $h-h_0$.
%---moved to text-backup.tex---%

\subsection{Tokamak}
For an axisymmetric tokamak, $\Sigma \times \{v=0\}$ also persists to an invariant submanifold for $\mu>0$, but only after taking its union with another invariant submanifold that crosses it transversely, corresponding to closed guiding-centre trajectories, and allowing the crossing to break generically.  The result was already illustrated in \cite{M94}. It is illuminating to analyse this example in some detail. {This analysis will demonstrate how and why strong isodrasticity generically breaks in magnetic configurations close to tokamaks.}

We use the magnetic field of Section~\ref{sec:tok}.  For the guiding-centre motion we choose to scale $e,m$ to $1$ and take $H = \tfrac12 v^2 + \mu |B|$, $\omega = \beta + d(vb^\flat)$. This is so that the effect of turning on small $\mu$ is on the Hamiltonian instead of the symplectic form, which makes it easier to understand.  We can take an axisymmetric vector potential $A$ for $B$, leading to $\beta = dA^\flat$.  Its (physical) component $A_\phi = \psi/R = \tfrac12 (r^2+z^2)/R$,
where $r=R-1$ (and in case it is useful, one can take $A_z=-C\log R, A_R = 0$).  By axisymmetry, 
\begin{equation}
p_\phi=R (A_\phi + vb_\phi) = \psi + vC/|B| \label{eqn:pphi_axisym}
\end{equation}
is conserved.  So on $p_\phi=p$,
\begin{equation}v = \frac{\sqrt{C^2+2\psi}}{RC} \tilde{p}
\label{eq:vfromp}
\end{equation}
where $$\tilde{p}=p-\psi.$$
Thus we obtain reduced Hamiltonian $H_p$ on $p_\phi=p$ (modulo $\phi$),
\begin{equation}
H_p(r,z) = \frac12 \frac{C^2+2\psi}{C^2R^2}(p-\psi)^2 + \frac{\mu}{R}\sqrt{C^2+2\psi}.
\label{eq:Hp}
\end{equation}
A typical contour plot of $H_p$ for $p>0$ and $\mu$ small is shown in Figure~\ref{fig:H_p}, which shows some banana trajectories and some circulating and counter-circulating trajectories, and makes clear it has three critical points.
\begin{figure}[htbp] %  figure placement: here, top, bottom, or page
 \centering
\includegraphics[width=3in]{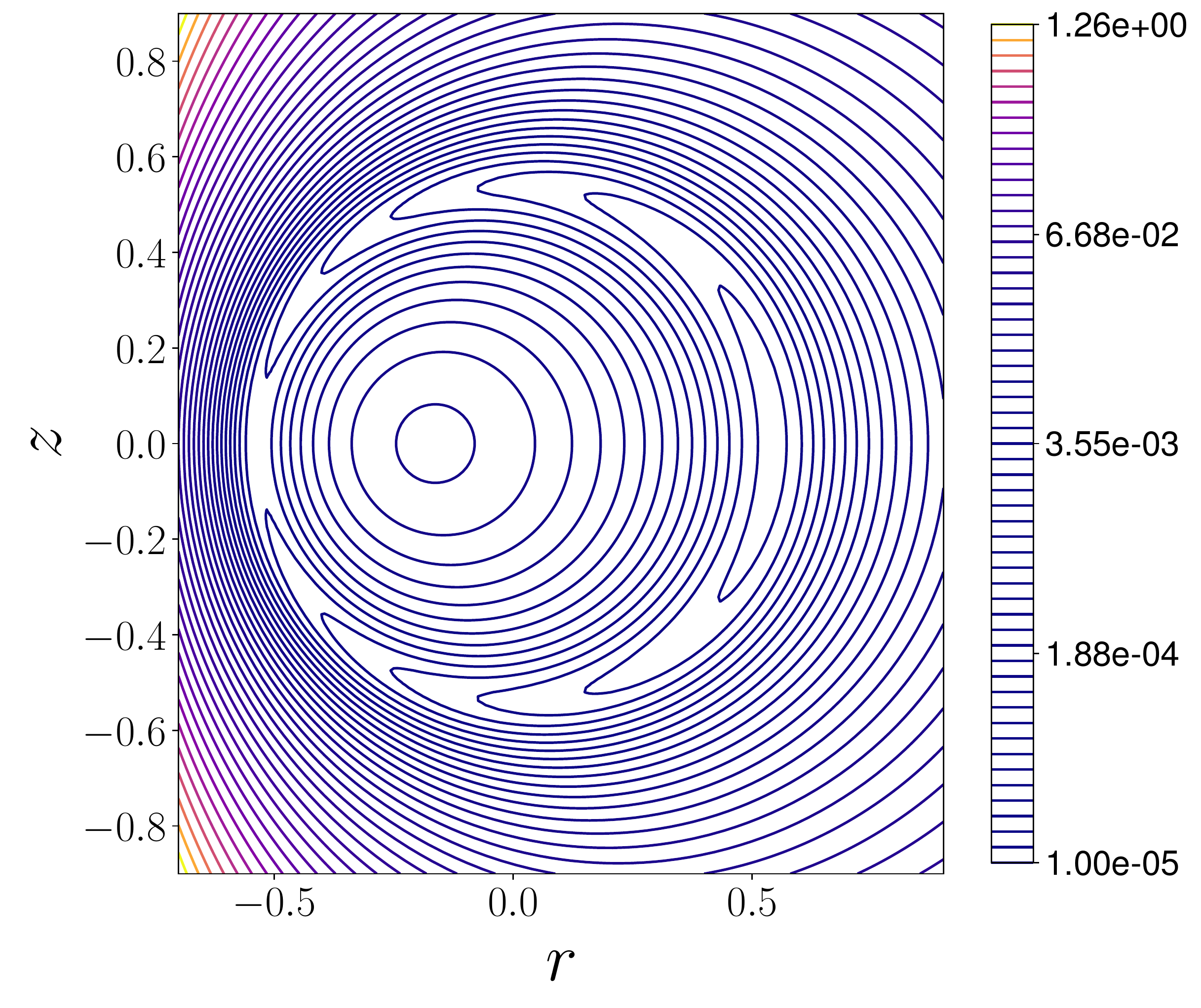}
\caption{Contours of $H_p$ in $(r,z)$ for $C=5.5, \mu=0.0002, p=0.145$.}
 \label{fig:H_p}
\end{figure}

We look for critical points of $H_p$ because they generate axisymmetric periodic orbits for $H$ (or exceptionally, a circle of equilibria). The ones of interest for isodrasticity are the hyperbolic ones, but for the moment we consider them all.  Using $\psi = \tfrac12 (r^2+z^2)$ and $R=1+r$, we see the $z$-derivative is zero iff $z=0$.  The $r$-derivative on $z=0$ is
$$\frac{\partial H_p}{\partial r}(r,0) = -\frac{C^2-r}{C^2R^3}\tilde{p}^2 - \frac{C^2+r^2}{C^2R^2}r\tilde{p} - \mu \frac{C^2-r}{R^2\sqrt{C^2+r^2}}.$$ % this expression checks out, details to be added in the calculations document.
So critical points of $H_p$ are the solutions $r$ of
\begin{equation}
(C^2-r)\tilde{p}^2+Rr(C^2+r^2)\tilde{p} + \mu\frac{RC^2(C^2-r)}{\sqrt{C^2+r^2}} = 0,
\label{eq:ptilde}
\end{equation}
where $\tilde{p} = p-r^2/2$.
We can consider this instead as an equation for $\tilde{p}$ given $r$.  It has real roots iff 
$$Rr^2(C^2+r^2)^{5/2} \ge 4\mu C^2(C^2-r)^2,$$
that is for $r$ outside an interval approximately $(-2\sqrt{\mu C}, 2\sqrt{\mu C})$.
The value of $\tilde{p}$ at the fold points is $\tilde{p} = -\frac{Rr(C^2+r^2)}{2(C^2-r))} \approx -r/2$.
For $r \ll 1$ but large compared with $\sqrt{\mu C}$, the dominant balances are (i) the second and third terms, giving a solution approximately $\tilde{p} = - \frac{\mu C}{r}$,
and (ii) the first and second terms, giving a solution approximately $\tilde{p}=-r$.  
These can be converted from $\tilde{p}$ to $v$ by (\ref{eq:vfromp}), which gives $v \approx \tilde{p}$ and so the same approximate formulae. 

Thus we obtain an invariant manifold $N$ consisting of circular orbits in $z=0$, as shown in Figure~\ref{fig:N4tok}.
The parts with $v \approx -\frac{\mu C}{r}$ are a small perturbation of $\Sigma \times \{v=0\}$ away from the gap.  The parts with $v \approx -r$ correspond to perturbation of another invariant submanifold for $\mu=0$, consisting of the circular periodic orbits in $z=0$ for which the vertical components of the curvature drift and parallel velocity balance (this can be computed exactly if desired, but has the leading form $v = -r$).

\begin{figure}[htbp] %  figure placement: here, top, bottom, or page
\centering
\includegraphics[width=4in]{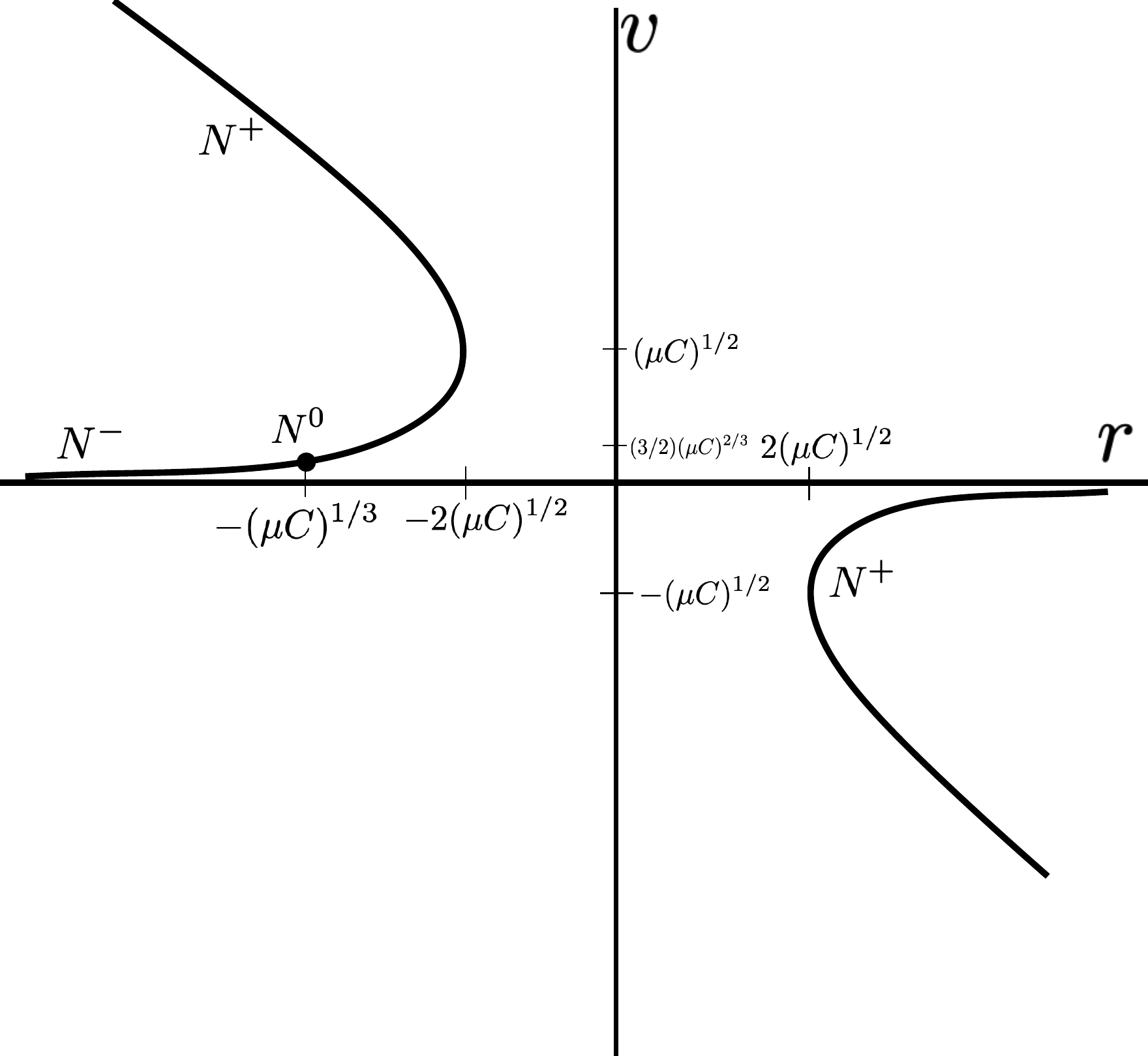} 
\caption{Sketch of the invariant manifold $N$ for FGCM in the tokamak example, showing its decomposition into $N^\pm$ and $N^0$. 
%{\color{red}sequence of breaks in $N^{+}$?}
}
\label{fig:N4tok}
\end{figure}

To identify which parts are normally hyperbolic or elliptic, we have to do further work. The boundary between them is determined by the turning points of $p_\phi$ along this curve of critical points.  There is just one such turning point and it is near $r=-(\mu C)^{1/3}$. 
To see this, $\delta p = 0$ iff $\delta\tilde{p}=-r\,\delta r$,  divide (\ref{eq:ptilde}) by $C^2-r$ for convenience and differentiate the result to obtain the condition for $\delta p = 0$, $\delta r \ne 0$:
\begin{equation}
\tilde{p} \left(-2r + \left(\frac{Rr(C^2+r^2)}{C^2-r}\right) ' \right) -\frac{Rr^2(C^2+r^2)}{C^2-r} + \mu C^2 \left(\frac{R}{\sqrt{C^2+r^2}}\right) ' = 0, 
\label{eq:degen}
\end{equation}
where the derivatives are with respect to $r$.  For $\mu$ and $r$ small this is $\tilde{p} -r^2 + \mu C \approx 0$.  This occurs on the lower left branch, where $\tilde{p} \approx -\mu C/r, r<0$, as a balance principally between the first and second terms.  Hence the result.  This gives $N^0$.  
Study of the type of the critical points of $H_p$ gives the decomposition of the complement into $N^\pm$ as shown in the figure. One can compute $N^0$ as a curve in ($\mu,E,p)$ parametrised by $r$.
Note that at $N^0$, $r \approx - (\mu C)^{1/3}, \tilde{p} \approx (\mu C)^{2/3}$, so $p \approx \frac32 (\mu C)^{2/3}$.  Thus, the regime for which $H_p$ has three critical points is approximately $p> \frac32 (\mu C)^{2/3}$.

%OLD AND NOT USEFUL
%The condition for a double root is equality in the above, so
%$$\mu = \frac{(C^2+r^2)^{5/2}Rr^2}{4C^2(C^2-r)^2}.$$
%The corresponding value of $p$ is
%$$p = \frac12 r^2 - \frac{(C^2+r^2)Rr}{2(C^2-r)}.$$
%[XXX but maybe this is not same as double root for $r$???]
%Making a parametric plot for $r\le 0$ (which is the relevant range XXX JUSTIFY) produces Figure~\ref{fig:pmu} for the division of the $(\mu,p)$ plane into regions for which $H_p$ has 1 or 3 critical points.
%\begin{figure}[htbp] %  figure placement: here, top, bottom, or page
% \centering
%\includegraphics[width=3in]{regions}
%\caption{Curve in the plane of $(\mu C, p)$ for $H_p$ to have a degenerate critical point; plotted for $C=5.5$, but the picture varies little with $C$ under the scaling of $\mu$ used here.  To the left of the curve, $H_p$ has three critical points; to the right and for $p<0$ it has just one.}
% \label{fig:pmu}
%\end{figure}

%Corresponding to the double root for $\tilde{p}$ is a fold in the set of critical points $r$ as a function of $p$ and $\mu$.  It is near $r = -2\sqrt{\mu C}$, $\tilde{p}=\sqrt{\mu C}$, hence near $v = \sqrt{\mu C}$.

%The energy of the circular periodic orbits can be obtained from (\ref{eq:Hp}). To leading order, those with small $v$ have energy $\mu C (1-r+\frac{\mu C}{2r^2})$ and the other branch has energy $\mu C + \tfrac12 r^2$.  

%Thus we obtain an invariant manifold $N$ for FGCM.  It can be divided into normally hyperbolic, elliptic and degenerate parts $N^-,N^+, N^0$ as in Figure~\ref{fig:N4tok}.  

This invariant manifold $N$ persists under breaking axisymmetry, because the degenerate case $N^0$ is an elementary saddle-centre periodic orbit. This is because the zero of equation (\ref{eq:degen}) is transverse.

Having established the existence of the invariant submanifold $N^-$, we now wish to understand its contracting submanifolds.  In the axisymmetric case, $p_\phi$ is conserved, so they form pairs of homoclinic connections, as can be divined from Figure~\ref{fig:H_p}.
To consider the effect of breaking axisymmetry, we have to abandon conservation of $p_\phi$, but energy is still conserved.  So the useful viewpoint is to consider the energy surfaces $H^{-1}(E)$ and the level sets of $p_\phi$ in them for the axisymmetric case and then consider the effects of losing conservation of $p_\phi$.  This is a view already promoted in \cite{M94} (though some aspects of the figures there are inaccurate).

Firstly, as in \cite{M94}, one can represent points of $H^{-1}(E)$ in the axisymmetric case by $(v,z)$ because given $(v,z,E)$ there is at most one compatible $r$ in $r^2+z^2 \le r_0^2$.  
To see this, 
\begin{equation}
\tfrac12 v^2 = E - \frac{\mu}{R}\sqrt{C^2+(R-1)^2+z^2}, 
\label{eq:v2}
\end{equation}
so its derivative 
$$\frac{\partial (v^2/2)}{\partial R} = \frac{\mu(C^2+1+z^2-R)}{R^2 \sqrt{C^2+(R-1)^2+z^2}} $$
is positive for $R \in (0,C^2+1+z^2)$, which contains $r^2+z^2\le r_0^2 < 1$.
Thus there can not be two values of $R$ with the same value of $(v,z,E)$.

%Because it is $v^2$ that enters (\ref{eq:v2}), for $R$ such that $v^2 > 0$ (given $(z,E)$) there is a pair $\pm v$ of values of $v$.  There is the single value $v=0$ where $(E^2/\mu^2-1)R^2 = C^2+1+z^2 - 2R$.  This forms a hyperbola, parabola, ellipse, single point or the empty set, according as the value of $E/\mu > 1, =1, \in(\frac{C}{\sqrt{C^2+1}},1), = \frac{C}{\sqrt{C^2+1}}, < \frac{C}{\sqrt{C^2+1}}$, respectively.

%Thus joint level sets of $(E,p_\phi)$ in $(r,z)$ as in Figure~\ref{fig:H_p} are mapped to???

We can obtain $R$ explicitly in terms of $(v,z,E)$ because the equation for $R$ is a quadratic:
\begin{equation}
R^2(E-v^2/2)^2=\mu^2(C^2+(R-1)^2+z^2).
\label{eq:R}
\end{equation}
To save studying more than one case for the sign of the coefficient of $R^2$, let us suppose that $C > \sqrt{3}$.  Then $|B| > 1$ for $C^2+(R-1)^2+z^2 > R^2$, i.e.~for $R < \tfrac12 (C^2+1+z^2)$, which is true for all $r^2+z^2 \le r_0^2$.  So $E-\tfrac12 v^2 >\mu$.  It follows that the quadratic has one positive real root $R$.  Writing $R=1+r$ it has
$$ r = \frac{-(h-u^2/2)^2+ \sqrt{1+((h-u^2/2)^2-1)(C^2+1+z^2)}}{(h-u^2/2)^2-1},$$
%$$R = \frac{-1+\sqrt{1+((h-u^2/2)^2-1)(C^2+1+z^2)}}{(h-u^2/2)^2-1},$$
where $h=E/\mu$ and $u = v/\sqrt{\mu}$.

Then $p_\phi$ can be expressed in terms of $E, \mu, v$ and $z$, by rewriting Eqn.~\eqref{eqn:pphi_axisym} using $|B| = \frac{E}{\mu} - \frac{v^2}{2 \mu} = h - u^2/2$, and substituting the above expression for $r$ into
$$p_\phi = \tfrac12(r^2+z^2)+\frac{\sqrt{\mu} u C}{h-u^2/2}.$$
A contour plot of $p_\phi$ in $(u,z)$ for some $\mu,C,E=h\mu$ is given in Figure~\ref{fig:p_phi}.
\begin{figure}[htbp] %  figure placement: here, top, bottom, or page
\centering
\includegraphics[width=3in]{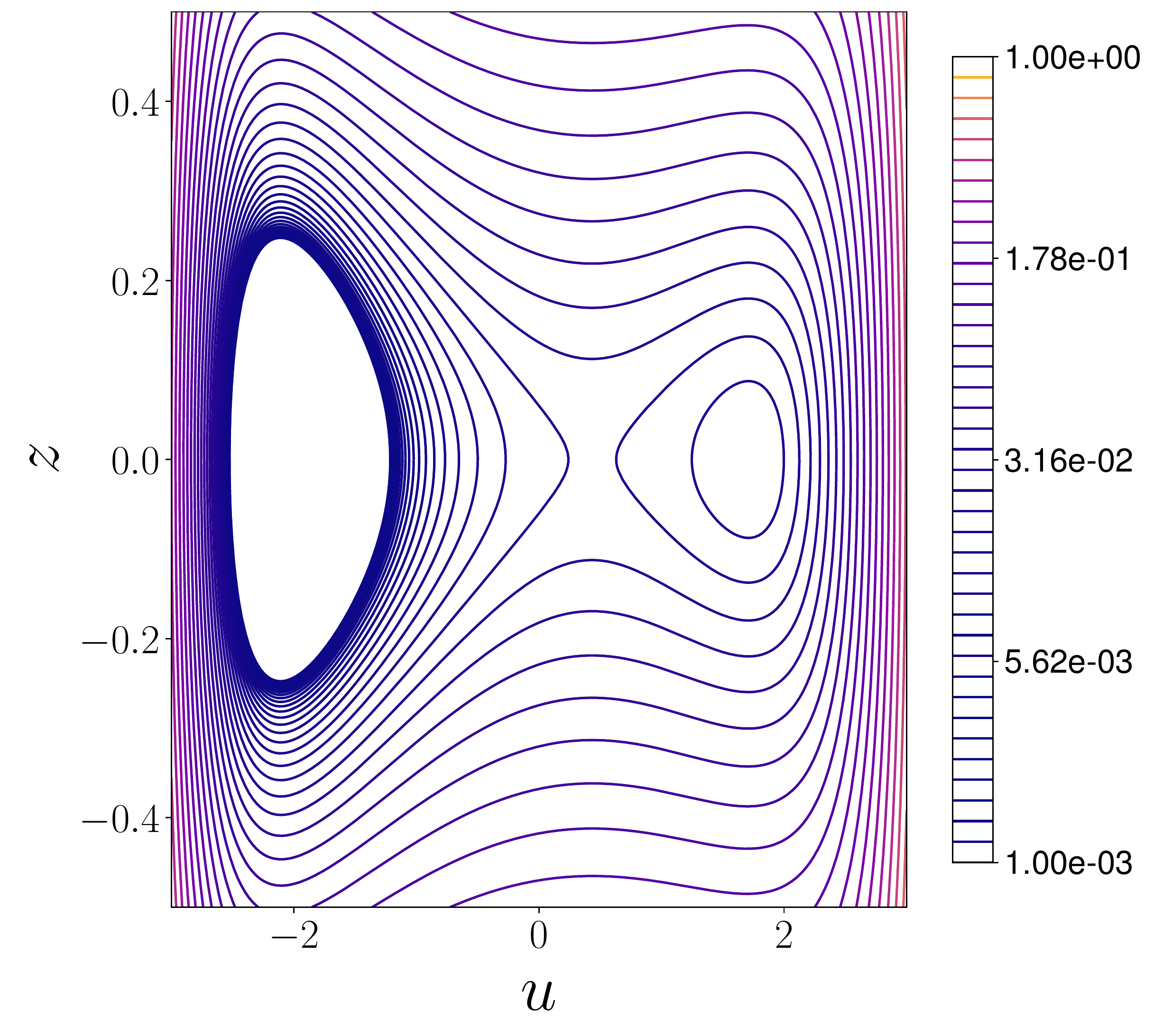} 
\caption{Some contours of $p_\phi$ in $(u,z)$ for $C=5.5, \mu=0.0002, h=7.4$.}
\label{fig:p_phi}
\end{figure}
Note that the previous considerations restrict us to $C>\sqrt{3}$, $h>1$ and $|u| < \sqrt{2(h-1)}$.

The region in $(\mu,E)$ for which there are three critical points of $p_\phi$ can be obtained from the previous analysis.
The energies for particles near the magnetic axis and with small $v$, as on the branches $v \approx -\frac{\mu C}{r}$, are near $\mu C$, so it is natural to consider a further scaled version $\cE = \frac{h}{C} - 1$.
The scaled energy at the degenerate critical point $N^0$ is $\cE \approx \frac32 (\mu C)^{1/3}$.
Degenerate critical points of $p_\phi$ for given $\cE$ correspond to degenerate critical points of $H$ for given $p$, so we obtain that the region of energy for which $p_\phi$ has three critical points is approximately $\cE > \tfrac32 (\mu C)^{1/3}$.

%CUT
%so the curve in $(\mu,H)$ with a degenerate critical point of $p_\phi$ is plotted in Figure~\ref{fig:Hmu}.
%\begin{figure}[htbp] %  figure placement: here, top, bottom, or page
%\centering
%\includegraphics[width=3in]{regionH} 
%\caption{Curve for a degenerate critical point of $p_\phi$ in $(\mu C, H)$ for $C=5.5$. $p_\phi$ has 3 critical points to the left of the curve, just one to the right.}
%\label{fig:Hmu}
%\end{figure}

In the regime of $\cE$ for three critical points of $p_\phi$, one is a hyperbolic point (corresponding to that for $H_p$).  It has a positive value of $v$ (it is born near $v=(\mu C)^{2/3}$ at $\cE \approx \tfrac32 (\mu C)^{1/3}$ and for larger $\cE$ has approximately $v = \mu C/\cE$). The rate of change of ${\phi}$ is given by
$$\tilde{B}_\pl R \dot{\phi} = v\frac{C}{R} + v^2 \frac{2\psi}{(C^2+2\psi)^{3/2}} + \mu \frac{2\psi-rC^2}{R^2(C^2+2\psi)},$$
which at the hyperbolic point has each term positive (remember $v>0$ and $r<0$), so $\dot{\phi}$ is positive there.
Thus in the full phase space including $\phi$, it corresponds to a hyperbolic periodic-orbit of guiding-centre motion.  The hyperbolic point has two homoclinic orbits, corresponding to the joint level set of $H$ and $p_\phi$ containing the hyperbolic point.  The mapping (\ref{eq:vfromp}), considered as giving $v$ from $(r,z)$, shows that the joint level set forms a figure of eight in $(v,z)$, arranged as in Figure~\ref{fig:p_phi}.  The closed levels of $p_\phi$ for given $E$ inside the right and left lobes of the eight correspond to positively and negatively circulating trajectories, respectively (actually, some of the negatively circulating trajectories bounce, as can be seen in the Figure from the sign change of $u$, but it is convenient to consider them as purely circulating in this exact treatment).  Those round the outside of the eight correspond to bouncing trajectories.
In the full phase space the figure-eight gives two homoclinic submanifolds to the hyperbolic periodic orbit.
For generic perturbation they can be expected to break, producing homoclinic oscillations and a zone of chaos around them.  See Figure~\ref{fig:mflds}.
\begin{figure}[htbp] %  figure placement: here, top, bottom, or page
\centering
\includegraphics[width=4in]{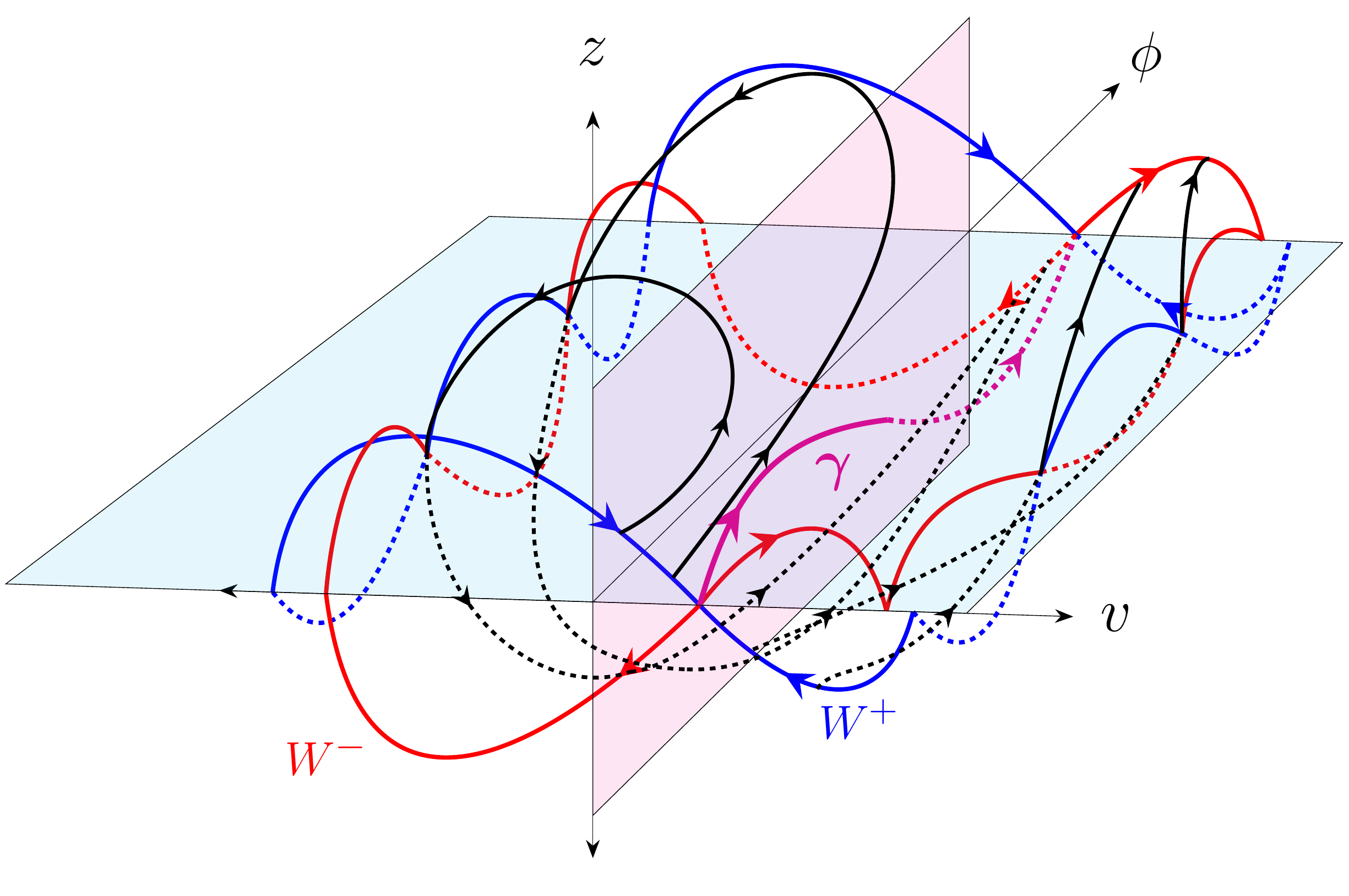} 
\caption{Expected picture for the contracting manifolds $W^\pm$ of the hyperbolic periodic orbit $\gamma$ after a perturbation breaking axisymmetry.}
\label{fig:mflds}
\end{figure}

Note that $\dot{\phi}$ is not of fixed sign (its sign is essentially that of $v$), so one can not draw a first return map to $\phi=0$ modulo $2\pi$.  But one should still get a similar homoclinic tangle (which would be worth illustrating).
This corresponds to repeated transitions between bouncing trajectories and trajectories circulating in either of the two directions. Note that trajectories on the counter-circulating side and the bouncing trajectories on the outside of the figure-eight still have a net drift in the positive direction if close to the separatrix. But further away, $\la \dot{\phi}\ra$ may change to negative.
It looks likely that for each of the counter-circulating and the bouncing classes there is a level set of $p_\phi$ for given $E$ for the unperturbed problem for which $\la \dot{\phi}\ra$ is zero; generically they break into island chains in the $\phi$ direction, creating what are called ``super-bananas'', together with their own broken separatrices.  The scenario merits more in-depth analysis from this point of view (and discussion of ripple-trapping).

The conclusion for isodrasticity is that to be strong isodrastic, one would need to restrict perturbations from axisymmetry to preserve the pairs of homoclinic orbits to the hyperbolic periodic orbits.  This requires two functions $f_\pm:\R \times \R \times \S^1 \to \R$ to be zero, $f_\pm(E,\mu,\phi)$ representing the splitting of the separatrix on the right or left, respectively, for energy $E$ and magnetic moment $\mu$ at phase $\phi$.
It is an open question how one might achieve this in general. Perhaps the helical fields of \cite{KIVM} would give some insights.
 
We note that quasi-symmetry  implies the same situation of double homoclinics to periodic orbits.  Indeed, the analysis proceeds just by replacing conservation of $p_\phi$ by conservation of $K = -e\psi + m v_\pl u\cdot b$, where $u$ is the quasi-symmetry vector field.

%\subsection{Stellarator}
%XXX how about treating a toy stellarator?  I think not: I don't know a good starting point, unlike the axisymmetric cases.

\section{Exact treatment of the splitting of separatrices}
\label{sec:splitting}

To decide whether $W^\pm$ for $N^-$ coincide there is a standard technique, usually called Melnikov's method, though its origins go back to Poincar\'e \cite{Poi}.  It is usually presented in the context of time-periodic perturbation of an autonomous 1 DoF system, whereas we want it for an autonomous 2DoF system (nevertheless, see \cite{R,HM82a,HM82b}, for example).  Furthermore, it is usually presented in the case that the unperturbed system has a manifold of hyperbolic periodic orbits with homoclinic connections, whereas here the unperturbed system has a manifold of equilibria with homoclinic orbits and the periodic orbits appear only as a result of the perturbation.  It is often presented as a first-order calculation in some perturbation parameter, but there is an exact version, to be presented here.
Furthermore, the integral of the Melnikov function between zeroes represents a flux of energy-surface volume making a given transition and the flux can be written as the difference in action between homoclinic orbits \cite{MM}.

%\section{Ways to view guiding centre motion as a perturbation of ZGCM}
%One way to consider ZGCM (\ref{eq:ZGCM}) as a perturbation of FGCM (\ref{eq:drift},\ref{eq:vdrift},\ref{eq:tildeB}) is to consider $\frac1e$ as a small parameter, relative to $\frac{BL}{mv}$ and $L\sqrt{\frac{B}{m\mu}}$ for typical values of $B,v,\mu$ and lengthscale $L$ for $B$.  But alternatively, one can scale out $m$ and $e$ to make FGCM have only one parameter $\tilde{\mu}$ and then consider ZGCM as the low energy regime of FGCM.

%We put $u= \frac{m}{e} v$, $\tau = \frac{e}{m} t$ and $\tilde{\mu} = \frac{m}{e^2}\mu$.  Then FGCM becomes
%\begin{eqnarray}
%\frac{dX}{d\tau} &=& \frac{1}{\tilde{B}_\pl} (u\tilde{B} + \tilde{\mu}b\times\nabla|B|) \\
%\frac{du}{d\tau} &=& -\tilde{\mu}\frac{\tilde{B}\cdot\nabla |B|}{\tilde{B}_\pl} \\
%\tilde{B} &=& B + u\ \curl b
%\end{eqnarray}
%and is the Hamiltonian dynamics of the Hamiltonian and symplectic form
%\begin{eqnarray}
%\frac{m}{e^2} H &=& \frac12 u^2 + \tilde{\mu}|B| \\
%\frac{1}{e}\omega &= & \beta + d(u b^\flat).
%\end{eqnarray}

%This viewpoint reduces the number of parameters and allows one to extend to higher-order guiding centre approximations that are relevant for high energy, in particular for the $\alpha$-particles produced by $D-T$ fusion.

%One could also non-dimensionalise arclength $s$ by a typical lengthscale $L$ for variation of $B$, $B$ by a typical field-strength $B_0$, and $\tilde{\mu}$ by $1/B_0$, but little is gained by this.

The NHS $N^-$ is a graph of arclength and parallel velocity over $\Sigma^-$. $\Sigma^-$ is transverse to the magnetic field, so we can label its points (and those of $N^-$) by labels for the fieldline through the point.  We write such labels as $(x,y) \in \R^2$, often shortened to $\xi \in \R^2$.  It is possible to choose them so that the magnetic flux 2-form $\beta = dx \wedge dy$, in which case $(x,y)$ are called Clebsch coordinates, but that is not really necessary.

For ZGCM, the stable and unstable manifolds $W^\pm(\xi)$ of $\xi \in \Sigma^-$ follow the fieldline through $\xi$ and have $u = \pm\sqrt{2(|B(\xi)|-|B|)}$.  We suppose that in a given direction from $\Sigma^-$, the fieldline through $\xi$ reaches a bounce point, where $|B| = |B(\xi)|$ again.  Then $W^\pm(\xi)$ merge at the bounce point.  The union of $W^\pm(\xi)$ over $\xi \in \Sigma^-$ gives a separatrix that it is impossible to cross.

For FGCM with given $\tilde{\mu}$, $W^\pm(\xi)$ for $\xi \in N^-$ follow close to the fieldline through $\xi$ and with parallel velocity $u = \pm\sqrt{2(\tilde{H}(\xi)-|B|)}$. $W^\pm(\xi)$ reach bounce points where $u=0$, equivalently $|B| = \tilde{H}(\xi)$, but in general at different points because they are on different fieldlines by that time.  

Define $\Xi^\pm(\xi)$ to be the fieldlines on which $W^\pm(\xi)$ find themselves at their first bounces.  The splitting of the separatrices of ZGCM under the perturbation parameter $\tilde{\mu}$ is to do with the displacement from $\Xi^+(\xi)$ to $\Xi^-(\xi)$ in the surface of section $\{u=0\}$.
More precisely, given $\xi \in N^-$ let $G(\xi)$ be its trajectory on $N^-$; then
the minimum over $\xi' \in G(\xi)$ of the displacement from $\Xi^+(\xi)$ to $\Xi^-(\xi')$ is an appropriate quantifier of the splitting at $\xi$, because if it is zero that would still make a homoclinic orbit, just with a phase change.  As a global quantifier of the splitting, one can take the maximum over $\xi$.  Note that, in general there is not a natural notion of subtraction in the space of fieldline labels, but one can make one locally.  Furthermore, the limit for small separations makes sense as a tangent vector in the space of fieldlines, which we will use for first-order treatment.

For the moment, we stay with the exact treatment, and show how to compute $\Xi^\pm(\xi)$, at least conceptually.
The idea is to integrate the rate of change in fieldline along $W^\pm(\xi)$.  Note that $W^\pm(\xi)$ are not trajectories in general.  They are the sets of points whose forward/backward trajectory converges together with that of $\xi \in N^-$.  But because there is only one unstable/stable normal direction to $N^-$, they are one-dimensional.  Furthermore, $W^\pm(\xi)$ are in the same energy level as $\xi$.

Because the dynamics on $N^-$ is two-dimensional and Hamiltonian, it consists of periodic orbits or exceptionally, equilibria or connecting orbits between equilibria.  The tangent to $W^+(\xi)$ at an equilibrium $\xi \in N^-$ is the contracting eigenvector of the derivative of the vector field.  An equilibrium in $N^-$ is actually in $\Sigma^-$ so its contracting eigenvector is relatively easy to find.  The tangent to $W^+(\xi)$ at a point $\xi \in N^-$ of a periodic orbit on $N^-$ is the contracting eigenvector of $D\phi_T(\xi)$ where $\phi$ is the flow of FGCM and $T$ is the period of the orbit.  This is not so straightforward to compute but is feasible.  For $\xi$ on a connecting orbit on $N^-$ from one equilibrium $\xi_1$ to another $\xi_2$, the tangent to $W^+(\xi)$ is given by flowing the contracting eigenvector at $\xi_2$ backwards along the connecting orbit to $\xi$.

The analogous constructions for the opposite direction of time give the tangent to $W^-(\xi)$ for all points $\xi \in N^-$.

For $\xi \in N^-$ an equilibrium, $W^\pm(\xi)$ are trajectories, and they start along the fieldline with $u = \pm \sqrt{-|B|''}\, s$ to first order in $s$.  The velocity $dX/d\tau$ can be pulled back to $\Sigma^-$ by the derivative of $b$ (plus a correction to bring it tangent to $\Sigma^-$) and thus the point $\Xi^+(\xi)$ can be obtained by integrating this pulled-back vector field on $\Sigma^-$.  Similarly for $\Xi^-(\xi)$.

For $\xi \in N^-$ on a periodic orbit, if we already found a small piece of $W^+(\xi)$ near $\xi$ then we can construct a longer piece by applying $\phi_{-T}$ to it, where $T$ is the period.  So we can take a small piece given approximately by the straight line in the tangent direction at $\xi$ and iterate it backwards by $\phi_T$.  The forwards contracting manifold of a periodic orbit $\gamma$ is the union of $W^+(\xi)$ over $\xi \in \gamma$.

%More generally,  case of $W^+(\xi)$.  If the equation for trajectories is written as $\dot{x}=v(x)$ then tangent trajectories are given by $\dot{\eta}(t) = Dv(x(t))\eta(t)$. We want tangent trajectories that decay exponentially to zero as $t \to +\infty$.  To normalise them, we choose a co-vector $\lambda$ at $x(0)$ and require $\lambda \eta(0) = 1$.  For example, for $x(0)$ near $N^-$ one could choose $\lambda = -d|B|'$, but it will need adjusting as $x(0)$ is taken further away, to make sure that the tangent to $W^+(\xi)$ is not in its kernel.  Then assuming $x(t)$ converges to the normally hyperbolic $N^-$, the map 
%$$\eta \mapsto (\dot{\eta}-Dv \circ x \,\eta,\lambda \eta(0)-1)$$
%on $C^1$ functions $\eta$ on $[0,\infty)$ to $C^0$ functions (with exponentially weighted norms) and a scalar has a unique solution, which is the desired tangent trajectory.  Evaluating it at $t=0$, $\eta(0)$ is a tangent to $W^+(\xi)$ at $x(0)$.

Note that the choice of surface of section to be the bounce surface is not necessary, but is somewhat natural in this problem.

%We plan to illustrate the computation of the NHS and their contracting manifolds in a future paper.

Note also that this exact picture gives an exact formula for the flux of energy-surface volume making a transition.  The transitions are represented by lobes formed by the backwards and forwards contracting manifolds between successive intersections.  The intersections correspond to homoclinic orbits.  The flux of energy-surface volume across such a lobe is precisely the difference in action between the two homoclinic orbits.  Here, the action of a curve is $\int eA^\flat + mv_\pl b^\flat$ and to compare the actions of two homoclinic orbits one has to take the limit of long segments whose backwards endpoints converge together (and to the periodic orbit) and forwards endpoints converge together (and to the periodic orbit).
This follows from \cite{MM}.

\section{Melnikov functions}
\label{sec:Mel}
%{\color{red}in what follows change $M_h$ to $\mathcal{M}_h$ for consistency with the Melnikov ``tease" earlier. (There I changed $M$ to $\mathcal{M}$ to avoid conflict with the symbol used for ZGC bounce orbit $M$.)--JB}
Here we will show that to first order the conditions for strong isodrasticity are precisely the weak isodrasticity conditions, and give a computational test for weak isodrasticity in terms of ``Melnikov functions'' that is simpler than that of Section~\ref{sec:wiso}.  Finally, we interpret the Melnikov functions in terms of flux of phase-space volume making transitions.

\subsection{Strong isodrasticity to first order}
To compute $\Xi^{\pm}$ to first order in $\tilde{\mu}$, we can integrate the rate of change of fieldline label along the unperturbed trajectory of ZGCM, compensated by the rate of change of the fieldline label at the base point.

Recall that the equations to first order can be written in scaled time $\tau$ as
$$\frac{dX}{d\tau} = ub+\frac{\sqrt{\tilde{\mu}}}{|B|} \left( u^2 c_\perp  + b \times \nabla |B| \right) ,$$
with 
\begin{equation}
u = \pm\sqrt{2(h-|B|)}
\label{eq:uh}
\end{equation}
and $c = \curl\, b$ (as already mentioned, $c_\perp$ can be written as $b\times \kappa$ with $\kappa = b\cdot \nabla b$).

Let $\xi$ be a coordinate on $\Sigma^-$ (e.g.~$x$ or $y$ from the previous section) and extend it along the fieldlines to a fieldline label.  Its derivative $a = d\xi$ (a covector) can be computed by setting $a(0)=d\xi$ on tangents to $\Sigma^-$ and $a(0) b = 0$, and integrating the adjoint equation along the fieldlines: 
\begin{equation}
\frac{da_i}{ds} = -a_j(s)\partial_i b^j(x(s)), \quad \frac{dx^i}{ds} = b^i(x(s)),
\label{eq:a}
\end{equation}
where $s$ is arclength. %(again one has to interpret the derivative of a vector field correctly).
The covector $a$ quantifies the linearised change in the fieldline label.
For trajectories starting on $N^-$, take $h$ in (\ref{eq:uh}) to be $|B|$ on $\Sigma^-$ to first order.
Then the first-order change $\Delta\xi$ in fieldline label $\xi$ along $W^-$ is given by integrating the rate of change of $\xi$ by the perturbed vector field along the unperturbed trajectory, while simultaneously subtracting off its rate of change at the base point.

Thus
$$\Delta \xi = \sqrt{\tilde{\mu}} \int \frac{u}{|B|} a\, c_\perp + \frac{1}{u} \left(a \frac{b\times \nabla |B|}{|B|} - C\right)\ ds,$$
where $C = a\frac{b\times\nabla |B|}{|B|}$ at the initial point on $\Sigma^-$.
The subtraction of $C$ compensates for the denominator $u$ at the start of the integration, where $u \sim \sqrt{-|B|''}\, s$ for small $s$.  The denominator does not give a problem at a generic bounce point because it behaves like $\sqrt{s_0-s}$ there, which gives an integrable singularity; a tidy way to compute the integral accurately is to switch variable of integration from $s$ to $u$ near the bounce.

The first-order change along $W^+$ is the negative of $\Delta\xi$.  So the displacement in the chosen fieldline label $\xi$ from $W^+$ to $W^-$ is to first order $$2 \Delta \xi = 2 \sqrt{\tilde{\mu}} \mathcal{M}_\xi,$$ with 
$$\mathcal{M}_\xi =  \int \frac{u}{|B|} a\, c_\perp + \frac{1}{u} \left(a \frac{b\times \nabla |B|}{|B|} - C\right)\ ds.$$
It is a function of the initial point on $\Sigma^-$ and we call it the {\em Melnikov function} for the fieldline label $\xi$, by analogy with formulae for splitting of separatrices in other contexts.  

If one takes two independent fieldline labels $x,y$ then the Melnikov function has two components $\mathcal{M}_x, \mathcal{M}_y$.  This is particularly important when one takes initial points on or near a critical point of $|B|$ on $\Sigma^-$.  Compare the discussion of the mirror machine in Section~\ref{sec:mirror}.

For a periodic orbit on $N^-$, however, we are more interested in the displacement in fieldline label $h=|B|$ where it crosses $\Sigma^-$ than the possible phase shift along the periodic orbit.  Thus we specialise to take $h$ as coordinate on $\Sigma^-$ and compute the resulting first-order change in $h$.  In this case several things simplify.  Firstly, we can take $a(0) = d|B|$ because $h=|B|$ on $\Sigma^-$ and $i_b d|B| = 0$ on $\Sigma^-$.   Secondly, $C=0$.  Then letting
$$k = \curl (ub)$$ we obtain
$$\mathcal{M}_h =  \int \frac{a k_\perp}{|B|} \ ds$$
and $a b = 0$ so we can drop the $\perp$. Thus we obtain the simple formula
\begin{equation}
\mathcal{M}_h = \int \frac{a k}{|B|}\, ds.
\label{eq:M_h}
\end{equation}
Of course this hides the computation of $k$ and also the fact that $k$ is singular at $s=0$ (but the singularity is annihilated by $a$).

Note that in practice it may be better to integrate with respect to time $T$ for fieldline flow  $dx/dT = B(x)$ than arclength $s$.  This would replace $ds/|B|$ by $dT$, and also simplify the computation of (\ref{eq:a}) as indicated in Appendix~\ref{app:sqrts}.

\subsection{Relation to weak isodrasticity}
Next, we relate the Melnikov function $\mathcal{M}_h$ to the ratio of $dh \wedge d\text{\j}$ to $\beta$ on $\Sigma^-$.
\begin{thm}
$dh \wedge d\text{\j} = \mathcal{M}_h\beta.$  
\end{thm}
\noindent Thus the main part of the weak isodrastic condition, that $d\text{\j}$ and $dh$ be linearly independent on $\Sigma^-$, is the first-order condition for strong isodrasticity, and $\cM_h$ is the function $\cM$ of Section~\ref{sec:wiso}.

\begin{proof}
Recall that $\text{\j} = \int u\, ds = \int u b^\flat$ along the segment of fieldline from $\Sigma^-$ to the first bounce point.  Hence for a tangent vector $v$ to $\Sigma^-$, $d\text{\j} (v) = \int i_{\tilde{v}}d(ub^\flat)$, where $\tilde{v}$ is obtained from $v$ by flowing with the derivative of $\dot{x}=b(x)$ ($\dot{\tilde{v}}^i = \tilde{v}^j \partial_j b^i$).
But $d(ub^\flat) = i_k\Omega$ and displacement along the segment is given by $b\, ds$, so
$$d\text{\j}(v) = \int i_b i_{\tilde{v}} i_k \Omega\, ds.$$
Now $\Omega = \beta \wedge b^\flat /|B|$, and $i_b \beta = 0, i_b b^\flat = 1$, so this reduces to
$$d\text{\j}(v) = \int  i_{\tilde{v}}i_k \beta \frac{ds}{|B|}.$$

Near a point of $\Sigma^-$ with $dh \ne 0$ we can choose $h$ as one fieldline label and can take another one $g$ such that $\beta = dh \wedge dg$ (by Darboux's theorem; this is a construction of Clebsch coordinates).  Extend $h$ and $g$ to fieldline labels along the fieldlines.  Let us take $v$ to have $dh\, v = 0$.  Then 
$$d\text{\j} (v) = \int i_k dh\, i_{\tilde{v}} dg \frac{ds}{|B|} = \int i_k dh \frac{ds}{|B|} i_v dg = \mathcal{M}_h\, i_v dg,$$ 
as $i_{\tilde{v}}dg$ is constant along the fieldline.
Then evaluate $dh \wedge d\text{\j} (w,v)$ on $\Sigma^-$ for an arbitrary $w$.  It is $dh(w) d\text{\j}(v) = dh(w) \mathcal{M}_h\, i_v dg = \mathcal{M}_h\, \beta(w,v)$.  Since $\Sigma^-$ is only two-dimensional, it suffices to evaluate 2-forms on any pair of independent vectors to determine them.  Hence $dh \wedge d\text{\j} = \mathcal{M}_h\, \beta$, i.e.~$\mathcal{M}_h$ is the ratio of $dh \wedge d\text{\j}$ to $\beta$ on $\Sigma^-$.

At a point of $\Sigma^-$ where $dh=0$ then both $\mathcal{M}_h$ and $dh \wedge d\text{\j}$ are zero, so the same result holds.
\end{proof}

% CUT OR MOVE Note that the equation of motion can be written
%$$\frac{dX}{d\tau} = ub+\frac{\sqrt{\tilde{\mu}}}{|B|} u k_\perp ,$$
%with $u = \pm\sqrt{2(h-|B|)}$ and $k = \curl(ub)$.

%XXX Also derive
%$$M(\xi) = \frac{1}{\Omega(b,X,Y)}\int\Omega(b,\phi_{s*}X,k)\ ds,$$
%or forget about it???

Computation of the Melnikov function via (\ref{eq:M_h}) (rather than by numerical differentiation of $h$ and $\text{\j}$) for some examples will be reported in a future paper.

To complete the discussion of strong isodrasticity at first order, we look at $\Sigma^0$.
The strong isodrasticity condition that $H=\tfrac12 u^2 + |B|$ is constant along components of $N^0$, evaluated to first order is equivalent to $h$ being constant on components of $\Sigma^0$.
Linear dependence of $dh$ and $d\text{\j}$ on $N^-$ extends to the boundary, so from $h$ constant on boundary components we also deduce that $\text{\j}$ is constant on them.

Hence strong isodrasticity at first order is equivalent to weak isodrasticity, as claimed.

\subsection{Interpretation as flux}
To conclude this section, we interpret the Melnikov function in terms of the first-order flux of phase-space volume making the given transition.
This is analogous to the interpretation of $dh\wedge d\j$ as transition flux for reduced dynamics in Section~\ref{sec:wiso}.

For a 2 DoF Hamiltonian system with Hamiltonian $H$, symplectic form $\omega$ and vector field $V$ defined by $i_V\omega = dH$, the phase-space volume form is $\tfrac12 \omega \wedge \omega$ and the energy-surface volume form $\eps$ on $H^{-1}(E)$ is defined so that $\tfrac12 \omega\wedge\omega = dH \wedge \eps$.  A standard calculation (e.g.~\cite{MM,M90}) shows that the energy-surface volume flux form $\phi = i_V\eps$ is just $\omega$.  

Specialising to the mirror machine example for illustration, it follows that the flux of energy-surface volume making the transition from free to bouncing  is the integral of $\omega$ over the lobe between $W^\pm$ on the bounce surface $B=\{v=0\}$ in Figure~\ref{fig:M4}.
In the scaling we are using, $\omega = \beta/\sqrt{\tilde{\mu}} + d(u b^\flat)$, so is just $\beta/\sqrt{\tilde{\mu}}$ on $\{v=0\}$.  Choose a local coordinate $g$ on $\Sigma^-$ such that $\beta = dh \wedge dg$ (possible where $dh\ne 0$ by Darboux's theorem) and flow the functions $h$ and $g$ along the field.  Then the flux is $\int dh \wedge dg/\sqrt{\tilde{\mu}}$ over the lobe.  But the change in $h$ from $W^+$ to $W^-$ is to first order $2 \sqrt{\tilde{\mu}}\mathcal{M}_h$ so the flux is $\int 2 \mathcal{M}_h dg$ along the arc of the level set of $h$ between zeroes of $\mathcal{M}_h$.

We can treat a range of energies simultaneously.  The flux-form for phase space volume is $i_V(\tfrac12 \omega \wedge \omega) = dH \wedge \omega$.  So, given an area $A$ on $\Sigma^-$ corresponding to a transition, the flux of phase-space volume making the transition is the integral of $dH \wedge \omega$ across the corresponding region on $\{v=0\}$.  To first order, $dH = dh$, because we are treating trajectories that graze the maximum in $|B|$.  So the flux is 
$\int_A 2\mathcal{M}_h dh \wedge dg = \int_A 2\mathcal{M}_h \beta$.  The factor 2 comes from the convention that $L$ is computed from one zero of $v$ to the other, whereas for a full period it would be twice as much.

So far, we have been treating the problem in scaled variables.  To turn this from scaled variables to the original variables, we use that the real symplectic form is $\sqrt{m\mu}$ times the scaled one (\ref{eq:scaledomega}), so real phase-space volume is $m\mu$ times the scaled one, and real time is $\sqrt{m/\mu}$ times scaled time. So real phase-space flux across $A$ is:
$$2 m^{1/2}\mu^{3/2} \int_A \mathcal{M}_h\beta .$$
A point with given $h$ on $\Sigma^-$ corresponds to real energy $\mu h$.

To convert these results to the flux of particles, we need to introduce a distribution function $\rho$, giving their density with respect to phase-space volume (this use of $\rho$ is distinct from that for the gyroradius).  It is best to include the dependence of the density on $\mu$, rather than treating a fixed $\mu$. The phase-space volume form $dq_1\wedge dq_2\wedge dq_3 \wedge dp_1\wedge dp_2\wedge dp_3$ for the 3DoF problem converts in gyro-coordinates to $m^2 \tilde{B}_\pl\, \Omega \wedge dv_\pl \wedge d\mu \wedge d\phi$, with $\Omega$ being ordinary volume for guiding-centre position and $\phi$ being gyrophase.  Integrating over gyrophase we obtain the gyro-averaged volume-form $2\pi m^2 \tilde{B}_\pl\, \Omega \wedge dv_\pl \wedge d\mu$.  This is preserved by the Hamiltonian guiding-centre flow (though not in general by (\ref{eq:approxdrift}), even if $\tilde{B}_\pl$ is replaced by $|B|$, %\rsm{is this true? I didn't check it, but I think Roscoe White claims non-conservation of Liouville volume in his book.}
another reason to prefer the Hamiltonian version). 
%{\color{red} The correct Liouville form has $\tilde{B}_\parallel$ instead of $|B|$. They agree to leading-order, but the full GC flow does not preserve the form with $|B|$ exactly. --JB}
So the number of particles of given type in a volume $W$ of gyro-averaged phase-space can be written as 
$\int_W \rho(X,v_\pl,\mu)\, 2\pi m^2 \tilde{B}_\pl(X)\, \Omega \wedge dv_\pl \wedge d\mu,$ where $\rho$ denotes the usual scalar guiding-center distribution function.

Note that the gyro-averaged volume-form is the wedge of $2\pi \tfrac{m}{e} d\mu $ with the guiding-centre volume-form.  The prefactor is because the standard convention for action variables was broken by {many} plasma physicists (though tends to be respected by {A.J.~Brizard and} high-energy particle physicists).  To see this factorisation of the gyro-averaged volume-form, reall from (\ref{eq:Lambda}) that the guiding-centre volume-form $\Lambda = em \tilde{B}_\pl\, \Omega \wedge dv_\pl$. 
%{\color{red} Again, only to leading-order. --JB}

Thus to obtain the flux of particles from the flux of guiding-centre volume, one has to multiply by $2\pi \tfrac{m}{e} \rho\, d\mu$ and integrate over $\mu$.
%{\color{red}look at this carefully; add corresponding discussion to weak isodrastic section. [done. I agree with formula. See thm 2 for addition to weak isodrastic section.] }
To leading order in $\mu$ we can replace $\tilde{B}_\pl$ by $|B|$.
So the flux of particles corresponding to an area $A$ on $\Sigma^-$, to leading order in $\mu$, is
$$\int_A \int_{\R_+} \tfrac{4\pi}{e} {(m\mu)^{3/2}} \rho(X,0,\mu)\, d\mu\,  \mathcal{M}_h \beta .$$

A point to note is that the picture to first order depends on $E$ and $\mu$ through only the combination $E/\mu$.  The rate of transition has a prefactor $\sqrt{\mu}$ (or other powers depending on scaling), but is otherwise independent of $\mu$ given $E/\mu$.  As pointed out by Roscoe White (personal communication), this gives hope that fusion $\alpha$-particles might remain in their initial class despite slowing down, because it seems for fusion $\alpha$-particles $E$ and $\mu$ decrease in such a way that $E/\mu$ remains roughly constant.

\section{Discussion}
We have generalised the concept of omnigenous magnetic fields and their analysis by \cite{CS}, to remove the requirement of a flux function (and even with a flux function, our condition is much weaker), obtaining the concept of weakly isodrastic field.  They are magnetic fields for which there are no transitions between classes of guiding-centre motion, under the assumption that the longitudinal invariant is conserved.
Furthermore, we have provided a quantification of deviation from the ideal case, namely Melnikov functions, thereby making available objective functions for optimisation of design.

We have extended our theory, by removing the assumption of conservation of the longitudinal invariant, to a notion of strong isodrasticity, which provides an exact prevention of transitions between classes.  We have proved that to first-order in the magnetic moment, strong isodrasticity is weak isodrasticity, thus justifying weak isodrasticity as a first-order approximation.
We have illustrated how isodrasticity is lost for general perturbations from axisymmetric fields for toy mirror machines and a toy tokamak.
The exact theory has the advantage that it can in principle be applied to higher order guiding-centre approximations, of potential relevance to the fusion alpha-particles.

We have shown how to construct many weakly isodrastic mirror fields, in particular that are not omnigenous.
The main question that remains is whether isodrasticity (weak or strong) is possible outside axisymmetry for a stellarator.  Quasisymmetry implies isodrasticity.  Although perfect quasisymmetry perhaps does not exist outside axisymmetry, close to quasisymmetric fields can be made (for recent examples, see \cite{LP}).
Close to quasisymmetric fields should be close to isodrastic, but perhaps there is a larger class of isodrastic fields than quasisymmetric.  Omnigenous fields are weak isodrastic.  Perhaps truly omnigenous fields have to be quasisymmetric, but close to omnigenous fields can be made, such as in Wendelstein 7-X.  As weak isodrasticity is weaker than omnigenity, there is hope that truly isodrastic fields can be made. One option is to make all the marginal cases heteroclinic, but we have shown (not included in this paper) that this reduces to omnigenity.  So we wish to make examples where most of the marginal cases are homoclinic, as discussed for typical perturbations of a tokamak in Section~\ref{sec:tok}, but there will be codimension-one cases of heteroclinic, leading to the double transition scenario of Appendix~\ref{app:double}.
We are looking for examples.  
For general Hamiltonian systems it is certainly true that there is a larger class of systems with some perfect separatrices than the integrable class (see Appendix~\ref{app:sep}).
Also, the concept of isodrasticity motivates some clear objective functions that could be fed to an optimisation to automate a search.  For example, one could use the maximum of the Melnikov function or the integral of its positive part.  Even if one does not make exact isodrasticity, these objective functions could be weighed against others in the design of stellarators.

The theory also suggests that one might be able to make controlled transitions between classes by slightly breaking isodrasticity via trim coils.  The Melnikov function for a normally hyperbolic submanifold specifies the flux of energy-surface volume making a given transition, so if one could learn the effects of trim coils on the Melnikov function then one could control the flux.
%It also suggests refining the analysis of divertors, particularly those of island-type, by understanding their effect on guiding-centre motion in the style of this paper.  That will have to wait for another paper.

%The theory also suggests ways to design controllable divertors.  A divertor is a structure in the magnetic field that leads a small fraction of particles to the wall,  goals being to give the main plasma a clear edge separated from the wall, to siphon off impurities and to remove heat from the plasma once it starts fusing (for eventual conversion to electricity).  The Melnikov function for a NHS specifies the flux of energy-surface volume making a certain transition.  If the transition is to trajectories to the wall then controlling the Melnikov function will control the flux to the wall.  In the context of stellarators, the principal transitions are from bouncing within one well to passing to the next well and potentially circulating many times thereafter.  So to exploit the idea, one might have to consider the circulating particles as the exhaust.  This is an opposite point of view from normal.  Alternatively, one could consider the bouncing particles to be the exhaust, but usually plasma physicists go to great lengths to keep most of them confined.

Finally, the theory sheds light on the effects of imperfections on
tokamaks and quasi-symmetric stellarators.  The marginal bouncing trajectories of the ideal case give rise to double homoclinic manifolds to hyperbolic periodic orbits for guiding-centre motion.  Breaking the symmetry in general breaks both these separatrices, leading to a stochastic layer in which trajectories transition between bouncing and co- and counter-circulating.
Furthermore, the guiding-centre dynamics near the magnetic axis is in general significantly disturbed from the axisymmetric case, leading to mixing in the core (though this is not necessarily a bad thing, as discussed by \cite{Boozer}).

\section*{Acknowledgements}
This work was supported by the Simons Foundation (601970, RSM) under the ``Hidden symmetries and fusion energy'' collaboration. 
%\rsm{and LANL}
We also acknowledge funding support from the US DOE Office of Advanced Scientific Computing Research (ASCR), DOE-FOA-2493 “Data-intensive scientific machine
learning and analysis”.
We are grateful to Nikos Kallinikos for the etymology and for making some of the figures, to Elizabeth Paul for trying out our ideas in various realistic nearly quasi-symmetric fields, and to others in the collaboration for their comments, notably Matt Landreman, Per Helander, Roscoe White, Eduardo Rodriguez and Gabriel Plunk.  Also we thank Anatole Neishtadt and John Cary for useful pointers to literature.

\appendix

\section{Electrostatic and gravitational fields and relativity}
\label{app:esrel}

\subsection{Electrostatic and gravitational fields}
To add the effects of an electrostatic or gravitational field on guiding centre motion, one takes
$$H = \tfrac12 mv_\pl^2 + \mu|B|+e\Phi + m V,$$ where $\Phi$ and $V$ are electrostatic and gravitational potentials, respectively.
The symplectic form is unchanged.
The resulting equations of motion are
\begin{align}
\dot{X} &= \left({v_\pl} \widetilde{B} + \frac{\mu}{e} b \times \nabla|B| + b \times \nabla \Phi + \frac{m}{e}b \times \nabla V \right)/\widetilde{B}_\pl \nonumber \\
\dot{v}_\pl &= - \frac{\widetilde{B}}{\widetilde{B}_\pl} \cdot \left(\frac{\mu}{m}\nabla |B| + \frac{e}{m}\nabla \Phi + \nabla V
\right). \nonumber
\end{align}

ZGCM is just 1 DoF motion of unit mass in the potential $\frac{\mu}{m}|B|+\frac{e}{m}\Phi + V$.  The longitudinal invariant  is modified to $$L =\int_{s_1}^{s_2} \sqrt{2m(E-\mu|B|-{e}\Phi-mV)}\, ds.$$

The surface $\Sigma$ is modified to be the set of points at which the first derivative $\mu |B|' + e\Phi' +m V' = 0$, where $'$ again denotes derivative along the magnetic field, and decomposes into $\Sigma^\pm$ and $\Sigma^0$ according to the sign of the second derivative.  But $\Sigma$ now depends on the ratio $\mu:e:m$.  Weak isodrasticity becomes that $dE\wedge dL = 0$ on $\Sigma^-$, where $L$ is for segments starting on $\Sigma^-$.  The exact FGCM dynamics is still 2DoF and has NHS with contracting submanifolds, whose intersections can be analysed the same way.

\subsection{Relativity}
As explained in an appendix to \cite{BKM}, relativity can be incorporated in either laboratory time or proper time.  We treat the former here. The first adiabatic invariant becomes $\mu = \frac{p_\perp^2}{2m|B|}$ with $p = \gamma m v$, $\gamma =(1-|v|^2/c^2)^{-1/2}$.  The guiding-centre Hamiltonian becomes $$H = c\sqrt{m^2c^2+p_\pl^2+2m\mu|B|}+e\Phi+mV.$$
The second adiabatic invariant becomes $$L =\int p_\pl\, ds = \int \sqrt{(\tfrac{E}{c}-e\Phi-mV)^2-m^2c^2-2m\mu|B|}\, ds$$
at $H=E$.  If $\Phi$ and $V$ are constant we again have $\Sigma$ the set of points where $|B|'=0$ and the same analysis of weak isodrasticity.  With $\Phi$ or $V$, however,
the shape of the argument of the square root now depends on $E$ so $\Sigma$ becomes $E$-dependent and the nice picture breaks down.

Nonetheless the exact FGCM still has a NHS continuing the non-relativistic $\Sigma$ to the relativistic case, at least for $\mu$ not too large.  The same picture of its contracting manifolds holds.  
For magnetically confined fusion devices, probably adiabatic invariance of $\mu$ fails before relativistic effects come in (the alpha particles produced by DT fusion have speed 4.3\% of the speed of light and in a 1T field have gyroradius 0.266 metres times the sine of their pitch angle).  But perhaps in some astrophysical contexts, relativistic guiding-centre motion is relevant.

\section{Fields with a flux function}
\label{app:fluxfn}
We take the opportunity to review 
the theory of magnetic fields possessing a flux function, including some results that we have not found in the literature.

Existence of a flux function is automatic for non-degenerate magnetohydrostatic (MHS) fields, i.e.~those satisfying $i_Ji_B\Omega = -dp$ ($J \times B = \nabla p$) for some function $p$ with $dp \ne 0$ ($\nabla p \ne 0$) almost everywhere, where $i_J\Omega = dB^\flat$ ($J=\curl\ B$):~one can take $\psi = p$ (though it is generally preferred to take the toroidal flux enclosed by a level set of $p$).  It is also automatic for most axisymmetric fields:~let $u$ be the vector field $\partial_\phi$ in cylindrical coordinates, then $d i_u i_B \Omega=0$ ($\curl(B\times u) = 0$) so $i_u i_B\Omega = d \psi$ ($B\times u = \nabla \psi$) for some function $\psi$ locally, and generally there is no cohomological obstruction, so then $\psi$ is a global function, $i_B d\psi = 0$ and $d\psi \ne 0$ except where $B$ is parallel to $u$.
%So assumption of a flux function is a reasonable starting point.  

The components of level sets of $\psi$ are called {\em flux surfaces}.  They are called {\em regular} if $d\psi \ne 0$ ($\nabla \psi \ne 0$) on them.  The bounded boundaryless regular flux surfaces are co-oriented by $\nabla \psi$ and hence oriented. Because they carry a nowhere-zero vector field $B$ their Euler characteristic is zero, so they are tori.  Furthermore, they carry an area-form $\cA$ satisfying $\cA \wedge d\psi = \Omega$, e.g.~$$\cA = i_n \Omega$$ ($\cA(\xi,\eta) = n\cdot(\xi \times \eta)$) with $n = \nabla \psi/|\nabla \psi|^2$.  The fieldline flow preserves $\cA$ on the flux surfaces, because applying $di_B$ to $\cA \wedge d\psi = \Omega$ produces $di_B\cA \wedge d\psi = 0$ and so $di_B\cA = 0$ on tangents to the flux surface.  So there is no asymptotic convergence of fieldlines in either direction of time.  This forces the flow to have a cross-section (a closed curve transverse to the vector field such that the trajectory of every point crosses it in forward and backward time) because the only other option for a nowhere-zero vector field on a 2-torus has a ``Reeb component'' (an annulus bounded by periodic orbits in opposite directions) and area can not be preserved.  See sec.4 of \cite{BGKM} for a summary of the theory of vector fields on a 2-torus.

Furthermore, integrating $i_B\cA$ from a reference point produces a local coordinate on the flux surface that is preserved by the flow.  
It follows that the return map to a cross-section is a rigid translation (in this coordinate).  Thus, the fieldline flow on each torus is equivalent to that of a (non-zero) constant vector field on a standard torus $\R^2/\Z^2$ up to a possible time-change (and examples can be made for which that is necessary) (equivalently, conjugate to a vector field with a constant direction).  In particular, the fieldlines wind around the torus with a well-defined winding ratio (limit of ratio of number of turns in two angles), called its {\em rotational transform} $\iota$, and in the rational case all the fieldlines are closed.

In the case of Diophantine winding ratio, i.e.~$|k\iota-m| \le C k^{-\sigma}$ for some $C>0$, $\sigma\ge 1$ for all integers $k,m$ with $k>0$, if the flow is smooth enough then it is
conjugate to a constant vector field without any need for a time-change factor, but the conjugacy is in general less smooth than the flow (by about $\sigma$ derivatives, depending whether one works in $C^r$ or Sobolev spaces).  The $C^\infty$ case is treated by \cite{KH} in Prop 2.9.5, where $C^\infty$ conjugacy results, but the method of proof there gives results for less smoothness.

\section{Omnigenity for passing trajectories}
\label{app:omnigen}
In \cite{LC}, omnigenity is shown to imply that $L$ is constant on flux surfaces for bouncing trajectories of given energy and class (actually, the extension to more than one class comes in \cite{PCHL}).
As mentioned at the end of \cite{BKM}, however, omnigenity also implies that $L$ is constant for circulating trajectories of given energy on each rational flux surface, and  to deduce omnigenity both conditions must be satisfied. 
In the meantime, we realised that constancy of $L$ for circulating trajectories on a rational surface is implied by constancy for bouncing trajectories.

In this appendix we first prove that omnigenity requires constancy of $L$ for all circulating trajectories on a rational flux surface.  Then we prove that this is implied by the same statement for bouncing trajectories.
After this, we give an equivalent proof to \cite{H} that the time-average of $V\cdot \nabla \psi$ is zero for circulating trajectories on irrational tori.

It follows that constancy of $L$ for bouncing trajectories of given energy and class on each flux surface implies omnigenity, a result that we feel had not been established correctly before.  The
continuity argument of \cite{H} from irrational to rational surfaces shows only that the flux-surface average of $i_V d\psi$ for circulating particles is zero on
each rational surface.  But the time-averages along
individual fieldlines on a rational surface need not agree with the flux-surface average.

\subsection{Omnigenity for circulating trajectories on rational flux surfaces}
\label{sec:rattrajs}
The longitudinal adiabatic invariant for general periodic orbits of zeroth-order guiding-centre motion (ZGCM) (not just bouncing ones) is 
$$L=\int_\gamma e A^\flat + m v_\pl b^\flat,$$  
where $\gamma$ is the segment of fieldline covered.  The term in $A^\flat$ gives zero when integrated along a bouncing trajectory because it backtracks exactly, which is why it is usually left out.  Also it is a constant for given flux surface when $\gamma$ is a closed loop restricted to the flux surface, so plays no role in the present discussion.  Thus we take $L = \int_\gamma mv_\pl b^\flat$ for both.  

For circulating periodic orbits $\gamma$ of ZGCM for a field with a flux function $\psi$, the first-order drift in $\psi$ averaged over one period $T$ is
$$\la \dot{\psi}\ra = \frac{1}{T} \int_\gamma V\cdot \nabla \psi \ dt = \frac{1}{T} \int_\gamma \frac{i_Vd\psi}{v_\pl} b^\flat,$$
with 
\begin{equation}
\tfrac12 mv_\pl^2 = E-\mu |B|.
\label{eq:vpl}
\end{equation}
The definition of omnigenity uses the non-Hamiltonian velocity $$V = 
%\frac{v_\pl}{\tilde{B}_\pl} B 
v_\pl b + \frac{mv_\pl^2}{e|B|}c_\perp + \frac{\mu v_\pl}{e|B|}b\times\nabla |B|$$
from (\ref{eq:approxdrift}) rather than the Hamiltonian one 
of (\ref{eq:Xdrift}-\ref{eq:tildeB}).  We drop the $\dot{v}_\pl$ component because we need only $V\cdot \nabla \psi$.
Then using $B\cdot \nabla \psi = 0$, we obtain
$$\la \dot{\psi}\ra = \frac1T \int_\gamma \frac{1}{e|B|} (mv_\pl c + \frac{\mu}{v_\pl}b \times \nabla |B|) \cdot \nabla \psi \ b^\flat
= \frac1T \int_\gamma (\tfrac{mv_\pl}{e|B|} c \cdot \nabla \psi - \tfrac{\mu}{e v_\pl} \xi \cdot \nabla |B|) b^\flat,$$
where $$\xi = (b \times \nabla \psi)/|B|.$$

To test whether $L$ is constant on a flux surface (for given $E$), it is enough to test whether
$dL\ \xi =0$, because tangents to a flux surface are linear combinations of the vector fields $B$ and $\xi$, and the change of $L$ along $B$ is zero.
Using the Lie derivative $L_\xi = i_\xi d + d i_\xi$ (distinguish the function $L$ and the Lie derivative $L_\xi$ along $\xi$) on differential forms,
$$dL\ \xi = \int_\gamma L_\xi(eA^\flat+mv_\pl b^\flat) = \int_\gamma e i_\xi dA^\flat + mi_\xi d(v_\pl b^\flat), $$
the total derivative term integrating to zero because $\gamma$ is closed.
Now $dA^\flat = i_B\Omega$ so along fieldline $\gamma$ it gives zero.
The second term gives
$$dL\ \xi = \int_\gamma m (i_\xi dv_\pl\, b^\flat - v_\pl i_\xi db^\flat).$$
Differentiating (\ref{eq:vpl}) we obtain $mv_\pl dv_\pl = -\mu d|B|$; and using $db^\flat = i_c\Omega$, we obtain $i_\xi db^\flat = -i_c i_\xi\Omega = i_c (b^\flat \wedge d\psi)/|B|$.
Thus
$$dL\ \xi = \int_\gamma -\frac{\mu}{v_\pl} i_\xi d |B|\ b^\flat - \frac{mv_\pl}{|B|} i_c b^\flat d\psi + \frac{mv_\pl}{|B|} b^\flat i_c d\psi.$$
The middle term is zero because $d\psi$ applied to a tangent to $\gamma$ is zero.
We are left with
$$dL\ \xi = \frac{e}{T} \la \dot{\psi}\ra.$$
This proves that for passing particles on a rational flux surface, $\la\dot{\psi}\ra = 0$ iff $L$ is constant on it for given energy.

\subsection{Omnigenity for bouncing trajectories implies omnigenity for circulating ones on rational surfaces}
We scale out the magnetic moment $\mu$ by writing $E = h \mu$ and $L = \sqrt{m \mu}\, j$ with $$j(h) = \int_\gamma \sqrt{2(h-|B|)}\, ds,$$
for segment $\gamma$ of fieldline.
As remarked in section~\ref{sec:wiso}, the function $j$ is the Abel transform of the length function $\ell$ for a fieldline.  $\ell(v)$ is the length of the subsegments of $\gamma$ for which $|B|<v$.  Its Abel transform (with a factor of $\sqrt{2}$) is
\begin{equation}j(h) = \int_{h_{\min}}^h \sqrt{2(h-v)}\, d\ell(v),
\label{eq:abel}
\end{equation} 
where $h_{\min}$ is the minimum of $|B|$ along it (one can replace $h_{\min}$ by $-\infty$ because $d\ell(v) = 0$ for $v<h_{\min}$).
This applies equally well to a closed fieldline with $h>h_{\max}$, the maximum of $|B|$ on it.
%, modulo multiplication by an integer corresponding to the number of times the fieldline crosses its maximum before closing.

The condition that $j$ be the same for all bouncing trajectories of given class on a flux surface with given $h$, for all values of $h$ for this class, implies that the length function $\ell$ is the same for each segment of fieldline for the class on the flux surface, by the Abel inversion formula, e.g.~\cite{Ke}:
$$\ell(v) = \frac{2}{\pi}\int_{h_{\min}}^v \frac{dj(h)}{\sqrt{2(v-h)}}.$$
This is a reformulation of a result of \cite{SS} (see also \cite{CS}).  One could write in more detail about the treatment of multiple classes of bouncing trajectory \cite{PCHL}, but hopefully the above is clear enough.
Note that it includes that $h_{\min}$ and $h_{\max}$ are the same for all fieldlines on the same flux surface.

Our point is that it follows from the length function being the same for each bouncing class and the Abel representation (\ref{eq:abel}) that $j(h)$ is the same for all circulating trajectories on the rational surface with given $h > h_{\max}$.

\subsection{Omnigenity for circulating trajectories on irrational surfaces}
ZGCM at given energy for
circulating particles on an irrational flux surface in given direction is uniquely ergodic (i.e.~it has a unique ergodic
probability measure).  Up to normalisation, the unique ergodic probability measure is given by the area-form $\frac{|B|}{v_\pl}\cA$ (where $\cA$ was defined in Appendix~\ref{app:fluxfn}).  This is because $\cA$ is preserved by $B$ but the time spent anywhere by ZGCM is a factor $|B|/v_\pl$ longer than for fieldline flow. 
Thus, the time-average of any continuous function along ZGCM at energy $E$ is equal to its surface-average with respect to $\frac{|B|}{v_\pl}\cA$.  On a flux surface, $V\cdot \nabla \psi\ \cA = i_V\Omega$ because $\Omega = \cA \wedge d\psi$ and $d\psi$ on tangents to a flux surface is zero.  Now using the non-Hamiltonian form for $V$, (\ref{eq:vpl}) and other relations as above, we obtain
$\frac{|B|}{v_\pl} i_V\Omega = \frac{|B|}{\tilde{B}_\pl} i_B\Omega + \frac{m}{e}d(v_\pl b^\flat)$.
The first term is zero on tangents to the flux surface and the integral of the second term over the flux surface is zero.

One might worry that the time required for convergence to the time-average might be long for some irrationals, but actually under the omnigenity condition we believe a uniform estimate is possible.

It is curious that the results of this subsection and subsection~\ref{sec:rattrajs} are exact for the non-Hamiltonian form for $V$, but would be  true only to first order if we used the Hamiltonian form (\ref{eq:Xdrift}).

\section{Relation to pseudo-symmetry}
\label{app:ps}
A weaker notion than omnigenity was introduced by Mikhailov, e.g.~\cite{M+}, called pseudo-symmetry.  There are various formulations, e.g.~\cite{Sk}, but the way we choose is that a magnetic field is said to be pseudo-symmetric if it has a flux function and on each flux surface the contours of $|B|$ are nowhere tangent to the magnetic field.  

A pseudo-symmetric field is not necessarily isodrastic.  Although the above formulation constrains the local maxima of $|B|$ for a pseudo-symmetric field to form (non-contractible) closed curves on each flux surface, it does not force the separatrix area $j$ to be the same for each homoclinic orbit coming from a local maximum with the given value of $|B|$.
%{\color{red} change last sentence to: ...it does not force $j$ to be the same for each homoclinic orbit coming from a local maximum with given value of $|B|$.}

Conversely, an isodrastic field is not necessarily pseudo-symmetric.  Firstly, an isodrastic field need not have a flux function, though if the marginal trajectories are all heteroclinic then it does ($h=|B|$ on $\Sigma^-$ extended along the fieldlines is a flux function).  Secondly, even if an isodrastic field has a flux function, the contours of $|B|$ on it could have an island chain instead of a curve of minima, as long as the fieldlines see only local minima on crossing the island chain (else new $\Sigma^-$ is created).

\section{$C^4$-generic $\Sigma^0$}
\label{app:generic}
A property is {\em generic} in a topological space if it happens on a countable intersection of open dense subsets.
\begin{thm}
Restricting attention to bounded subsets, $C^4$-generically $\Sigma$ is a $C^3$ surface and $\Sigma^0$ is a $C^2$-curve on $\Sigma$, separating it into $\Sigma^\pm$.
\end{thm}

\begin{proof}
First recall that $\Sigma$ is defined by $|B|'=0$ so if $B$ is $C^r$ with $r\ge 2$ then by the submersion theorem \cite{La} the subset where $d|B|' \ne 0$ is a $C^{r-1}$ surface.  The joint condition that $|B|'=0$ and $d|B|' = 0$ is four conditions ($\div\ B = 0$ does not make any restriction) on three variables, so generically does not happen. Thus $C^r$-generically, $\Sigma$ is a $C^{r-1}$-surface.
Note that the condition $d|B|' \ne 0$ is met in particular on $\Sigma^\pm$, which are the subsets where $|B|'' > 0, < 0$ respectively. Furthermore, it follows that $\Sigma^\pm$ are transverse to $B$.

Secondly, the set $\Sigma^0$ is defined by $|B|'=0, |B|''=0$, so by the submersion theorem the subset where the derivative  $D(|B|',|B|'')$ has rank 2 is a $C^{r-2}$ curve.  The joint condition that $|B|'=0, |B|''=0$ and 
$D(|B|',|B|'')$ does not have rank 2 is 4 conditions on 3 variables, so generically does not happen.  Thus $C^{r}$-generically, $\Sigma^0$ is a $C^{r-2}$ curve.  $\Sigma^0$ lies on $\Sigma$ because $|B|'=0$ on it.  It separates $\Sigma$ into $\Sigma^\pm$ because it is the subset with $|B|''=0$.
\end{proof}

%Weak isodrasticity requires that $|B|$ be constant along $\Sigma^0$.  We show that under the rank-2 condition this does not happen.  For $s$ tangent to $\Sigma^0$, $i_sd|B|' = i_sd|B|'' = 0$, so if also $i_sd|B|=0$ then the three derivatives are linearly dependent.  If $|B|''' \ne 0$ this is not possible, so $|B|'''=0$. XXX COMPLETE THIS

Next we study the generic behaviour of $|B|$ near $\Sigma^0$.
Firstly, near a generic point of $\Sigma^0$, there is a choice of fieldline labels $x,y$ and a coordinate $t$ along fieldlines, near to arc length $s$, such that 
\begin{equation}
|B| = f(x,y)+yt+xt^2+kt^3
\label{eq:Bgen}
\end{equation}
for some function $f$ and constant $k\ne 0$.  Note that there is no remainder in this expression; this can be achieved by normal form principles of singularity theory \cite{Mi}.  The restriction $\div \ B = 0$ plays no role in this, only in the mapping between $(x,y,t)$ and physical space (compare the proof of realisability in Section~\ref{sec:real}).
Then $|B|' = (y+2xt+3kt^2)t'$, where $t'$ denotes $\frac{\partial t}{\partial s}$ holding $x,y$ constant,
so $\Sigma$ is locally the surface 
\begin{equation}
y=-2xt-3kt^2.
\label{eq:D.y}
\end{equation}
Thus we can use $(x,t)$ as coordinates on $\Sigma$.
Furthermore $$|B|'' = (2x+6kt)t'^2 + (y+2xt+3kt^2)t'' =(2x+6kt)t'^2+|B|'t''/t'.$$
%= [2x+6kt^2 - |B|'s''(t)]/s'(t)^2.$$
Now $|B|'=0$ on $\Sigma$, and $t'$ is near $1$,
so $\Sigma^0$ is locally the curve $$x=-3kt,\ y=3kt^2.$$  
Its projection to fieldline labels $(x,y)$ is the curve $3ky = x^2$.  It is a standard fold in $\Sigma$.
$\Sigma^\pm$ are the parts of $\Sigma$ with $x+3kt>0, <0$ respectively.
The result is shown in Figure~\ref{fig:genSigma0}.  We see that there are short bouncing segments on some fieldlines, namely those for which the cubic has a well, i.e.~for which $y < -2xt-3kt^2$.
\begin{figure}[htbp] %  figure placement: here, top, bottom, or page
   \centering
   \includegraphics[width=3in]{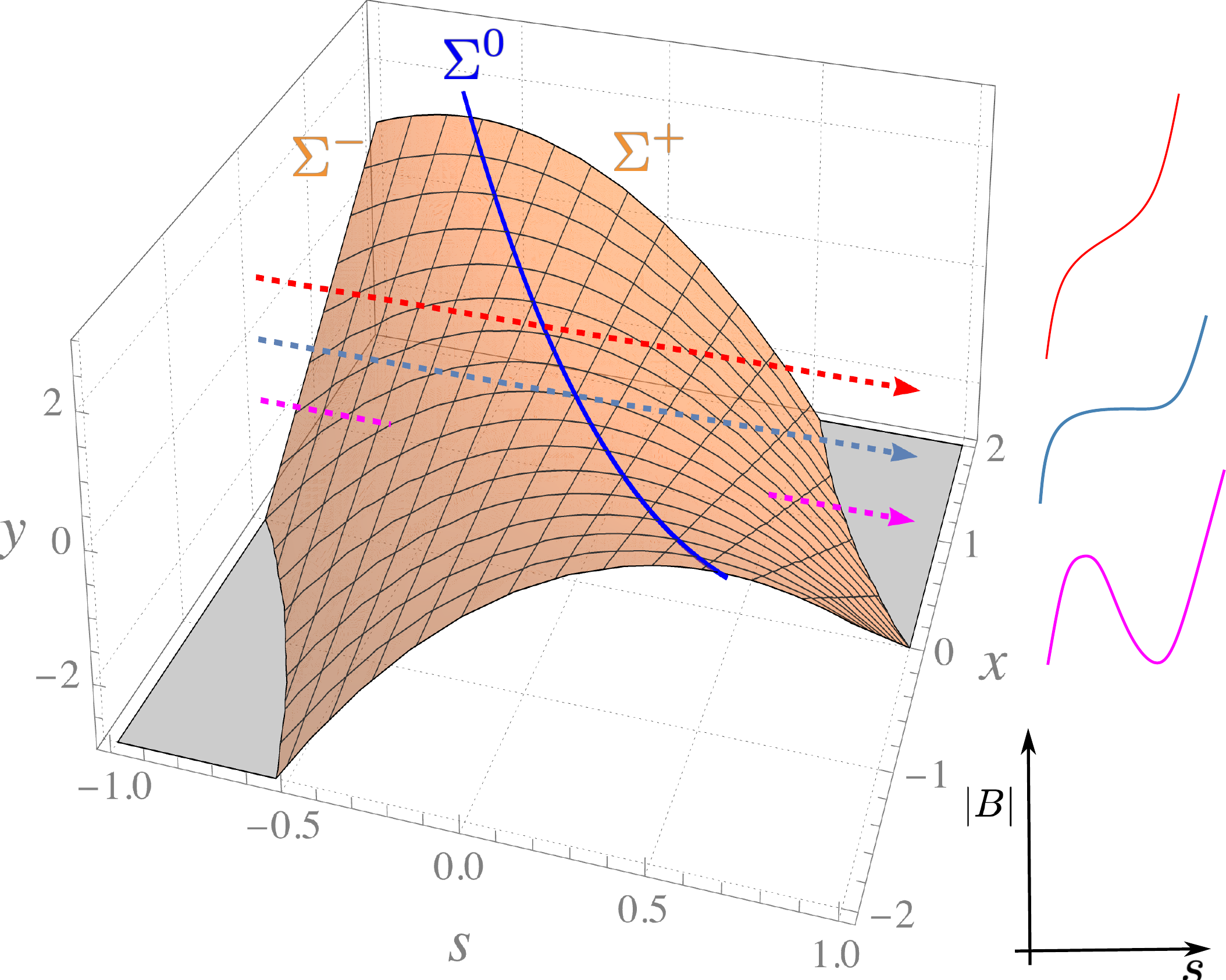} 
   \caption{Generic neighbourhood of a point $(0,0,0)$ on $\Sigma^0$.  $x,y$ are fieldline labels and $s$ is arclength along the field. It shows that $\Sigma^0$ forms a smooth (blue) curve on $\Sigma$, separating $\Sigma$ into $\Sigma^\pm$ on the right and left, respectively. Three fieldlines are indicated, with $x=0$, $y=-1,0,+1$ and the shapes of $|B|$ along them are sketched.}
   % \rsm{add a fieldline that misses $\Sigma$; ideally, just add the fieldlines and the sketches of $|B|$ along them to (a) and suppress (b)}
   \label{fig:genSigma0}
\end{figure}

Secondly, we have to take care of exceptional points $p$ of $\Sigma^0$ at which  $k=0$.
%or the derivative of the mapping from a transverse section $S$ to $B$ through $p$ to $(|B|', |B|'')$ on $S$ fails to have rank 2.
%Generically, at such an exceptional point either $k=0$ and the rank is 2, or $k\ne 0$ and the rank is 1.
Then (\ref{eq:Bgen}) is modified generically to
$$|B|= f(x,y)+yt+xt^2+at^4$$
for some $a \ne 0$.  Repeating the analysis, we obtain that $\Sigma$ is locally $y=-2xt-4at^3$, and $\Sigma^0$ is locally $x=-6at^2,\ y=8at^3$. Its projection to fieldline labels $(x,y)$ is a standard semi-cubic cusp:~$27 a y^2 + 16 x^3 = 0$.

%In the case that the rank is 1, the result is generically the effect of composing Figure~\ref{fig:genSigma0} with a generic  folding map in $(x,y)$.  A generic folding map does not have fold curve tangent to $y=0$, so it does not change the fact that $\Sigma^0$ is locally a smooth curve and that $\Sigma$ is locally a smooth surface.  [XXX Work out what it looks like]

%Locally, by normal form theory for $C^4$ functions, there are fieldline labels $(x,y)$, a reparametrisation $s$ of the fieldlines, and a function $f$ such that
%$|B| = ks^3+xs^2+ys+f(x,y)$ with $k \ne 0$ except at finitely many points, $C^4$-generically.
%
%Then $|B|' = (y+2xt+3kt^2+4at^3)t'(s)$, so $\Sigma$ is locally $$y=-2xt-3kt^2-4at^3.$$
%Furthermore $$|B|'' = (2x+6kt+12at^2)t'(s)^2 + (y+2xt+3kt^2+4at^3)t''(s) 
%= [2x+6kt^2+12at^2 - |B|'s''(t)]/s'(t)^2.$$
%Now $|B|'=0$ on $\Sigma$,
%so $\Sigma^0$ is $$x=-3kt-6at^2,\ y=3kt^2+8at^3.$$

%Having justified that $\Sigma^0$ is a smooth curve on smooth surface $\Sigma$, we next compute the generic behaviour of $|B|$ level sets on $\Sigma$ near $\Sigma^0$.
We can combine the treatment of generic points of $\Sigma^0$ and those having $k=0$ by including both the $kt^3$ and $at^4$ terms in the expression for $|B|$, with $k,a$ not simultaneously zero.
Then on $\Sigma$, $$|B|=f(x,-2xt-3kt^2-4at^3)-xt^2-2kt^3-3at^4.$$
In particular, along $\Sigma^0$, $$|B|= f(-3kt-6at^2,3kt^2+8at^3) + kt^3+3at^4,$$ so using $t$ now as a coordinate  along $\Sigma^0$, $\frac{d|B|}{dt} = -3k f_{,x}$ at $t=0$, where $f_{,x}$ denotes the partial derivative of $f$ with respect to $x$ at $(0,0)$.  At points of $\Sigma^0$ where $kf_{,x} \ne 0$ the field is not weakly isodrastic, because $|B|$ is not constant along $\Sigma^0$ (a little more analysis shows that the level sets of $|B|$ cross $\Sigma^0$ with cubic tangency). 
%For $j>0$, the function $H_j$ also fails isodrasticity.
A closed component of $\Sigma^0$ has at least one minimum and one maximum of $|B|$, so there are points on it where $f_{,x}=0$ or $k=0$.  A question is whether both types of point have to occur, but generically they do not happen simultaneously.  We analyse the local picture for the $|B|$ levels on $\Sigma$ near such points.

At a point of $\Sigma^0$ where $f_{,x}=0$ one has generically a non-degenerate local maximum or minimum of $|B|$ along $\Sigma^0$.  The level sets of $|B|$ on $\Sigma$ look locally like Figure~\ref{fig:fx0}.
\begin{figure}[htbp] %  figure placement: here, top, bottom, or page
   \centering
   \includegraphics[width=3in]{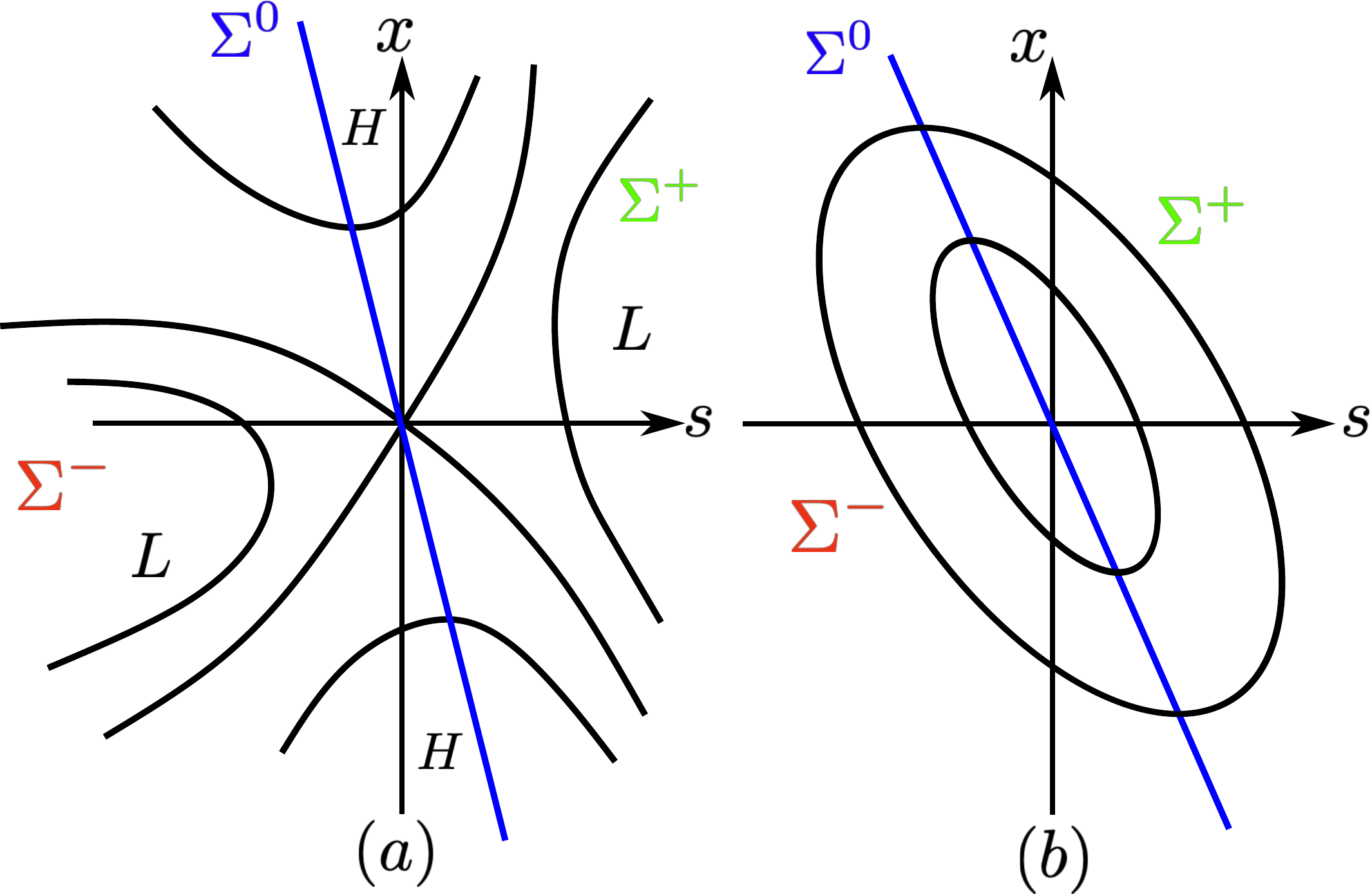} 
   \caption{Level sets of $|B|$ on $\Sigma$ for $f_{,x}=0$ in the case that $f_{,y}>0$ and (a) $f_{,xx}>0$, (b) $f_{,xx}<0$ where $H$ and $L$ denote high and low regions of $|B|$.}
   \label{fig:fx0}
\end{figure}
% \sn{Make sigma 0 intersection with constant H, horizontal.}
So the level curves of $|B|$ cross $\Sigma^0$ except in the hyperbolic sector of the first case.

For points of $\Sigma^0$ where $k=0$, note that on $\Sigma^0$, $k=0$ iff $B$ is tangent to $\Sigma^0$.  So one has generically one of the cases of Figure~\ref{fig:k0}.
\begin{figure}[htbp] %  figure placement: here, top, bottom, or page
   \centering
   \subfigure[]{\includegraphics[width=2.8in]{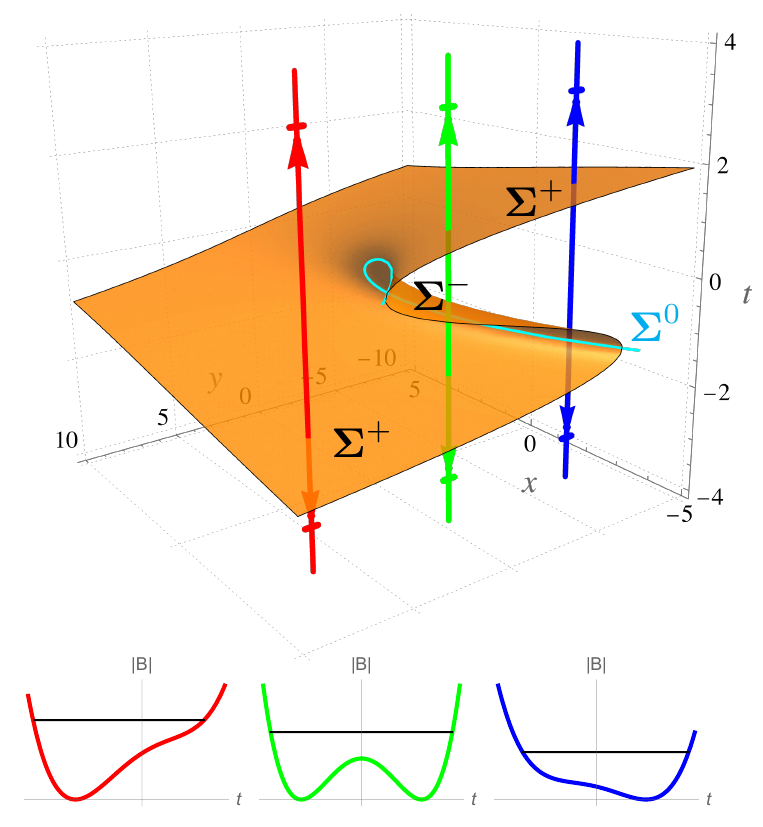}}
   \subfigure[]{\includegraphics[width=2.8in]{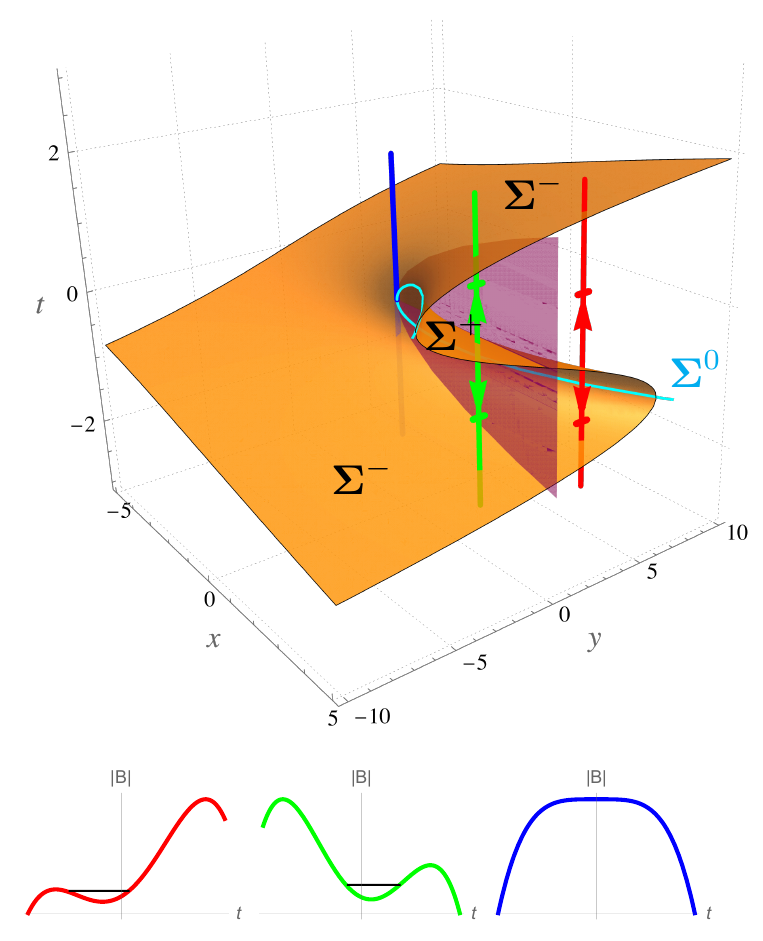}}
   \caption{The disposition of $\Sigma$ and $\Sigma^0$ relative to the magnetic field lines near a point $x=y=0$ with $k=0$; (a) quartic minimum, (b) quartic maximum. In each case, three fieldlines are shown (denoted by red, green, and blue lines) and a bouncing segment on each of them is denoted by bar arrows. The $|B|$ profiles along the fieldlines are shown in the bottom panel, together with the chosen level of $E/\mu$ for the bouncing segment.}
   \label{fig:k0}
\end{figure}
In the case of $|B|$ coming to a quartic minimum, all level curves of $|B|$ collide with $\Sigma^0$.
The case of $|B|$ coming to a quartic maximum produces a curve of switch of lowest hill.  This ``Maxwell curve'' is $y=0$, $x=2t^2$.  All $|B|$ levels on $\Sigma^+$ cross $\Sigma^0$.

If we want weak isodrasticity and there is a curve of $\Sigma^0$ then we must have $k f_{,x}=0$ everywhere on it.

To analyse the dynamics more thoroughly near $\Sigma^0$ one needs to compute $j$.  In particular, to test isodrasticity near but not on $\Sigma^0$ one needs $\text{\j}$ on $\Sigma^-$.  We can compute this to a good approximation for the short homoclinics that occur near the generic points of $\Sigma^0$ in the well of the cubic.
From a point of $\Sigma^-$ labelled by $(x,t)$, we use $t'(s) \approx 1$ to get the approximation
$$\text{\j} = \int_{s_0}^{s_1} \sqrt{2(C-ys-xs^2-ks^3)}\, ds,$$
where $C$ is the value at the local maximum $s_0$ of $ys+xs^2+ks^3$ and $s_1$ is the other point at the same height.  Write $s = s_0+v/k$ and $\alpha = \sqrt{-|B|''/2}$ at the point of $\Sigma^-$ and recall that $|B|'' = 2x + 6ks_0$.  Then
$$\text{\j} = \frac{\sqrt{2}}{k^2} \int_0^{\alpha^2} v \sqrt{\alpha^2-v}\, dv.$$
Put $u^2=\alpha^2-v$, to obtain
$$\text{\j} = \frac{2\sqrt{2}}{k^2} \int_0^\alpha (\alpha^2-u^2)u^2\, du = \frac{(-|B|'')^{5/2}}{15 k^2}.$$
Using $|B|'' = 2x + 6kt$ on $\Sigma$, we obtain
$$d\text{\j} = \frac{(-|B|'')^{3/2}}{3k^2}(dx + 3k\, dt).$$
Recalling the expression for $h=|B|$ on $\Sigma$,
$$dh = f_{,x} dx + f_{,y} (-2tdt-(2x+6kt)dt)-t^2dx-(2xt+6kt^2)dt.$$
Thus
$$dh \wedge d\text{\j} = \frac{(-|B|'')^{3/2}}{3k^2} (-2x f_{,y}-3kf_{,x}+y)\, dx \wedge dt,$$
using $y=-2xt-3kt^2$ (\ref{eq:D.y}).
It is zero iff
$y=2xf_{,y}+3kf_{,x}$.
This PDE can be solved for $f$, resulting in $h$ being an arbitrary function of $x^2-3ky$ (the deviation from $\Sigma^0$), but the main conclusion can be read off immediately by putting $(x,y)=(0,0)$:~$kf_{,x}=0$ at $(0,0)$. This implies for $k \ne 0$ that the derivative of $h$ along $\Sigma^0$ is zero, confirming one of our necessary conditions for weak isodrasticity.

It is useful to compare $dh\wedge d\text{\j}$ to $\beta$.  On $\Sigma$, locally $\beta = c(x,y)\, dx \wedge dy$ for some non-zero smooth function $c$, because $x$ and $y$ are fieldline labels.  Using $y=-2xt-3kt^2$ and the expression for $|B|''$, we obtain $dx \wedge dy =-c|B|'' dx \wedge dt$.  So the Melnikov function $$\cM = \frac{(-|B|'')^{1/2}}{3ck^2}(y-2xf_{,y}-3kf_{,x}).$$

Finally, for isodrasticity one also needs the corresponding condition for segments leaving $\Sigma^-$ in the opposite direction, which is not obtainable by local analysis.  A typical picture for the phase space $F_j$ around a fieldline with a cubic critical point is sketched in Figure~\ref{fig:F_j_cubic}, using the normal form (\ref{eq:Bgen}) and a generic assumption that the separatrix-area of the cubic critical point varies at non-zero rate along the fold curve $\Sigma^0$ (but note that this assumption is incompatible with weak isodrasticity, because we proved that $\j$ is constant along $\Sigma^0$ for weak isodrastic fields.)
\begin{figure}[htbp] %  figure placement: here, top, bottom, or page
   \centering \includegraphics[width=3in]{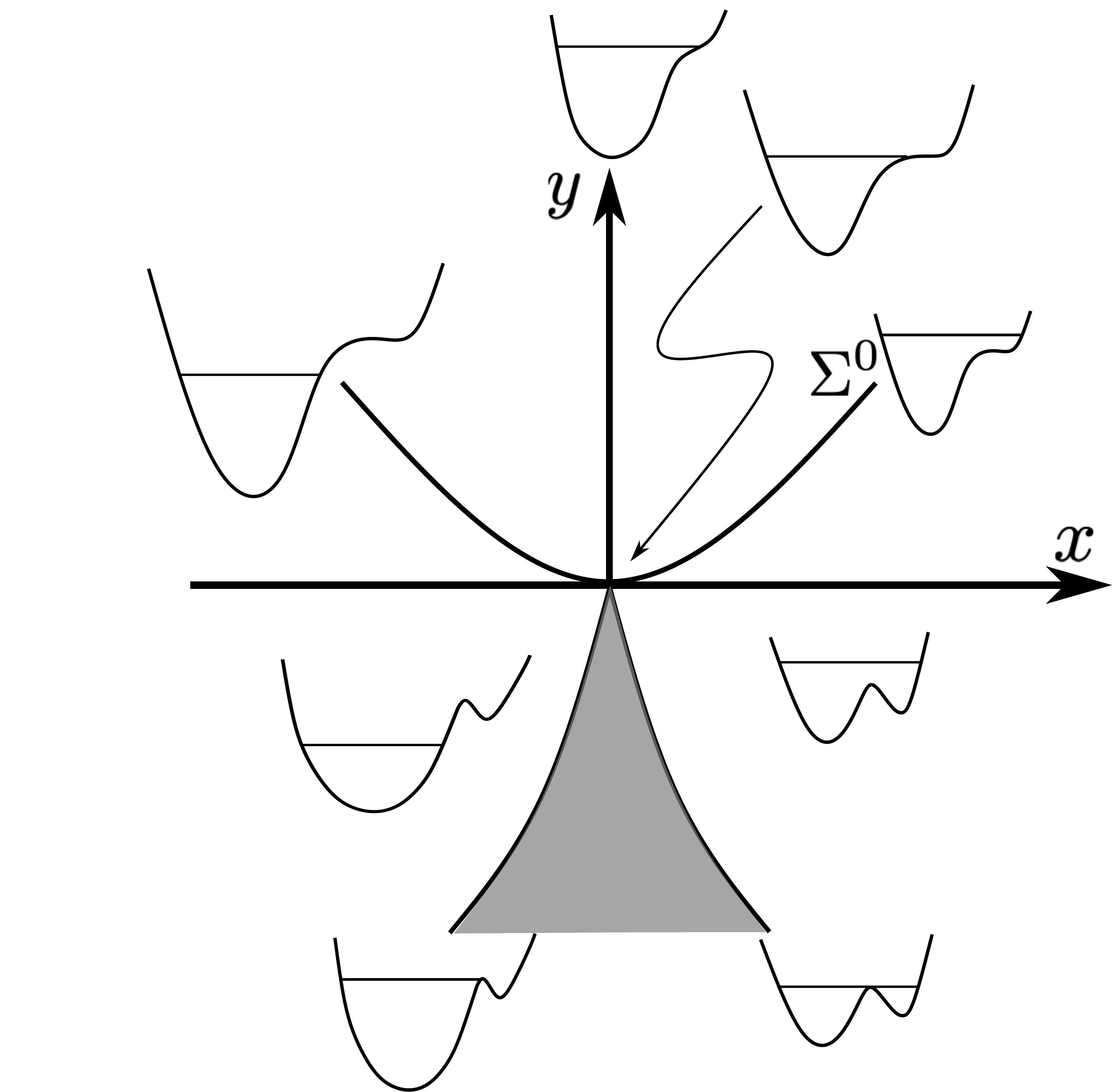}
   \caption{Sketch of the phase space $F_j$ near a fieldline with a cubic critical point with generic unfolding. $\Sigma^0$ is the curve with cubic critical points:~$3ky=x^2$ in fieldline labels $(x,y)$. %\rsm{try to make all the actions look the same}  
   The cusped wedge is excluded. 
   % \sn{updated, check for conservation of j.}
   %\rsm{perhaps we should shade it?}.  
   }
   \label{fig:F_j_cubic}
\end{figure}
The cubic critical point unfolds in the downward vertical direction, to a local maximum and a local minimum.  Moving horizontally to the right, the separatrix-area for the cubic critical point decreases, thus to maintain constant $j$, the energy-level changes as indicated.  The boundary of $F_j$ is a cusped curve with width $\Delta x$ asymptotically proportional to $(-y)^{5/4}$.  This follows from the above computation of $\j$ for the little well in the cubic.  The lefthand curve is determined by making the separatrix area for the main well be $j$; on the righthand curve the sum of the separatrix areas for the main well and the little well is $j$.

% As $j_0$ increases, the cusped curve moves to the left.  Isodrasticity requires the derivatives of the separatrix areas along the two branches of the cusp to be dependent with $dh$.  Under the assumptions behind the figure, their derivatives are independent below $\Sigma^0$.  So dependence with $dh$ is impossible unless $dh=0$, i.e.~$|B|$ is constant on $\Sigma^-$.  \rsm{Is this right?}
% This very strong condition has been named ``isoprominence'' and studied in \cite{BDM}. This conclusion would not hold if any of our genericity assumptions fail. For example, if the separatrix area were constant along $\Sigma^0$.

In conclusion, generic $\Sigma^0$ is incompatible with isodrasticity except with special design.
The typical ways isodrasticity fails near $\Sigma^0$ reveals a sensitivity of axisymmetric and quasisymmetric designs to imperfections.

\section{Derivative of reduced Hamiltonian\label{app:dh}}
Here we give an alternative proof of Proposition \ref{prop:dh} that illustrates its general connection to the theory of nearly-periodic systems. 
 {

\begin{proof}
In the bouncing region of the guiding center phase space the guiding center vector field $V_\epsilon$ defines a Hamiltonian nearly-periodic system with limiting roto-rate $R_0 = \frac{T}{2\pi}\,V_0$ and exact $\epsilon$-dependent symplectic form $\omega_\epsilon$. Here Hamiltonian means $\iota_{V_{\epsilon}}\omega_\epsilon = dH_\epsilon$, for some Hamiltonian $H_\epsilon$, and $T$ denotes the true period for bounce motion. As for all Hamiltonian nearly-periodic systems \cite{Bu3}, there exists an all-orders roto-rate $R_\epsilon$ such that $\iota_{R_\epsilon}\omega_\epsilon = dJ_\epsilon$, where $J_\epsilon$ denotes the all-orders bounce adiabatic invariant. By \cite{Bu3} the formula for the first-order term in the roto-rate is 
\begin{align*}
    R_1 = \mathcal{L}_{R_0}I_0\widetilde{V}_1,
\end{align*}
where $I_0$ denotes the inverse of $\mathcal{L}_{V_0}$ restricted to the subspace of vector fields with vanishing $U(1)$-average, and $\widetilde{V}_1 = V_1 - \langle V_1\rangle$ is the first-order guiding center vector field $V_1$ less its $U(1)$-average. It is easy to show that $H_0 = J_0 = 0$, and that $H_1,J_1$ correspond to the usual leading-order expressions for the guiding center energy and bounce invariant. Since both $R_\epsilon$ and $V_\epsilon$ are Hamiltonian vector fields we have
\begin{align*}
    \iota_{V_1}\omega_0 + \iota_{V_0}\omega_1& = dH_1\\
    \iota_{R_0}\omega_0 + \iota_{R_0}\omega_1 & = dJ_1.
\end{align*}
In particular $dH_1 = \iota_{V_1}\omega_0 + \frac{2\pi}{T}(dJ_1 - \iota_{R_1}\omega_0)$. This formula is useful because it simplifies computing the derivative of $H_1$ along $U(1)$-invariant vector fields $\mathcal{X} = \bm{u}\cdot \partial_{\bm{X}} + a\,\partial_{v_\parallel}$ that preserve $J_1$, i.e. $dJ_1(\mathcal{X}) = 0$. In particular,
\begin{align*}
    dH_1(\mathcal{X}) &= \omega_0(V_1,\mathcal{X}) - \frac{2\pi}{T}\omega_0(R_1,\mathcal{X})\\
    & = \omega_0(V_1,\mathcal{X}) - \frac{2\pi}{T}\omega_0(\mathcal{L}_{R_0}I_0\widetilde{V}_1,\mathcal{X})\\
    & = \omega_0(V_1,\mathcal{X}) - \frac{2\pi}{T}\mathcal{L}_{R_0}I_0\bigg[\omega_0(\widetilde{V}_1,\mathcal{X})\bigg]\\
    & = \omega_0(V_1,\mathcal{X}) - \omega_0(\widetilde{V}_1,\mathcal{X})\\
    & = \langle \omega_0(V_1,\mathcal{X})\rangle\\
    & = \frac{1}{T}\int_0^T \bigg(\mu\bm{u}\cdot \nabla_\perp|B| - mv_\parallel^2\,\bm{b}\cdot \bm{c}_\perp\times\bm{u}\bigg)\,dt.
\end{align*}
This establishes the Proposition after using $dt = b^\flat/v_\parallel$ and energy conservation.
\end{proof}
 }

\section{Double transitions}
\label{app:double}
As a special event, a fieldline segment can approach double transition.  Some ways this can occur are indicated in Figure~\ref{fig:double}. They involve formation of a heteroclinic cycle.  There are others involving formation of a degenerate critical point too, but they are more special.
\begin{figure}[htbp] %  figure placement: here, top, bottom, or page
   \centering
   \includegraphics[width=4in]{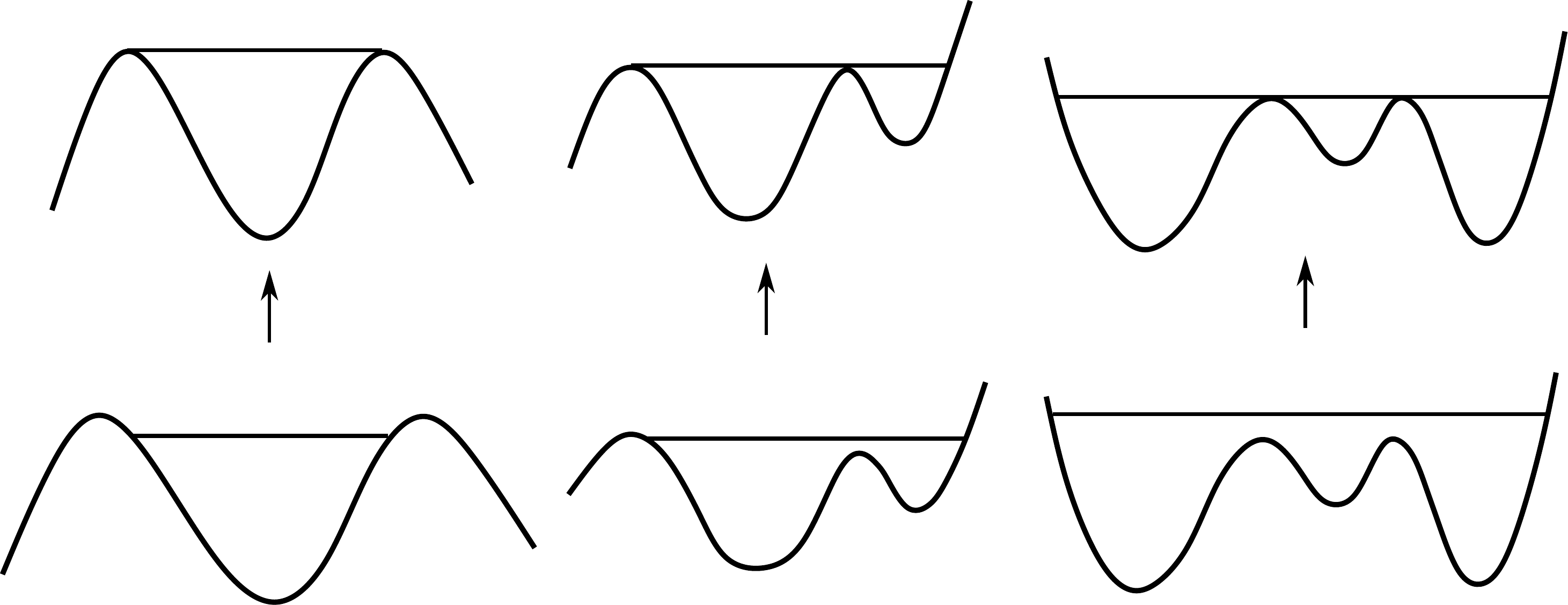}
   \caption{Some ways to double transition.}
   \label{fig:double}
\end{figure}
In the 2D space of fieldlines, a heteroclinic cycle requires only one condition, namely that $|B|$ have two local maxima at the same height, so it happens generically along curves in the space of fieldlines.

It leads to corners in the reduced spaces $F_j$, where boundaries corresponding to two different transitions meet.  See Figure~\ref{fig:corner}.
\begin{figure}[htbp] %  figure placement: here, top, bottom, or page
   \centering
   \includegraphics[width=3in]{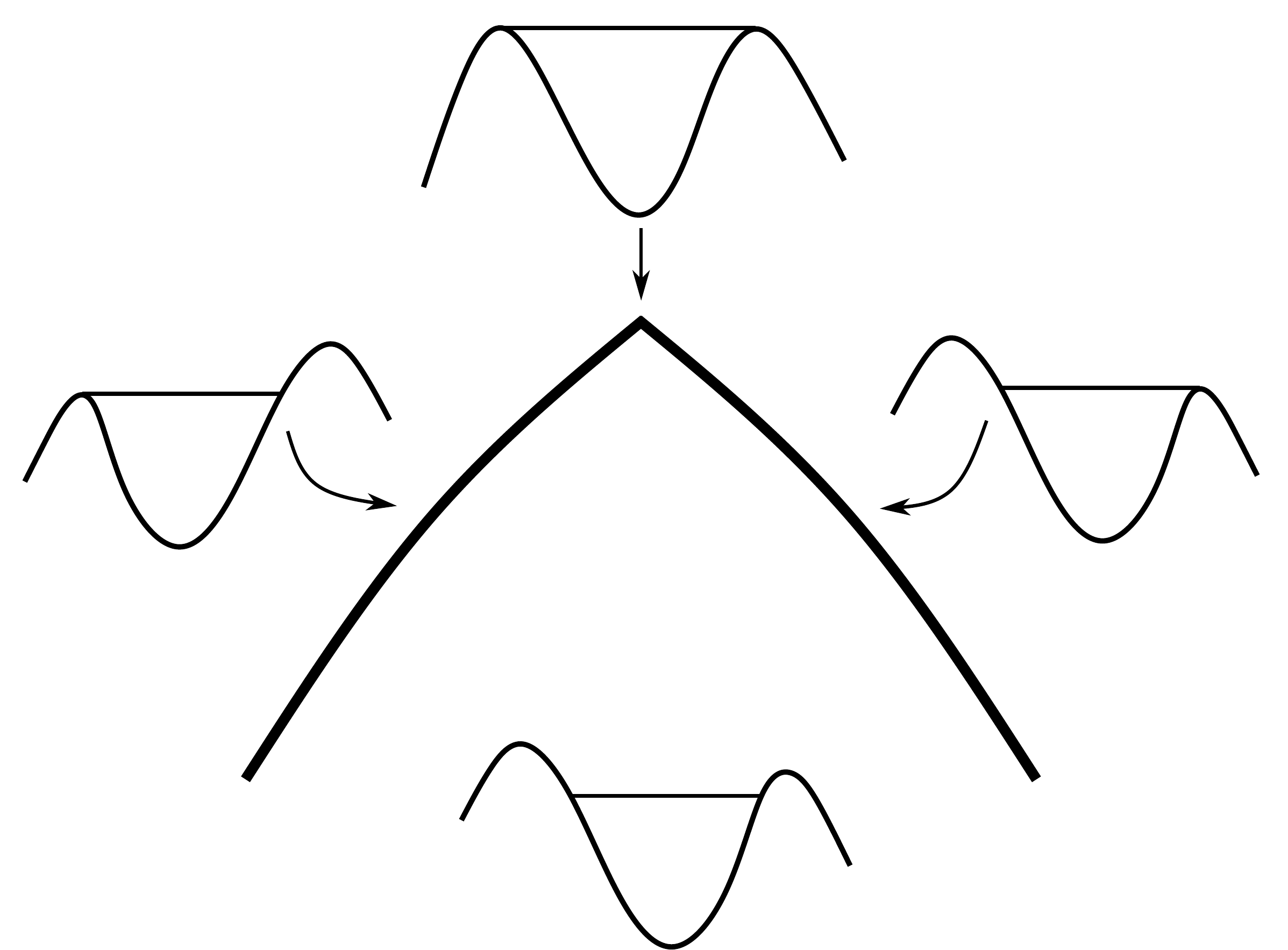} 
   \caption{An example of a corner in the reduced space $F_j$.}
   \label{fig:corner}
\end{figure}

One could make explicit examples, for example with $|B|$ a quartic in arclength whose coefficients depend on two fieldline labels, as for the formulae leading to Figure~\ref{fig:k0}(a).  Although evaluating $j$ requires elliptic integrals, the boundary cases can be computed explicitly (generalising the cubic case of Section~\ref{sec:construct} and Appendix~\ref{app:generic}). 

Near points of double transition, transitions between several classes can occur.  Isodrasticity requires that $\ker dH_j$ contain the tangent to the boundary at single transition points.  Taking the limit to a transverse corner, this implies that $dh=0$ there.  
%\rsm{Is this right? JB: Seems settled since we have good reason to say $dH_j = dh$ along boundary. --JB} 
The consequences of this are left to a future publication.

Note that the times spent by a periodic orbit near the saddles of a heteroclinic cycle of a Hamiltonian system are asymptotically proportional to the Lyapunov times of the saddles (i.e.~the inverses of their positive Lyapunov exponents).  So using (\ref{eq:dH_j}), along the curve for equal height, $dH_j$ is asymptotic to the convex combination of $dh$ at the two ends, weighted by their Lyapunov times.

Note also that $d\j$ goes to infinity at generic corners because the time spent near the second saddle grows logarithmically as it is approached along a level curve of $\j$, and the derivative of $\j$ is related to this time.

\section{Dipole field}
\label{app:dipole}
As a simple illustration of $\Sigma$ and the reduced Hamiltonian $H_j$, we consider GCM in a dipole field
$$B = \frac{3 \cos\theta\ \hat{r}-\hat{z}}{r^3}$$
in spherical polar coordinates (radius $r$, colatitude $\theta$, longitude $\phi$), illustrated in Figure~\ref{fig:dipoleNK}.  The dipole strength has been scaled to $4\pi$.
\begin{figure}[htbp] %  figure placement: here, top, bottom, or page
   \centering
   \includegraphics[width=4.5in]{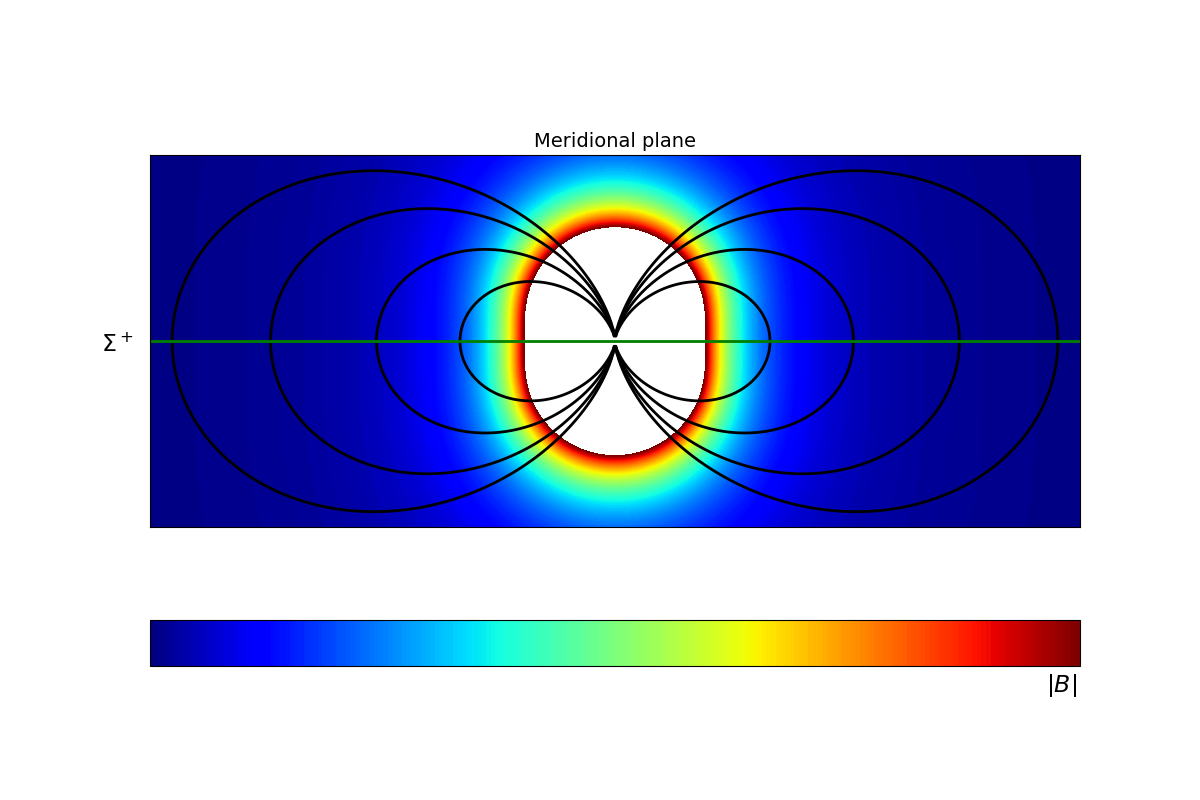}
   \caption{A meridional section through a dipole field, showing some fieldlines and the field strength in colour (blue to red indicates weak to strong and white indicates $|B|$ exceeds a threshold). $\Sigma^+$ is the equatorial plane.}
   \label{fig:dipoleNK}
\end{figure}
The fieldlines are $r = r_e \sin^2\theta$, $\phi=$ constant, where $r_e$ is the radius at which they cross the equatorial plane.  Along the fieldlines,
$$|B| = \frac{\sqrt{3 \cos^2 \theta + 1}}{r_e^3 \sin^6 \theta} .$$
The guiding centres bounce across the equatorial plane between the regions of stronger field near the poles.  Then $\Sigma$ is the equatorial plane and it consists of only $\Sigma^+$ (see Figure~\ref{fig:exSigma}(a)). 

For a bouncing segment, $j$ can be written as $j = {r_e^{-1/2}} F(h r_e^3)$, where $h$ is $|B|$ at the bounce points and 
$$ F(k) = \int_{\theta_0}^{\theta_1} \sqrt{2\left(k-\frac{\sqrt{3\cos^2 \theta + 1}}{\sin^6 \theta}\right)} \sqrt{2\cos^2 \theta + 1}\, \sin \theta\, d\theta,$$
with $\theta_i(k)$ being the zeroes of the first square root (symmetric about $\pi/2$).
$F$ is defined for argument greater than or equal to $1$, has $F(1)=0$, positive derivative (including at $k=1$ where $F'(1)=\tfrac{\pi}{3}$), and goes to infinity as $k \to \infty$.
FGCM is described in the adiabatic approximation with given value of $j$ (and scaled time $\tau$) by the Hamiltonian $H_j$ on $\Sigma^+$, defined by $H_j = r^{-3} F^{-1}(\sqrt{r} j)$ (we write $r_e=r$ on $\Sigma^+$),
and the symplectic form that is just the magnetic flux-form $r^{-2} dr \wedge d\phi$ on $\Sigma^+$. Now 
$$\frac{dH_j}{dr} = -\frac{3}{r^4}F^{-1}(\sqrt{r}j) + \frac{j}{2 r^{7/2} F'(k)} = r^{-4} \left(-3k + \frac{F}{2F'}\right),$$
where $k = H_j r^3$.
We didn't check, but presumably $F/F' < {6k}$ (this is certainly true for $k=1$ where $F/F'$ is zero, and at $k=\infty$ where $F/F' \sim {2k}$) so $H_j$ is a strictly decreasing function of $r$ for given $j$.  In any case, the level sets of $H_j$ are concentric circles, and
the bouncing segments precess around the dipole axis at constant rate
$\frac{d\phi}{d\tau} = r^2 \tfrac{dH_j}{dr}$.

If axisymmetry is broken but not too much in $C^1$, as perhaps for the earth's magnetosphere, and we ignore changes far away (such as due to the solar wind), then $\Sigma^+$ deforms into a nearby surface and $H_j$ into a nearby function.  Its level sets deform to closed curves near the original circles.  The segments continue to precess, but in general no longer at constant rate.  The precession period can be calculated from (\ref{eq:precession}).

As there is no $\Sigma^{-0}$, there are no transitions and the field is automatically isodrastic.  The only thing to check for confinement is which set of precessing segments to populate.  For example, in the context of the earth's magnetosphere, the desired region is the outside of the earth.  Then in the axisymmetric case, scaling the earth's radius to 1 and the field strength to 1 at the earth's equator to fit with the above, a calculation shows that the segments that do not hit the earth are those for which $h < \sqrt{4-3/r_e}$, where $h$ is the ratio of the energy to $\mu$.

The earth's magnetosphere is perturbed not only by deviations from a dipole of magnetic generation in the earth but also by interaction with the magnetic field of the solar system and solar wind, which change the arrangement of the fieldlines, notably introducing magnetic nulls.  We do not pursue those effects here.

\section{Computational practicalities}
\label{app:sqrts}
Firstly, we show how square roots can be avoided in computation of many of the quantities required to find $\Sigma$ and the Melnikov function (though not all).  Then we give suggestions for the computation of $\text{\j}$ on $\Sigma^-$.

\subsection{Eliminating square roots}
We have written expressions like $|B|' = b \cdot \nabla |B|$ but if $B$ is given in components, $|B|$ involves taking a square root, so $b=B/|B|$ involves dividing by a square root.  So does  differentiating $|B|$.  It can be better to write such quantities in terms of $|B|^2$.
For example, $$|B|' = 
%b\cdot \nabla |B| = 
%B \cdot \nabla |B|/|B| = 
\tfrac12 B \cdot \nabla |B|^2 /|B|^2.$$
Here, we collect various such formulae.

The grad-B drift involves $b\times \nabla |B|$, which can be written as $\frac12 B \times \nabla |B|^2 /|B|^2$.

The curvature drift involves $(\curl\, b)_\perp /|B|$.  Now $$\curl\, b = \curl\, \frac{B}{|B|} = \frac{J}{|B|} - \frac{\nabla |B|}{|B|^2} \times B = \frac{J}{|B|} + \tfrac12 B \times \frac{\nabla |B|^2}{|B|^3}.$$  So $$(\curl\ b)_\perp/|B| = J_\perp /|B|^2 + \tfrac12 B \times \nabla |B|^2/|B|^4.$$

We can also remove square roots from the computation of the covector $a$ in (\ref{eq:a}) by switching to fieldline-flow time $T$. We have $da_i/dT = -|B| a_j\partial_i b^j$ and $\partial_i b^j = \frac{1}{|B|} \partial_i B^j - \frac{B^j}{2|B|^3}\partial_i|B|^2$.  So $$\frac{da_i}{dT} = -a_j\left(\partial_i B^j - \frac{B^j}{2|B|^2}\partial_i |B|^2\right) = -a_j\partial_i B^j,$$
because $aB=0$.  

\subsection{Integrating along a fieldline}
Next, $\text{\j}$ is defined at a point $X_0 \in \Sigma^-$ by taking $h=|B(X_0)|$ and then letting $\text{\j} = \int \sqrt{2(h-|B|)}\, ds$ with respect to arclength $s$ along the fieldline through $X_0$ to the first bounce, i.e.~where $|B|=h$ again.

Suppose we are integrating in the positive direction along $B$ (the obvious changes apply for the other direction).
It is slightly more convenient to integrate with respect to fieldline flow time $T$ than arclength $s$, i.e.~start at $T=0$ with $j=0$ and $X = X_0$ and integrate 
$$ \frac{dX}{dT} = B(X), \quad \frac{dj}{dt} = v |B|,$$ with $v=\sqrt{2(h-|B|)}, $
until the first bounce, thereby banishing square roots to only the second of the two equations.

More importantly, it is better to switch to integration with respect to $w=-v$ when approaching the first bounce, i.e.~by eliminating time from ZGCM,
$$ \frac{dX}{dw} = -\frac{w b}{|B|'}, \quad \frac{dj}{dw} = \frac{w^2}{|B|'},$$
starting from $w = -\sqrt{2(h-|B|)}$ at the switch point and stopping at $w=0$.  Then put $\text{\j}$ equal to the final value of $j$.  The switch can be made at any point between the last local minimum of $|B|$ along the fieldline and the bounce point.

Nonetheless, in the codimension-1 case that the first place where $|B|=h$ again is also a zero of $|B|'$, this switch to integration with respect to $w$ is not appropriate.  Instead one can just integrate with respect to $T$ with termination condition $|B|'=0$ when near $B=h$ again.

These tips are also relevant for computing the Melnikov function.

\section{Perturbed toy tokamak}
\label{appsect:pert_tokamak}

Here we give the expressions used in the calculations for the weak and strong form of isodrasticity for the toy tokamak. 

The perturbation to break the axisymmetric field is given by 
\begin{align}
\nabla \times A_{\eps} = \left( - {\eps \cos \phi}, 0, \dfrac{\eps z \cos \phi}{R} \right), \nonumber
\end{align} 
where $\eps$ is the perturbation parameter (and not the GC parameter in Sec.~\ref{sec:intro}).

%\rsm{need to make a consistent choice of symbol; $\eps$ is used for the GC parameter in Sec.1, so perhaps use $\epsilon$ for deviation from axisymmetry?  But then would need to make changes in the main text too}
%\sn{Going with what we discussed, specifying that here we use $\eps$ to denote perturbation and not GC parameter. Saves us from making changes to the main text!}

The non-axisymmetric field magnitude,
\begin{align}
    |B| = \dfrac{\sqrt{C^2 + (\eps R \cos \phi + z)^2 + (R - R_0 + \eps z \cos \phi)^2}}{R}, \nonumber
\end{align}
and
\begin{align}
   \dfrac{d|B|}{ds} &= |B|^{\prime}  = \frac{f_1 + f_2 + f_3}{R^{2} \left(C^{2} + \left(R \eps \cos \phi + z\right)^{2} + \left(R - R_{0} + \eps z \cos \phi\right)^{2}\right)},\nonumber \\
\text{where} \; f_1 & = - C \eps \left(R \left(R \eps \cos \phi + z\right) + z \left(R - R_{0} + \eps z \cos \phi\right)\right) \sin \phi \nonumber \\
f_2 & = R \left(R - R_{0} + \eps z \cos \phi\right) \left(R \eps \cos \phi + \eps \left(R - R_{0} + \eps z \cos \phi\right) \cos \phi + z\right) \nonumber \\
f_3 & = \left(R \eps \cos \phi + z\right) \left(C^{2} - R_{0}(R - R_{0}) + 2 (R - R_0) \eps z \cos \phi + \eps^{2} z^{2} \cos^{2}{\left(\phi \right)} + z^{2}\right). \nonumber
\end{align}

% \begin{align}
%     |B|^\prime = b \cdot \nabla |B| = \dfrac{(\varepsilon \cos \phi + z)(C^2 - C \varepsilon \sin \phi + (\varepsilon \cos \phi + z)^2 + r^2)}{R^2(C^2 + (\varepsilon \cos \phi + z)^2 + r^2)}. \nonumber
% \end{align}

In $\tilde{B}_{\pl} = |B| + \sqrt{\tilde{\mu}} u (\curl \, b) \cdot b$, we rewrite
\begin{equation}
    (\curl \, b) \cdot b = \frac{- C R - C R_{0} + C \eps z \cos \phi - R (R - R_{0}) \eps \sin \phi + \eps z^{2} \sin \phi}{R \left(C^{2} + \eps^{2} \cos^{2}{\left(\phi \right)}(R^{2} + z^{2}) + (R - R_{0})^2 + 2 \eps z \cos \phi(2 R - R_{0}) + z^{2}\right)}. \nonumber
\end{equation}

\section{Slow manifold computation}
\label{app:slow}

The procedure from \cite{M04} to compute a first-order symplectic slow manifold $N^-$ for a Hamiltonian system from a zeroth-order one ($\Sigma^- \times \{u=0\}$ in our case) is to compute a symplectically orthogonal foliation to the latter and to find the unique nearby critical point of $H$ on each leaf.

The tangents to $\Sigma^- \times \{u=0\}$ are the vectors $(\delta x, \delta u) \in \R^3\times \R$ satisfying $\delta x \cdot \nabla |B|'=0, \delta u=0$.
The symplectic form is $\tilde{\omega} = \tilde{\mu}^{-1/2} \beta + u\, db^\flat + du \wedge b^\flat$.  This requires $\mu >0$ but the method works the same.
Thus a vector $(\xi,w)$ is symplectically orthogonal to  $\Sigma^-\times \{u=0\}$ iff
$$\tilde{\mu}^{-1/2} \beta(\xi,\delta x) + w b\cdot \delta x = 0$$
for all tangents $\delta x$ to $\Sigma^-$.  Recall that $\beta = i_B\Omega$.  For a generic point of $\Sigma^-$ where $b$ is not perpendicular to $\Sigma^-$, the choice $\delta x = b \times \nabla |B|'$ shows that $\xi = \alpha b + \gamma b\times \nabla |B|'$ for some $\alpha$ and $\gamma$.  Next, the choice $\delta x = \nabla |B|' \times (b \times \nabla |B|')$ yields
\begin{equation}
w = -\tilde{\mu}^{-1/2} \frac{\beta(\gamma b\times \nabla |B|', \nabla |B|' \times (b \times \nabla |B|')) }{ b\cdot (\nabla |B|' \times (b\times \nabla |B|'))} = -\tilde{\mu}^{-1/2}\gamma |B| |B|'',
\label{eq:w}
\end{equation}
after some cancellation.

We take the locally linear foliation given by the above symplectically orthogonal planes to $\Sigma^-\times\{u=0\}$.  
%As $b\times\nabla |B|'$ is tangent to $\Sigma^-$, displacement in that direction is equivalent to starting from a different point of $\Sigma^-$.  
So the leaf associated to a point $x \in \Sigma^-$ is parametrised by displacement $\alpha$ from $\Sigma^-$ along $b$, $\gamma b \times \nabla|B|'$ tangent to $\Sigma^-$ and scaled parallel velocity $ w = -\tilde{\mu}^{-1/2} \gamma |B||B|''$.

The Hamiltonian $H=\tfrac12 w^2 +|B|$ constrained to this leaf is $$H = \tfrac{\gamma^2}{2\tilde{\mu}} |B|^2 |B|''^2 + |B(x+\alpha b + \gamma b\times \nabla |B|')|.$$ Now $|B|'=0$ on $\Sigma^-$, so neglecting the effect of curvature of $\Sigma^-$, critical points with respect to displacement $\alpha$ along $b$ have $\alpha=0$ (it would be good to  estimate the error).  Criticality with respect to $\gamma$ is given to leading order by $$\frac{\gamma}{\tilde{\mu}}|B|^2|B|''^2 +(b\times \nabla |B|') \cdot \nabla |B|=0.$$  There's also a term proportional to $\gamma$ from the second derivative of $|B|$ but the $w$ term dominates it by the factor $1/{\tilde{\mu}}$. So we end up with
$$\gamma = -\tilde{\mu} \frac{b\times \nabla |B|' \cdot \nabla |B|}{|B|^2|B|''^2}.$$
In particular, using (\ref{eq:w}),
$$w = \sqrt{\tilde{\mu}} \frac{b\times \nabla |B|' \cdot \nabla |B|}{|B||B|''}.$$

\section{Systems with perfect separatrices}
\label{app:sep}
There exist Hamiltonian systems that have one or more perfect separatrices but are not integrable.

Firstly, one can make area-preserving twist maps with this property.  The idea is to use the construction by de la Llave in the appendix to \cite{Ma}.  Given a (lift of a) degree-one homeomorphism $g$ of the circle $\R/\Z$, define $$h(x) = g(x)+g^{-1}(x)-2x,$$ and define map $T$ for $(x,y)$ on the cylinder $\R/\Z \times \R$  by
$$ y' = y + h(x),\quad x'=x+y'$$
It is an area-preserving twist map and the circles $y=x-g(x)$ and $y=x-g^{-1}(x)$ are invariant, with dynamics $x'=g^{-1}(x)$ and $x'=g(x)$ respectively.  So we choose $g$ to be a circle homeomorphism with two fixed points, e.g.~$g(x) = x + k\sin 2\pi x$ with $0<k<\tfrac{1}{2\pi}$.  Then we obtain a map $T$ with two period-one islands with perfect separatrices.  But in general it is non-integrable, as illustrated in Figure~\ref{fig:apmap} (we presume that one could prove this if desired).
\begin{figure}[htbp] %  figure placement: here, top, bottom, or page
    \centering
    \includegraphics[width=3in]{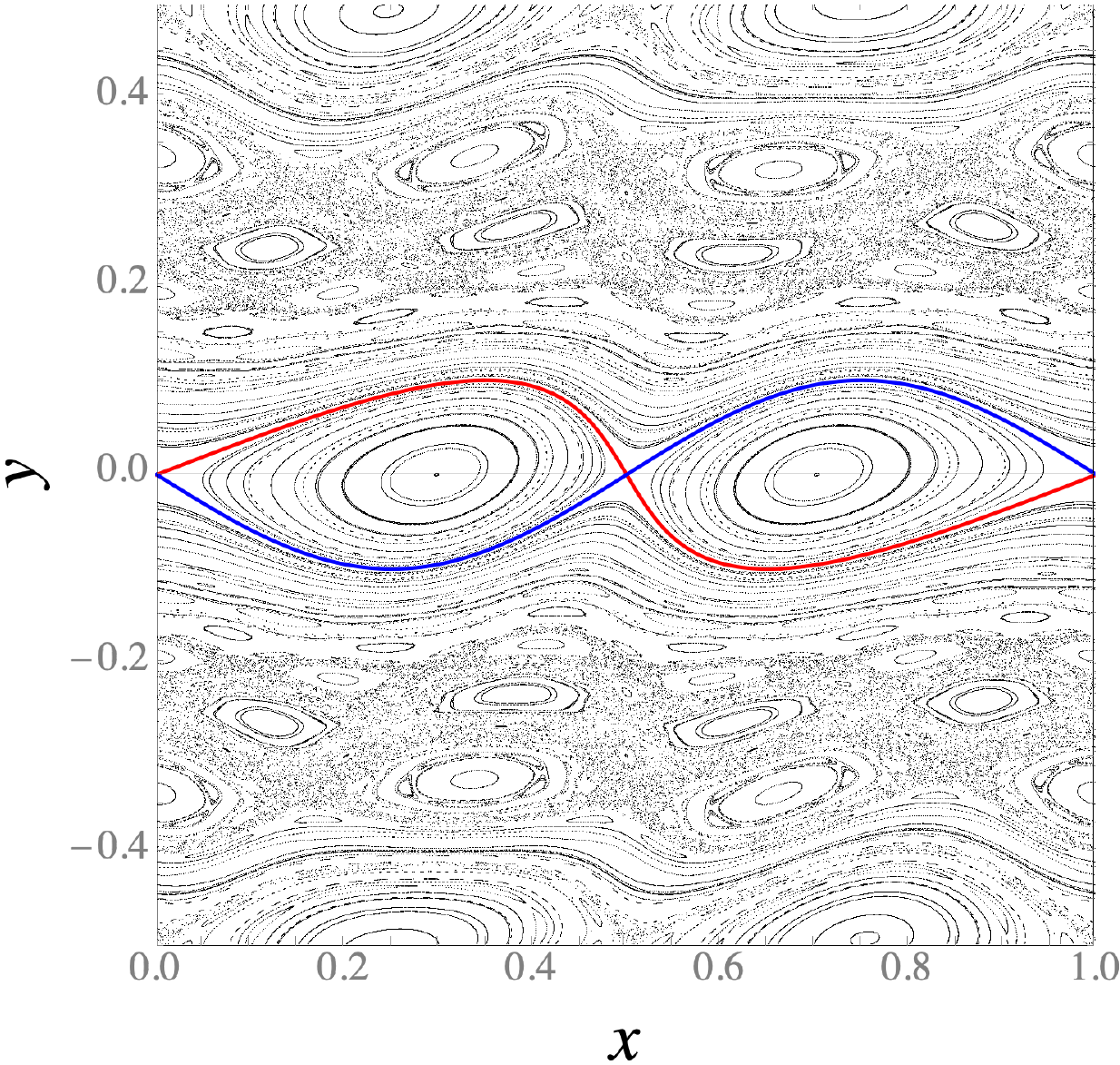}
    \caption{Some orbits of the map $T$ for $k=0.1$ with the $y = x - g^{-1}(x)$ invariant circle shown as the red and the $y = x - g(x)$ as the blue curve.}
    \label{fig:apmap}
\end{figure}
One could modify the choice of $g$ to make an example with a single period-one island but it would require careful matching of derivatives of $g$ to the left and right of the fixed point if one wants $T$ to be smooth.

Similarly, one can make continuous-time Hamiltonian systems with a perfect separatrix without imposing integrability.  For example, given a function $S:\R^2\to \R$ with constants $A,B$ such that $S(x+m,y+n) = S(x,y) + mA+nB$, and the vector field $\dot{q}=\nabla S(q)$ induced on $\R^2/\Z^2$, then $H=\tfrac12|p|^2-\tfrac12|\nabla S(q)|^2$ has invariant graph $p=\nabla S(q)$ on which $\dot{q}=\nabla S(q)$.  So choose $S$ with a periodic orbit  repelling on one side, attracting on the other.  What is not clear to us is how to make examples with perfect separatrices for all energies simultaneously.

%\section*{Checking perturbation is divergence-free}

%[RSM: I suggest to delete this.  It is trivial to people who deal with magnetic fields.]

%To show that the form of the perturbation to the axisymmetric field is divergence free, we compute
%\begin{align}
%\nabla \cdot \nabla \times A = %\dfrac{1}{R}\dfrac{\partial %}{\partial R} R (\nabla \times %A)_R + \dfrac{1}{R} %\dfrac{\partial }{\partial %\phi}(\nabla \times A)_{\phi} + %\dfrac{\partial }{\partial %z}(\nabla \times A)_z = 0
%\end{align}
%where the subscript of the parentheses denotes the component in the cylindrical coordinates. 

\section{Persistence of $\Sigma^{0}$ and $\Sigma^{+}$}
\label{app:pers+}

In this appendix, we address the question of whether also $\Sigma^0$ and $\Sigma^+$ have continuations to invariant submanifolds for small $\mu$.

We begin with $\Sigma^0$.  
If a 2DoF Hamiltonian system has a 2D manifold $N^0$ consisting of elementary saddle-centre periodic orbits (``elementary'' means that some generic conditions are satisfied) then firstly it is part of a 3D manifold $N$ consisting of periodic orbits, which decomposes into $N^-, N^0$ and $N^+$, with the orbits on $N^-$ being hyperbolic and those on $N^+$ being elliptic.
Secondly, all nearby Hamiltonian systems have a nearby such $N$ locally.
These results follow from  \cite{Me}, which is formulated in the context of area-preserving maps.

This is a useful result once we have $\mu>0$, but unfortunately it does not apply to $\mu=0$ because firstly the points of $\Sigma^0\times\{0\}$ are all equilibria, not periodic orbits.  Secondly, the Poisson bracket is degenerate at $\mu=0$ so we are not starting from a genuine Hamiltonian system.

We suspect that under the conditions that $\Sigma^0$ be a generic curve and $|B|$ be constant along its components then at least $\Sigma^{-0}\times \{0\}$ 
persists to an invariant $N^{-0}$.  The $|B|$ constant condition is necessary for the reduced dynamics for $\mu>0$ to have $\Sigma^0$ invariant.

One might ask why we do not conjecture that the whole of $\Sigma\times\{0\}$ persists to an invariant $N$.  The answer is that generically we expect resonances to break $\Sigma^+\times \{0\}$.  We give some explanation.

$\Sigma^+ \times \{0\}$ is a slow manifold but normally elliptic instead of hyperbolic.  In general, the best one can deduce for a normally elliptic slow manifold is that the true system has a sequence of submanifolds $N^+_m$ that are invariant to $n^{th}$ order, but in general the sequence does not converge \cite{M04}.  In this low-dimensional case, however, the reduced dynamics on $\Sigma^+$ consists of periodic orbits that are non-degenerate except at integer resonances (in this context, integer resonances means the linearised bounce frequency is an integer multiple of the precession frequency).  If we cut out a neighbourhood of these integer resonances, there is a true invariant submanifold $N^+$ nearby consisting of periodic orbits of the guiding-centre dynamics, which are elliptic except near half-integer resonances (where they generically turn inversion hyperbolic).
This procedure to construct $N^+$ fails at integer resonances (including at $\Sigma^0$ where the elliptic frequency goes to zero).

%The true dynamics on $N^\pm$ is given by restricting the guiding-centre Hamiltonian and symplectic form to them.  Being only two-dimensional and symplectic, most of the dynamics on them is periodic (the exceptions being equilibria and orbits connecting them).

%If the Poisson bracket were non-degenerate and $\Sigma^0 \times \{ 0\}$ consisted of elementary saddle-centre periodic orbits (defined in Appendix~\ref{app:contin}) of $\tilde{H}$ for $\mu=0$, the continuation of the whole of $\Sigma \times \{0\}$ to an invariant submanifold $N$ would be automatic, including a decomposition into $N^-,N^0,N^+$, and $\tilde{H}$ being constant on $N^0$ (but not necessarily with the coincidence of the contracting submanifolds of $N^-$).  This follows from results of Meyer formulated in the case of area-preserving maps (see Appendix~\ref{app:contin}).

%Unfortunately, the result can not be applied to $\Sigma^0\times\{0\}$ and $\mu=0$ because the Poisson bracket is degenerate for $\mu=0$ and $\Sigma\times \{0\}$ consists of equilibria, not periodic orbits.  If we scale time to keep periodic orbits in the limit $\mu \to 0$ then the normal dynamics becomes infinitely fast.
%But probably the result could be extended to this context if one adds that $|B|$ is constant along $\Sigma^0$.
%In any case, if for some $\mu>0$ the above hypotheses are satisfied then the scenario holds for all nearby $\mu$.

Integer resonances all correspond to saddle-centre periodic orbits but they are unlikely to be elementary:~that would require in particular that the field strength happen to have a turning point at them.  Thus we expect $\Sigma^+$ to break at integer resonances, if there is something for them to resonate with. For axisymmetric fields the resonances are not excited by turning on $\mu>0$, but for general non-axisymmetric ones they are excited. 
The (angular) bounce frequency is $\omega_b = \sqrt{\tfrac{\mu}{m}|B|''}$ and the (angular) precession frequency is $\omega_p = 2\pi \tfrac{\mu}{e}\frac{\partial h}{\partial\Phi}$, so the condition for integer resonance $n$ is 
$$\frac{\partial \Phi}{\partial h} \sqrt{|B|''} = 2\pi n \sqrt{\tilde{\mu}},$$
where $\tilde{\mu} = \tfrac{m}{e^2}\mu$.
Unfortunately, for typical fields with the derivative of the lefthand side non-zero this means that there is a large set of resonances for $\tilde{\mu}$ small.  Nonetheless, if the lefthand side is non-zero they have large $n$ and for smooth enough fields the resonances can be expected to be very weak.  See \cite{M04} for an example where they are exponentially weak for large $n$.  Also, adiabatic invariance of $L$ typically makes the motion near $\Sigma^+$ bounded.

We close this appendix by illustrating the problem of integer resonances for $\Sigma^+$ by the tokamak example of~\ref{sec:tok}.  Then $\Phi = 2\pi\psi$ plus a constant, so 
$$\frac{1}{2\pi}\frac{\partial \Phi}{\partial h} \sqrt{|B|''} = -\frac{r^{3/2}R(C^2+r^2)^{1/4}}{C^2-r}$$ on $\Sigma^+$.  So we get resonance $-n$ at $r\approx C\tilde{\mu}^{1/3} n^{2/3}$.  As $\mu \to 0$ each $-n$ resonance tends along $\Sigma^+$ to the magnetic axis, so Figure~\ref{fig:N4tok} would generically show a sequence of breaks on breaking axisymmetry.

%\section{Periodic orbits near an elementary saddle-centre orbit}
%\label{app:contin}
%\rsm{Probably this should not be included}
%An {\em elementary saddle-centre periodic orbit} of a 2DoF Hamiltonian system (see \cite{Me} where it is called an extremal periodic point) is a periodic orbit for which there is a double Floquet multiplier $+1$ and the Poincar\'e generating function $G$ for the return map to a transverse section in suitable coordinates $(q,p)$ with the change $e$ in energy as a parameter, satisfies at $(0,0,0)$:
%$$G_{1}=G_{2}=G_{22}=G_{12}=0,$$
%$$G_{11} \ne 0, G_{23}\ne 0, G_{222}\ne 0,$$
%where subscript $n$ denotes partial derivative with respect to the $n^{th}$ argument.
%Recall that a Poincar\'e generating function $G$ generates an area-preserving map via
%$$p'-p = G_1(q'+q,p'+p), \quad q-q' = G_2(q'+q,p'+p),$$
%so fixed points of the map correspond to critical points of $G$ and Floquet multipliers $+1$ to degeneracy of the matrix $G_{ij}$ with $i,j \in \{1,2\}$.

%It follows \cite{Me} that an elementary saddle-centre periodic orbit lies in a smooth curve of periodic orbits, being hyperbolic on one side of the saddle-centre case and elliptic on the other, with the energy reaching a non-degenerate turning point at the saddle-centre case.  This situation is preserved by small smooth perturbation.

%Note that an alternative generating function (and associated definition of elementary) was used in \cite{Me2}.


\begin{thebibliography}{MMMM}
\bibitem[Ar]{Ar} Arnol'd VI, Instability of dynamical systems with many degrees of freedom, Sov Math Dokl 6 (1964) 581--5.
\bibitem[BGKM]{BGKM} Baesens C, Guckenheimer J, Kim S, MacKay RS, Three coupled oscillators:~Mode-locking, global bifurcations and toroidal chaos, Physica D 49 (1991) 387–475.
\bibitem[Ba]{Ba} Balescu R, Transport Processes in Plasmas, Neoclassical Transport Vol.II (North-Holland, 1988).
\bibitem[B+]{B+} Beidler CD, Kolesnichenko YaI, Marchenko VS, Sidorenko IN, Wobig H, Stochastic diffusion of energetic ions in optimized stellarators, Phys Plasma 8 (2001) 2731--8.
\bibitem[Bo]{Boozer} Boozer AH, Why carbon dioxide makes stellarators so important, Nucl Fusion 60 (2020) 065001.
%\bibitem[Bo2]{Bo2} Boozer AH, Constraints on stellarator divertors from Hamiltonian mechanics, arxiv:2206.06368
\bibitem[Bo84]{Bo84} Boozer AH, Time-dependent drift Hamiltonian, Phys Fluids 27 (1984) 2441--5.
%\bibitem[BDM]{BDM} Burby JW, Duignan N, Meiss JD, Minimizing separatrix crossings through isoprominence, arXiv:2210.10218 (2022).
\bibitem[BSQ]{Bu} Burby JW, Squire J, Qin H, Automation of the guiding centre expansion, Phys Plasma 20 (2013) 072105.
\bibitem[BE]{BuEl} Burby JW, Ellison CL, Toroidal regularization of the guiding center Lagrangian, Phys Plasma 24 (2017) 110703.
\bibitem[BS]{Bu3} Burby JW, Squire J, General formulas for adiabatic invariants in nearly-periodic Hamiltonian systems, J Plasma Phys 86 (2020) 835860601.
\bibitem[BH]{Bu2} Burby JW, Hirvijoki E, Normal stability of slow manifolds in nearly-periodic Hamiltonian systems, J Math Phys 62 (2021) 093506.
\bibitem[BKM]{BKM} Burby JW, Kallinikos N, MacKay RS, Some mathematics for quasi-symmetry, J Math Phys 61 (2020) 093503.
\bibitem[BKM2]{BKM2} Burby JW, Kallinikos N, MacKay RS, Approximate symmetries of guiding centre motion, J Phys A 54 (2021) 125202
\bibitem[CS]{CS} Cary JR, Shasharina SG, Omnigenity and quasihelicity in helical plasma confinement systems, Phys Plasmas 4 (1997) 3323--33.
\bibitem[CF]{Ce} Cerfon AJ, Freidberg JP, ``One size fits all” analytic solutions to the Grad-Shafranov equation, Phys Plasmas 17 (2010) 032502.
\bibitem[Co]{Co} Conley CC, Low energy transit orbits in the restricted three-body problem, SIAM J Appl Math 16 (1968) 732--746.
\bibitem[Du]{Du} Dumas SH, The KAM story (World Sci, 2014).
\bibitem[F+]{F+} Faustin JM, Cooper WA, Graves JP, Pfefferl\'e D, Geiger J, Fast particle loss channels in Wendelstein 7-X, Nucl Fusion 56 (2016) 092006.
\bibitem[Fe]{F} Fenichel N, Persistence and smoothness of invariant manifolds for flows, Indiana U Math J 21 (1971) 193--225.
\bibitem[GB]{GB} Garren DA, Boozer AH, Magnetic field strength of toroidal plasma equilibria, Phys Fluids B 3 (1991) 2805--21.
\bibitem[G+]{G+} Garren AA, Riddell RJ, Smith L, Bing G, Henrich LR, Northrop TG, Roberts JE, Individual particle motion and the effect of scattering in an axially symmetric magnetic field, Univ of California Radiation Laboratory report UCRL-8076 (1958); permalink https://escholarship.org/uc/item/9rb7q30n
\bibitem[GRR]{GRR} Golab AJ, Robinson JC, Rodrigo JL, On 2D harmonic extensions of vector fields and stellarator coils. arXiv:2108.07643 (2021).
\bibitem[GT]{GT} Goldston RJ, Towner HH, Effects of toroidal field ripple on suprathermal ions in tokamak plasmas, J Plasma Phys 26 (1981) 283--307.
%\bibitem[GMS]{GMS} Greene JM, MacKay RS, Stark J, Boundary circles for area-preserving maps, Physica D 21 (1986) 267--295
\bibitem[HM]{HM} Hall LS, McNamara B, Three-dimensional equilibrium of the anisotropic, finite-pressure guiding-center plasma:~Theory of the magnetic plasma, Phys Fluid 18 (1975) 552--65.
\bibitem[He]{H} Helander P, Theory of plasma confinement in non-axisymmetric magnetic fields, Rep Prog Phys 77 (2014) 087001.
\bibitem[HM82a]{HM82a} Holmes, P. \& Marsden, J. Horseshoes in perturbations of Hamiltonian systems with two degrees of freedom. Commun. Math. Phys., 82, (1982), 523--544. 
\bibitem[HM82b]{HM82b} Holmes, P. \& Marsden, J. Melnikov’s method and Arnold diffusion for perturbations of integrable Hamiltonian systems. J Math Phys 23,4, (1982), 669--675. 
\bibitem[HPS]{HPS} Hirsch MW, Pugh CC, Shub M, Invariant manifolds, Lect Notes Math 583 (Springer, 1977).
\bibitem[KIVM]{KIVM} Kallinikos N, Isliker H, Vlahos L, Meletlidou E, Integrable perturbed magnetic fields in toroidal geometry:~An exact analytical flux surface label for large aspect ratio, Phys Plasma 21 (2014) 064504.
\bibitem[KH]{KH} Katok A, Hasselblatt B, Introduction to the modern theory of dynamical systems (Cambridge U Press, 1995).
\bibitem[Ke]{Ke} Keller JB, Inverse problems, Am Math Monthly 83 (1976) 107--18.
\bibitem[KW]{KW} Krajnak V, Waalkens H, The phase space geometry underlying roaming reaction dynamics, J Math Chem 56 (2018) 2341--78.
\bibitem[Kr]{Kr} Kruskal MD, Asymptotic theory of Hamiltonian and other systems with all solutions nearly periodic, J Math Phys 3 (1962) 806.
\bibitem[Ku]{K} Kuehn C, Multiple time scale dynamics (Springer, 2015).
\bibitem[LC]{LC} Landreman M, Catto PJ, Omnigenity as generalized quasisymmetry, Phys Plasma 19 (2012) 056103.
\bibitem[LP]{LP} Landreman M, Paul E, Magnetic fields with precise quasisymmetry for plasma confinement, Phys Rev Lett 128 (2022) 035001.
\bibitem[La]{La} Lang S, Introduction to Differentiable Manifolds, 3rd edn (New York:~Interscience, 1967)
\bibitem[Li]{L} Littlejohn RG, Variational principles of guiding centre motion, J Plasma Phys 29 (1983) 111--25.
\bibitem[M90]{M90} MacKay RS, Flux over a saddle, Phys Lett A 145 (1990) 425--7.
\bibitem[M94]{M94} MacKay RS, On the motion of guiding centres, in:~Transport, chaos and plasma physics, eds Benkadda S, Doveil F, Elskens Y (World Sci, 1994) 96--101.
\bibitem[M04]{M04} MacKay RS, Slow manifolds, in:~T Dauxois, A Litvak-Hinenzon, RS MacKay, A Spanoudaki (eds), Energy localisation and transfer (World Sci, 2004) 149--92.
\bibitem[M20]{M20} MacKay RS, Differential forms for plasma physics, J Plasma Phys 86 (2020) 925860101
\bibitem[MM]{MM} MacKay RS, Meiss JD, Flux and differences of action for continuous-time Hamiltonian systems, J Phys A 19 (1986) L255--9.
%\bibitem[MS]{MS} MacKay RS, Stark J, Locally most robust circles and boundary circles for area-preserving maps, Nonlinearity 5 (1992) 867--888
\bibitem[MR]{MR} Marsden JE, Ratiu TS, Introduction to mechanics and symmetry, 2nd ed (Springer, 1999).
\bibitem[Mat]{Ma} Mather JN, Non-existence of invariant circles, Ergod Th Dyn Sys 4 (1984) 301--9.
\bibitem[Men]{Men} Menyuk CR, Particle motion in the field of a modulated wave, Phys Rev A 31 (1985) 3282--90.
\bibitem[Mey]{Me} Meyer KR, Generic bifurcation of periodic points, Trans Am Math Soc 149 (1970) 95--107.
%\bibitem[Me2]{Me2} Meyer KR, Generic bifurcations in Hamiltonian systems, in: Dynamical Systems Warwick 1974, ed Manning A, Lect Notes Math 468 (Springer, 1975) 62--70.
\bibitem[M+]{M+} Mikhailov MI, Cooper WA, Isaev MYu, Shafranov VD, Skovoroda AA, Subbotin AA, Improved stellarator systems, in ISPP-18 ``Piero Caldirola'', Theory of Fusion plasmas, eds Connor JW, Sindoni E, Vaclavik J (SIF, Bologna, 1999) 185--198.
\bibitem[Mi]{Mi} Milnor J, Morse theory (Princeton U Press, 1969).
\bibitem[Mo]{Mo} Moser JK, Stable and random motions in dynamical systems (Princeton U Press, 1973)
\bibitem[MCB]{MCB} Mynick H, Chu, Boozer A, Class of Model Stellarator Fields with Enhanced Confinement, Phys Rev Lett 48 (1982) 322--6.
\bibitem[Ne]{N} Neishtadt AI, On the change in the adiabatic invariant on crossing a separatrix in systems with two degrees of freedom, PMM USSR 51:5 (1987) 586--92.
\bibitem[PCHL]{PCHL} Parra FI, Calvo I, Helander P, Landreman M, Less constrained omnigeneous stellarators, Nuclear Fusion 55 (2015) 033005.
\bibitem[P+]{P+} Paul EJ, Bhattacharjee A, Landreman M, Velasco JL, Energetic particle loss mechanisms in reactor-scale equilibria close to quasi-symmetry, Nuclear Fusion 62 (12), (2022), 126054.
\bibitem[Poi]{Poi} Poincar\'e H,  Sur le probl\`eme des trois corps et les \'equations de la dynamique, Acta Math 13 (1890) 1--270.
\bibitem[Pos]{P} Post RF, The magnetic mirror approach to fusion, Nucl Fusion 27 (1987) 1579--739.
\bibitem[RR]{RR} Renardy M, Rogers RC, An introduction to partial differential equations, 2nd ed (Springer, 2004).
\bibitem[Ro]{R} Robinson RC, %Melnikov's method for autonomous Hamiltonians, arxiv \\
Horseshoes for autonomous Hamiltonian systems using the Melnikov integral, Ergod Th Dyn Sys 8* (1988) 395--409.
\bibitem[RBNV]{R18}Ross, S., BozorgMagham, A., Naik, S. \& Virgin, L. Experimental validation of phase space conduits of transition between potential wells. Phys. Rev. E. 98 (5), (2018), 052214. 
% \sn{Validation of flux over saddle formula in 2dof system and numerical method for computing global tube submanifolds.}.
\bibitem[SZ]{SZ} Salamon D, Zehnder E, KAM theory in configuration space, Comment Math Helv 64 (1989) 84--132.
\bibitem[Sk]{Sk} Skovoroda AA, Pseudosymmetry near a magnetic surface in a plasma confinement system, Plasma Phys Reports 26 (2000) 550--9.
\bibitem[SS]{SS} Skovoroda AA, Shafranov VD, Isometric magnetic confinement systems, Plasma Phys Rep 21 (1995) 886--906. %937--
\bibitem[So]{S} Solov'ev LS, The theory of hydromagnetic stability of toroidal plasma configurations, Sov Phys JETP 26 (1968) 400--7.
%\bibitem[S+]{S+} Subbotin AA et al, Nucl Fusion 46 (2006) 921
\bibitem[We]{W} Weinstein A, Connections of Berry and Hannay type for moving Lagrangian submanifolds, Adv Math 82 (1990) 133-59.
\end{thebibliography}
\end{document}